\documentclass{amsart}
\usepackage[foot]{amsaddr}
\usepackage[pagebackref,colorlinks,citecolor=blue,linkcolor=magenta]{hyperref}
\usepackage[utf8]{inputenc}
\usepackage{microtype}
\usepackage[numbers]{natbib}
\usepackage{verbatim}
\usepackage{subcaption}
\usepackage[dvipsnames]{xcolor}

\usepackage{booktabs}
\newcommand\headercell[1]{
	\smash[b]{
		\begin{tabular}
			[t]{@{}c@{}} #1
		\end{tabular}
	}}

\usepackage{enumitem}
\usepackage[ruled,linesnumbered,noend]{algorithm2e}
\SetKwInOut{Input}{input}
\SetKwInOut{Output}{output}
\SetKwInOut{Parameter}{parameter}
\makeatletter
\let\c@algocf\c@subsubsection
\makeatother

\let\oldnl\nl
\newcommand{\nonl}{\renewcommand{\nl}{\let\nl\oldnl}}

\usepackage{mathtools}
\usepackage{amsthm,amssymb}
\usepackage{bbm}

\usepackage{tikz}
\usepackage{csvsimple}

\theoremstyle{theorem}
\newtheorem{thrm}[subsubsection]{Theorem}
\newtheorem{lemma}[subsubsection]{Lemma}
\newtheorem{coro}[subsubsection]{Corollary}

\newtheorem{prop}[subsubsection]{Proposition}
\newtheorem{ex}[subsubsection]{Example}

\theoremstyle{definition}
\newtheorem{defn}[subsubsection]{Definition}

\DeclareMathOperator{\assn}{assn}
\newcommand{\an}{\mathrm{an}}
\newcommand{\de}{\mathrm{de}}
\newcommand{\nondes}{\mathrm{nondes}}
\newcommand{\pa}{\mathrm{pa}}
\newcommand{\ch}{\mathrm{ch}}
\newcommand{\ma}{\mathrm{ma}}
\newcommand{\cc}{\mathrm{cc}}
\newcommand{\ngh}{\mathrm{ne}}

\newcommand{\G}{\mathcal{G}}
\newcommand{\D}{\mathcal{D}}
\newcommand{\U}{\mathcal{U}}

\newcommand{\M}{\mathcal{M}}
\newcommand{\E}{\mathcal{E}}

\newcommand\independent{\protect\mathpalette{\protect\independenT}{\perp}}
\def\independenT#1#2{\mathrel{\rlap{$#1#2$}\mkern2mu{#1#2}}}

\newcommand{\mathdash}{\relbar\mkern-9mu\relbar}

\newcommand\sbullet[1][.5]{\mathbin{\vcenter{\hbox{\scalebox{#1}{$\bullet$}}}}}

\title[Marginal independence structures underlying Bayesian Networks]{Combinatorial and algebraic perspectives on the marginal independence structure of Bayesian networks}

\author{Danai Deligeorgaki}
\author{Alex Markham}
\author{Pratik Misra}
\author{Liam Solus}
\address{Department of Mathematics, KTH Royal Institute of Technology, Sweden}
\email{\{danaide, markham, pratikm, solus\}@kth.se}


\begin{document}
\sloppy

\begin{abstract}
	We consider the problem of estimating the marginal independence structure of a Bayesian network from observational data, learning an undirected graph we call the unconditional dependence graph.
	We show that unconditional dependence graphs of Bayesian networks correspond to the graphs having equal independence and intersection numbers.
	Using this observation, a Gr\"obner basis for a toric ideal associated to unconditional dependence graphs of Bayesian networks is given and then extended by additional binomial relations to connect the space of all such graphs.
	An MCMC method, called \texttt{GrUES} (Gr\"obner-based Unconditional Equivalence Search), is implemented based on the resulting moves and applied to synthetic Gaussian data.
	\texttt{GrUES} recovers the true marginal independence structure via a penalized maximum likelihood or MAP estimate at a higher rate than simple independence tests while also yielding an estimate of the posterior, for which the $20\%$ HPD credible sets include the true structure at a high rate for data-generating graphs with density at least $0.5$.
\end{abstract}

\maketitle

\section{Introduction}
Directed acyclic graphs (DAGs) are used to model conditional independence and causal relations underlying complex systems of jointly distributed random variables.
For a DAG $\D = (V,E)$ with node set $V = \{v_1,\ldots, v_n\}$ and edge set $E$, the \emph{DAG model} $\M(\D)$ is the set of probability density functions $f(x_{v_1},\ldots, x_{v_n})$ satisfying
\begin{equation}
	\label{eqn: DAG model def} X_{v_i} \independent X_{V\setminus\nondes_\D(v_i)} | X_{\pa_\D(v_i)} \qquad \mbox{for all $i\in\{1,\ldots, n\}$},
\end{equation}
where $\pa_\D(v_i) = \{ v_j \in V : v_j\rightarrow v_i\in E\}$, and $\nondes_\D(v_i)$ is the set of $v_j\in V$ for which there is no directed path from $v_i$ to $v_j$ in $\D$.
A density is \emph{Markov} to $\D$ if it lies in $\M(\D)$.
Identifying a DAG to which a data-generating distribution is Markov provides rudimentary causal information about the distribution by interpreting \eqref{eqn: DAG model def} as: $X_{v_i}$ is independent of all variables not affected by \(X_{v_i}\), given its direct causes.

DAG models are fundamental in \emph{causal inference}, where the aim is to infer causal effects in a complex system \citep{pearl2009causality}.
This process often begins with \emph{causal discovery}, where one estimates a DAG to which the data-generating distribution is Markov.
The model $\M(\D)$ is characterized the set of conditional independence relations encoded by the \emph{d-separations} in $\D$ \citep{lauritzen1996graphical}.
Hence, DAGs with the same d-separations represent the same model and form a \emph{Markov equivalence class} (MEC), limiting identifiability.
With observational data alone, and no additional parametric assumptions on the data-generating distribution, we can only estimate a DAG up to its MEC \citep{pearl2009causality}.

In applications, such as in medicine and biology \citep{pertea2018chess,sachs2005causal}, one often uses additional data collected via \emph{interventional experiments} (e.g. randomized controlled trials) to refine an MEC.
Such experiments typically target a subset of variables in the system, and the choice of these targets affects which elements in the class can be rejected as candidates for the true causal system \citep{eberhardt2007interventions, hauser2012characterization, murphy2001active, tong2001active, wang2017permutation, yang2018characterizing}.
To do this efficiently, it is desirable to have good methods for identifying targets.
This problem is often addressed via budget-constraints where only a function of the causal graph is learned \citep{agrawal2019abcd, cho2016reconstructing, ness2017bayesian, murphy2001active, tong2001active} or by active learning methods that identify optimal targets given the graph estimate from previous experiments \citep{hauser2014two,he2008active}.
Such methods may be less desirable when a single experiment is time-consuming; for example, in large-scale knock-out experiments in gene regulatory networks \citep{mehrjou2021genedisco}.

An alternative approach is to identify a single set of targets for individual intervention by estimating a set of possible source nodes in the true underlying causal system.
Since $X_v \independent X_w$ in a distribution Markov to a DAG $\D$ for any two source nodes $v,w$ of $\D$, we can identify the collection of all \emph{marginally independent} nodes in the system, i.e., all pairs $v,w$ for which $X_v\independent X_w$.
Furthermore, models based on such low-order conditional independence relations can still provide useful estimates of causal effects \citep{wienobst2020recovering, wille2006low} and even isolate relevant biological processes \citep{magwene2004estimating, wille2004sparse}.
This can be useful in large systems where estimating a DAG may be infeasible.
Hence, estimating the marginal independence structure of the underlying DAG can provide useful information in causal inference.

In this paper, we develop the combinatorial and algebraic theory for modeling and estimating the marginal independence structure of a DAG model.
There are several contributions, which we break down in the following:

\subsection*{The combinatorics of unconditional equivalence} In Section~\ref{sec:uncond-depend}, we provide a framework for representing the marginal independence structure of a DAG using an (undirected) \emph{unconditional dependence graph} (UDG).
UDGs were previously studied in \citep{textor2015learning, pgm22, wienobst2020recovering}, in which characterizations of DAGs admitting the same marginal independence structure were derived.
In Theorem~\ref{thrm:alt-defs}, we add to this theory by providing four characterizations of the UDG of a DAG.
We call the set of all DAGs that have the same UDG an \emph{unconditional equivalence class} (UEC).
A UDG is thus a representation of a UEC, which we call a \emph{UEC-representative}.

Not all undirected graphs are UEC-representatives, but those that are possess several useful combinatorial properties.
In Theorem~\ref{thrm:alpha-equals-delta} we show that UEC-representatives are exactly the undirected graphs whose independence number and intersection number are equal.
We further observe that UEC-representatives possess a unique minimum edge clique cover that can be identified from any maximum independent set of nodes in the graph.
As a corollary, we show the generally NP-Hard problems of computing the independence and intersection numbers (as well as the associated maximum independent sets and minimum edge clique covers) are solvable in polynomial time for UEC-representatives.
This section is self-contained and accessible given a background in graph theory.

\subsection*{The algebra of unconditional equivalence}
In Section~\ref{subsec: grobner basis} we use our characterization of UEC-representatives to define a toric ideal whose associated fibers contain all UEC-representatives with a specified set of ``source nodes'' and pairwise intersections and unions of their neighbors.
A quadratic and square-free reduced Gr\"obner basis for this toric ideal is identified (Theorem~\ref{thrm:rgb for less variables}).
By the Fundamental Theorem of Markov Bases \citep{diaconis1998algebraic, petrovic2017survey}, the Gr\"obner basis gives a set of moves for exploring the UEC-representatives within a fiber.
In Section~\ref{sec: traversing}, we extend these moves via additional binomial operations to a set of moves that completely connects the space of all UEC-representatives on $n$ nodes (Theorem~\ref{thrm:exploring the entire space}).
The resulting connectivity theorem yields a method for exploring the space of UEC-representatives in the language of binomials, which is applied in Section~\ref{sec:gr-algor} to estimate the marginal independence structure of a DAG model.
This section uses classic results on Gr\"obner bases and toric ideals.
It makes use of the results derived in Section~\ref{sec:uncond-depend}.

\subsection*{Complexity reduction.}
Using the algebraic methods developed in Sections~\ref{sec:grobner} and~\ref{sec: traversing}, we obtain a search algorithm over the space of UEC-representatives on $n$ nodes in the language of polynomials.
However, the polynomials used in this search are computationally inefficient when implemented directly.
To reduce the complexity, in Section~\ref{sec:dag-reduct-repr}, we introduce the \emph{DAG-reduction} of a UEC-representative and prove that the algorithm can be rephrased in terms of DAG-reductions so as to reduce complexity.
This reduction in complexity makes feasible an implementation of the identified search method.
To do this, it is shown that there exists DAGs in a given UEC that are \emph{maximal} in the UEC with respect to edge inclusion.
We observe that these maximal DAGs form a MEC contained within the UEC.
It follows that every UEC can be identified with a unique MEC of DAGs.
The \emph{completed partially directed acyclic graph} (CPDAG) of this MEC is characterized, and then used to produce the DAG-reduction of the UEC, which is more computationally efficient than the UEC-representative in terms of both time and space complexity.
The results in this section are accessible to readers with knowledge of graphical models.

\subsection*{MCMC estimation of the marginal independence structure of a DAG model}
In Section~\ref{sec:grues:-markov-chain} the DAG reduction search in Section~\ref{sec: DAG reductions} is implemented in the form of a Markov Chain Monte Carlo method, called \emph{GrUES (Gr\"obner-Based Unconditional Equivalence Search)}.
\texttt{GrUES} can completely explore the space of UEC-representatives, thereby making possible the identification of an optimal UEC-representative for the data.
\texttt{GrUES} also yields an estimate of the posterior distribution of the UEC-representatives, allowing the user to quantify the uncertainty in the estimated marginal independence structure.

In subsection~\ref{sec:appl-synth-data}, we apply \texttt{GrUES} to synthetic data generated from random linear Gaussian DAG models to evaluate its performance empirically.
It is benchmarked against pairwise marginal independence testing, with performance evaluated for varying numbers of nodes, graph sparsity and choices of prior, including a noninformative prior as well as a prior that allows the user to incorporate beliefs about the number of source nodes in the data-generating causal system.


We observe that for relatively sparse or relatively dense models, \texttt{GrUES} successfully identifies the marginal independence structure of the data-generating model at a rate higher than that achieved via simple independence tests.
Highest Posterior Density (HPD) credible sets are also estimated that give relatively fine estimates of the true UDG.
These results suggest that \texttt{GrUES} provides an effective method for the estimation of the marginal independence structure of a DAG model, while allowing for the flexibility of incorporating prior knowledge about the causal system.


\section{Unconditional Dependence}
\label{sec:uncond-depend}
In this section we describe graphical representations for the marginal independence structure of a data-generating distribution Markov to a DAG $\D$.
Our representative of choice will be an undirected graph called the unconditional dependence graph of the DAG.
We first begin with some necessary preliminaries.

\subsection{Preliminaries}
\label{subsec: preliminaries}
Given a positive integer $n$, we let $[n]\coloneqq \{1,\ldots,n\}$.
Let $\D = (V^\D,E^\D)$ be a directed acyclic graph (DAG) with node set $V^\D$ and edge set $E^\D$.
When it is clear from context, we write $V$ and $E$ for the nodes and edges of $\D$, respectively.
If $|V| = n$ then the $n\times n$ matrix $A_\D = [a_{v,w}]$ in which \[ a_{v,w} =
	\begin{cases}
		1 & \mbox{if $v\rightarrow w\in E$}, \\ 0 & \mbox{otherwise}
	\end{cases}
\] is called the \emph{adjacency matrix} of $\D$.
In an adjacency matrix, we identify $V$ with $[n]$ and order the rows (columns) in increasing order from left-to-right (top-to-bottom).
The \emph{skeleton} of a DAG $\D$ is the undirected graph given by forgetting edge directions in $\D$.
For an undirected graph $\U = (V,E)$ on $n$ nodes, the adjacency matrix of $\U$ is $A_\U = [a_{v,w}]$ where $a_{v,w} = a_{w,v} = 1$ if $v \mathdash w\in E$ and $a_{v,w} = 0$ otherwise.
Vertices $v,w\in V$ are called \textit{adjacent} if $v\rightarrow w$, $w\rightarrow v$ or $v \mathdash w$ is in $E$.
Given an undirected graph $\U$ and a vertex $v\in V^\U$ we let $\ngh_\U(v)$ denote the \emph{open neighborhood} of $v$, i.e., the set of nodes adjacent to $v$, and $\ngh_\U[v]:=\ngh(v)\cup\{v\}$ be the \emph{closed neighborhood} of $v$.
We write $\ngh(v)$ and $\ngh[v]$ when the graph $\U$ is understood.
If $v\rightarrow w\in E$ then $v$ is a \emph{parent} of $w$ and $w$ is a \emph{child} of $v$.
A \emph{walk} is a sequence of nodes $(v_1,\ldots, v_m)$ such that $v_i$ and $v_{i+1}$ are adjacent for all $i\in[m-1]$.
A walk in which all nodes are distinct is a \emph{path}.
A walk $(v_1,\ldots, v_m)$ in $\D$ is \emph{directed} (from $v_1$ to $v_m$) if $v_i\rightarrow v_{i+1}\in E$ for all $i\in[m-1]$.
A directed walk in $\D$ is a directed path.
If there is a directed path from $v$ to $w$ in $\D$ we say $v$ is an \emph{ancestor} of $w$ and $w$ is a \emph{descendant} of $v$.
For $A\subseteq V$ we define the \emph{parents}, \emph{children}, \emph{ancestors}, and \emph{descendants} of $A$ to be the union over all parents, children, ancestors and descendants of all nodes in $A$, respectively.
We let $\pa_\D(A), \ch_\D(A), \de_\D(A),$ and $\an_\D(A)$ denote the set of parents, children, descendants and ancestors of $A$ in $\D$, respectively.
Note that $v\in\de_\D(v)$ and $v\in\an_\D(v)$.
When the DAG $\D$ is understood, we drop the subscript $\D$.
A \textit{collider} is a pair of edges $t \rightarrow u, w \rightarrow u$ (also written \(t \rightarrow u \leftarrow w\)).
If $t$ and $w$ are nonadjacent, then \(t \rightarrow u \leftarrow w\) is further called a v-\emph{structure}.
If a path contains the edges $t \rightarrow u$ and $w \rightarrow u$, then the vertex $u$ is called a \textit{collider} on the path.
A path is called \textit{blocked} if it contains a collider.
A colliderless path that does not repeat any vertex is called a \textit{(simple) trek}.
(Note that trek has a more general definition where the edges and vertices can be repeated. 
We will refer to this more general trek as a colliderless walk.)
Two subsets $A$ and $B$ of $V$ are \textit{$d$-connected} given $\emptyset$ if and only if there is a trek between some $v \in A$ and $w \in B$.
We let $A\not\perp_\D B $ denote that $A$ and $B$ are $d$-connected given $\emptyset$ in $\D$.
We say that $A$ and $B$ are \emph{$d$-separated} given $\emptyset$, denoted \(A \perp_\D B\), if they are not $d$-connected given $\emptyset$.

Here, we only defined d-connected and d-separated given the empty set, as this will be sufficient for this paper.
A more general definition in which $A$ and $B$ are d-connected (d-separated) given a possibly nonempty set $C$ is used to describe the conditional independence relations associated to a DAG $\D$.
Given a DAG $\D = (V,E)$ the \emph{DAG model} $\M(\D)$ is the collection of all distributions that are Markov to $\D$ (according to equation~\eqref{eqn: DAG model def}).

\begin{thrm}[\citet{lauritzen1996graphical}]
	\label{thrm:factorization}
	The distribution of $(X_1,\ldots, X_n)$ belongs to $\M(\D)$ if and only if $X_A\independent X_B | X_C$ whenever $A$ and $B$ are d-separated given $C$ in $\D$.
\end{thrm}

An important observation to be made from the above theorem is that two different DAGs $\D$ and $\D^\prime$ can satisfy $\M(\D) = \M(\D^\prime)$ since it is possible that $\D$ and $\D^\prime$ have the same set of d-separation statements.
Two such DAGs are called \emph{Markov equivalent} and are said to belong to the same \emph{Markov equivalence class (MEC)}.

\subsection{The unconditional dependence graph of a DAG}
\label{subsec: unconditional dependence graph}
When considering jointly distributed random variables $X = (X_1,\ldots, X_n)$, the term \emph{unconditional independence} or \emph{marginal independence} refers to conditional independence statements of the form $X_A \independent X_B | X_C$ where $C = \emptyset$, i.e., the \emph{independence} relations $X_A \independent X_B$ that hold in the joint distribution.
If $X$ is Markov to a DAG $\D$ in which $i$ and $j$ are distinct source nodes of $\D$ then $X_i\independent X_j$.
Hence, learning the marginal independence structure of a model allows us to identify disjoint sets of nodes that contain candidate source nodes for the DAG model, which can be useful in the context of causal inference.
This motivates the following definition.
\begin{defn}
	\label{defn:udg}
	The \emph{unconditional dependence graph} of a DAG $\D = (V,E)$ is the undirected graph \(\U^\D = (V,\{\{v,w\} : v\not\perp_{\D} w; \;v, w \in V\})\).
\end{defn}

When the DAG $\D$ is clear from context, we write $\U$ for $\U^\D$.
Similar to the case of Markov equivalence of DAGs, it is possible that two distinct DAGs $\D$ and $\D^\prime$ encode the same set of unconditional d-separation statements $v\perp w$.
Two DAGs $\D = (V,E^\D)$ and $\D' = (V,E^{\D'})$ are said to be \emph{unconditionally equivalent} if whenever two nodes $i,j\in V$ are $d$-separated given $\emptyset$ in $\D$, the nodes $i$ and $j$ are d-separated given $\emptyset$ in $\D'$, i.e., $\U^\D = \U^{\D^\prime}$.
\citet[Lemma~5]{markham2022} show that unconditional equivalence is indeed an equivalence relation over the family of \emph{ancestral graphs} (see \citep{richardson2002ancestral} for a definition) and consequently is also an equivalence relation over DAGs.
The collection of all DAGs that are unconditionally equivalent to $\D$ is called its \emph{unconditional equivalence class (UEC)}.
We represent each unconditional equivalence class of DAGs by their unconditional dependence graph as \[\{\U\}\coloneqq \{\D : \U^{\D}=\U\}.
\]
This is a collection of DAGs that is possibly different from the MEC of $\D$.
Since UECs are defined in terms of a subset of the \(d\)-separations in a DAG, the partition of DAGs on $n$ nodes into MECs is a refinement of the partition of DAGs into UECs; i.e., each UEC can be written as a union of certain MECs.
Hence, esimating the unconditional dependence graph $\U$ of a DAG $\D$ gives a representative of all MECs of DAGs that encode the same set of unconditional independence relations.

We now derive four characterizations of the unconditional dependence graph $\U^\D$ of a DAG $\D$ to be used in methods for estimating the marginal independence structure of $\D$ from data.
These characterizations are presented in Theorem~\ref{thrm:alt-defs}, whose statement requires the following definitions:

We say that an ordered pair $(v,w)$ (or an edge $v\rightarrow w$) is \emph{implied by transitivity} in $\D$ if $v\in\an_\D(w)\setminus(\{w\}\cup\pa_\D(w))$.
The set of \emph{maximal ancestors} of $A$ in $\D$, denoted $\ma_\D(A)$, is the set of all $v\in \an_\D(A)$ for which $\an_\D(v) = \{v\}$.
A node $v\in V$ is called a \emph{source} node of $\D$ if $\pa_\D(v) = \emptyset$.
It follows that $\ma_\D(V)$ is the collection of all source nodes in $\D$.
We say an ordered pair $(v,w)$ (or an edge $v\rightarrow w$) is \emph{partially weakly covered} if $\ma_\D(v)\subseteq\ma_\D(w)$, $v\notin\an_\D(w)$, $w\notin\an_\D(v)$ and $\pa_\D(v) \neq \emptyset$.
When the DAG $\D$ is understood from context, we simply write $\ma(A)$ for $\ma_\D(A)$.
The following gives an example of the various definitions presented thus far.

\begin{ex}
	\label{ex:unconditional dependence}
	Let $\D$ and $\D'$ be two DAGs as shown in Figure~\ref{figure:2 DAGs}.
	In $\D$, $\pa_\D(4) = \{1,3\}$ and $\an_\D(4) = \{1,2,3,4\}$.
	The ordered pair $(2,4)$ is implied by transitivity.
	Vertices $1$ and $2$ are source nodes in both $\D$ and $\D^\prime$.
	In both graphs, $1$ and $2$ are the only vertices $d$-separated given $\emptyset$, as there is no trek connecting $1$ and $2$ but one exists between all other pairs.
	Hence, $\D$ and $\D'$ have the same unconditional dependence graph $\mathcal{U}$ and belong to the same UEC.
	The graph $\U$ is the undirected graph on node set $[4]$ in which the only missing edge is $1 \mathdash 2$.
	The absence of this single edge indicates that $1$ and $2$ must be the source nodes in any DAG in the UEC.

	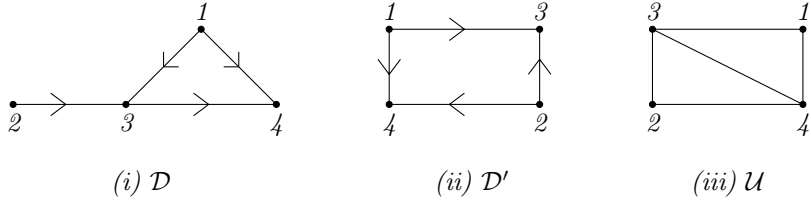
\begin{figure}
		\begin{tikzpicture}
			\filldraw[black]
			(-1.5,0) circle [radius=.04] node [below] {2}
			(0,0) circle [radius=.04] node [below] {3}
			(1,1) circle [radius=.04] node [above] {1}
			(2,0) circle [radius=.04] node [below] {4}
			(.2,-.75) circle [radius=0] node [below] {(i) $\D$};
			\draw
			(-1.5,0)--++(3.5,0)
			(0,0)--++(1,1)
			(1,1)--++(1,-1);
			\draw
			(-1,.15)--(-.8,0)
			(-1,-.15)--(-.8,0)
			(0.5,0.7)--(.5,.5)
			(0.7,.5)--(.5,.5)
			(.9,.15)--(1.1,0)
			(1.1,0)--(.9,-.15)
			(1.5,0.7)--(1.5,.5)
			(1.3,.5)--(1.5,.5);


			\filldraw[black]
			(3.5,0) circle [radius=.04] node [below] {4}
			(5.5,0) circle [radius=.04] node [below] {2}
			(3.5,1) circle [radius=.04] node [above] {1}
			(5.5,1) circle [radius=.04] node [above] {3}
			(4.6,-.75) circle [radius=0] node [below] {(ii) $\D'$};

			\draw
			(3.5,0)--(5.5,0)
			(3.5,0)--(3.5,1)
			(5.5,0)--(5.5,1)
			(3.5,1)--(5.5,1)

			(4.3,0)--(4.5,.15)
			(4.3,0)--(4.5,-.15)
			(3.35,0.6)--(3.5,0.4)
			(3.5,0.4)--(3.65,0.6)
			(4.5,1)--(4.3,1.15)
			(4.5,1)--(4.30,.85)
			(5.5,0.6)--(5.35,0.4)
			(5.5,0.6)--(5.65,0.4)
			;


			\filldraw[black]
			(7,0) circle [radius=.04] node [below]
				{2}
			(9,0) circle [radius=.04] node [below]
				{4}
			(7,1) circle [radius=.04] node [above]
				{3}
			(9,1) circle [radius=.04] node [above]
				{1}
			(8,-.75) circle [radius=0] node [below] {(iii) $\mathcal{U}$};
			\draw
			(7,0)--(9,0)
			(7,0)--(7,1)
			(9,0)--(9,1)
			(7,1)--(9,1)
			(7,1)--(9,0);
		\end{tikzpicture}
		\caption{DAGs $\D$ and $\D'$ that have the same unconditional dependence graph $\mathcal{U}$.}\label{figure:2 DAGs}
	\end{figure}
\end{ex}

Some operations on a DAG $\D$ are necessary for the statement of Theorem~\ref{thrm:alt-defs}.
Let $t(\D)$ denote the \emph{transitive closure} of \(\D\), i.e., the DAG given by iteratively adding all edges \(v \rightarrow w\) for any pair of nonadjacent nodes \(v, w \in V\) for which $E^\D$ contains the edges \(v \rightarrow v'\) and \(\ v' \rightarrow w\).
The operator $r(\D)$ reverses all edges in $\D$; i.e., replacing $v\rightarrow w$ with $v\leftarrow w$ for all $v\rightarrow w\in E^\D$.
Let $m(\D)$ denote the undirected graph formed by adding an edge between each pair of nonadjacent nodes in $V$ that have a common child, and then undirecting all edges.
This operation is called \emph{moralization} in the graphical models literature.
The operator \(a(M)\) turns a given symmetric matrix \(M\) into a $0,1$-matrix via the component-wise mapping \[ a(M)_{v,w} \coloneqq
	\begin{cases}
		1,\ \text{if}\ v \not= w\ \text{and}\ M_{v,w} \not= 0 \\ 0,\ \text{otherwise},
	\end{cases}
	.
\]
Finally, we let \(M^\top\) denote the transpose of the matrix $M$.

Using the above notation and operators, we now state our characterizations of the unconditional dependence graph of a DAG $\D$, extending a result in \citep{pgm22}.

\begin{thrm}
	\label{thrm:alt-defs}
	Let $\D = (V,E)$ be a DAG with $n = |V|$ nodes.
	The unconditional dependence graph $\U$ of $\D$ is equal to each of the following:
	\begin{enumerate}
		\item \(\U_1 = (V, \{\{v,w\} : \mathrm{an}_\D(v) \cap \mathrm{an}_\D(w) \not= \emptyset\})\); that is, two distinct nodes share an edge in \(\U\) if and only if they have a common ancestor in \(\D\),
		\item \(\U_2 = \left(V, \bigcup_{m \in \ma_\D(V)} \{\{v, w\} \in \de_\D(m) \times \de_\D(m) : v \not= w\}\right)\),
		\item \(\U_3 = m(t(r(\D)))\), and
		\item $\U_4$ described by the adjacency matrix \(A_{\U_4} = a(T^\top T)\), where \(T = \sum\limits_{p = 0}^{n -1} A^p_\D\).
	\end{enumerate}
\end{thrm}

\begin{proof}

	The equality between $\U$, $\U_1$ and $\U_2$ is shown by \citet[Theorem~1]{pgm22}.
	We will show that $\U_1=\U_3$ and $\U_1=\U_4$.

	For \(\U_3\), observe that \(r\) swaps all child/parent relations in \(\D\) (i.e., reverses the direction of the edges), \(t\) makes all ancestors into parents (i.e., for every node, \(t\) adds an edge between that node and each of its descendants), and \(m\) adds an undirected edge between two parents if they share a child and undirects all directed edges.
	Notice that if nodes \(u,v\) have a common ancestor in \(\D\) then \(r\) makes them have a common descendant.
	Applying \(t\) then makes them have a common child, and finally applying \(m\) makes them adjacent in \(\U_3\).
	Hence, if two nodes in $\D$ have a common ancestor, they are adjacent in $\U_3$.
	Conversely, if two nodes are adjacent in $\U_3$ then they were either (1) adjacent in $\D$, (2) not adjacent in $\D$ but adjacent in $t(r(\D))$, or (3) not adjacent in $t(r(\D))$ but adjacent in $m(t(r(\D)))$.
	In cases (1) and (2), the two nodes have a common ancestor in $\D$ as they must have been connected by a directed path in $\D$.
	In case (3), the two nodes have a common ancestor in $\D$ corresponding to one of their children in $t(r(\D))$.
	Hence, distinct nodes are adjacent in \(\U_3\) if and only if they have a common ancestor in \(\D\).
	So \(\U_3 = \U_1\).

	For $\U_4$, recall that the $(v,w)$ entry of $A_{\D}^p$ is the number of directed paths of length $p$ from a vertex $v$ to a vertex $w$ in $\D$ (see, for example, \citep{diestel2005graph}).
	Now let $U\coloneqq T^\top T = [u_{v,w}]_{v,w = 1}^n$ and $T = [t_{v,w}]_{v,w=1}^n$ where we identify $V$ with $[n]$.
	The sum over \(p \in \{0,\ldots, n-1\}\) in the computation of \(T\), results in the entry \(t_{v,w}\) being nonzero if and only if \(v \in \mathrm{an}_\D(w)\).
	Considering the matrix product entries \(u_{v,w} = T^{\top}_{v, \sbullet } T_{\sbullet,w}\), notice that \(u_{v,w}\) is nonzero if and only if there exists some \(c\) such that both \(t_{c, v}\) and \(t_{c,w}\) are nonzero, i.e., such that \(c \in \an_\D(v) \cap \an_\D(w)\).
	Thus, for distinct vertices $v,w$, we have \(a(T^\top T)_{v,w}=1\) if and only if $v$ and $w$ are adjacent in $\U_1$.
	Moreover, by definition \(a(T^\top T)_{v,v}=0\), which completes the proof.

\end{proof}

Note that the definition of $\U_1$ presented in Theorem~\ref{thrm:alt-defs} could equivalently be phrased in terms of treks.
Namely, the edge set of $\U_1$ contains the edge $v \mathdash w$ if and only if there is a trek between $v$ and $w$ in $\D$.

\begin{ex}
	\label{ex:equivalent definitions}
	Consider the DAG \(\D\) and its unconditional dependence graph $\U$, as seen in Example~\ref{ex:unconditional dependence}, (i) and (iii).
	We will now demonstrate the equality between $\U$ and the graphs \( \U_1, \U_2, \U_3,\) and \(\U_4\) arising for $\D$ in Theorem~\ref{thrm:alt-defs}.

	Observe that all pairs of vertices in $\D$ except $(1,2)$ share a common ancestor (i.e., are connected by a trek), which implies that $\U_1$ is exactly equal to $\U$.
	For $\U_2$, observe that the maximal ancestors of $\D$ are $1$ and $2$.
	The descendants of $1$ and $2$ in $\D$ are $\{1,3,4\}$ and $\{2,3,4\}$ respectively.
	Taking the union of the cross products $\de_{\D}(1) \times \de_{D}(1)$ and $\de_{\D}(2) \times \de_{D}(2)$ gives us the edge set $\{(1,3),(1,4),(3,4),(2,3),(2,4)\}$, which is precisely the edge set of $\U$.
	The construction of $\U_3$ in Figure~\ref{figure:dependence graph U_2} shows that $\U_3=\U$.
	Finally, for $\U_4$, we first look at the adjacency matrix $A_{\D}$.
	We have \[ A_\D^0 = {I}_4, \quad A_{\D}=
		\begin{bmatrix}
			0 & 0 & 1 & 1 \\ 0 & 0 & 1 & 0 \\ 0 & 0 & 0 & 1 \\ 0 & 0 & 0 & 0
		\end{bmatrix}
		, \quad A_{\D}^2=
		\begin{bmatrix}
			0 & 0 & 0 & 1 \\ 0 & 0 & 0 & 1 \\ 0 & 0 & 0 & 0 \\ 0 & 0 & 0 & 0
		\end{bmatrix}
		, \quad A_{\D}^3= [0].
	\]
	This gives us
	\[
		T=
		\begin{bmatrix}
			1 & 0 & 1 & 2 \\
			0 & 1 & 1 & 1 \\
			0 & 0 & 1 & 1 \\
			0 & 0 & 0 & 1
		\end{bmatrix}
		\quad \mbox{ and } \quad
		T^\top T=
		\begin{bmatrix}
			1 & 0 & 1 & 2 \\ 0 & 1 & 1 & 1 \\ 1 & 1 & 3 & 4 \\ 2 & 1 & 4 & 7
		\end{bmatrix}
		.
	\]
	Since the $(1,2)$ entry of $T^\top T$ is zero and the rest of the entries above the main diagonal are nonzero, we have that $a(T^\top T)$ is exactly the adjacency matrix of $\U$.
	Hence, $\U,\U_1,\U_2,\U_3,$ and $\U_4$ are indeed all equivalent.
	\begin{figure}
		\begin{tikzpicture}

			\filldraw[black]
			(-1.5,0) circle [radius=.04] node [below] {2}
			(0,0) circle [radius=.04] node [below] {3}
			(1,1) circle [radius=.04] node [above] {1}
			(2,0) circle [radius=.04] node [below] {4}
			(.2,-.9) circle [radius=0] node [below] { $(i)\;r(\D)$};
			\draw
			(-1.5,0)--++(3.5,0)
			(0,0)--++(1,1)
			(1,1)--++(1,-1);
			\draw
			(-.8,0)--(-.6,.15)
			(-.8,0)--(-.6,-.15)
			(0.4,0.6)--(.6,.6)
			(0.6,.4)--(.6,.6)
			(1.05,.15)--(.9,0)
			(.9,0)--(1.05,-.15)
			(1.4,0.6)--(1.6,.6)
			(1.4,.4)--(1.4,.6);

			\filldraw[black]
			(3,0) circle [radius=.04] node [below] {2}
			(4.5,0) circle [radius=.04] node [below] {3}
			(5.5,1) circle [radius=.04] node [above] {1}
			(6.5,0) circle [radius=.04] node [below] {4}
			(4.7,-.9) circle [radius=0] node [below] { $(ii)\;t(r(\D))=m(t(r(\D)))$};
			\draw
			(3,0)--(6.5,0)
			(4.5,0)--(5.5,1)
			(5.5,1)--(6.5,0);
			\draw [dashed](3,0) .. controls (4.5,-.77) .. (6.5,0);
			\draw
			(3.7,0)--(3.9,.15)
			(3.7,0)--(3.9,-.15)
			(4.9,0.6)--(5.1,.6)
			(5.1,.4)--(5.1,.6)
			(5.55,.15)--(5.4,0)
			(5.4,0)--(5.55,-.15)
			(5.9,0.6)--(6.1,.6)
			(5.9,.4)--(5.9,.6)
			(5,-.52)--(5.15,-.3)
			(5,-.52)--(5.2,-.58);

			\filldraw[black]
			(8,0) circle [radius=.04] node [below]
				{2}
			(10,0) circle [radius=.04] node [below]
				{4}
			(8,1) circle [radius=.04] node [above]
				{3}
			(10,1) circle [radius=.04] node [above]
				{1}
			(9,-.75) circle [radius=0] node [below] {$(iii)\;\U_3$};
			\draw
			(8,0)--(10,0)
			(8,0)--(8,1)
			(10,0)--(10,1)
			(8,1)--(10,1)
			(8,1)--(10,0);
		\end{tikzpicture}
		\caption{Constructing $\U_3$ for the DAG $\D$ in Figure \ref{figure:2 DAGs}.}\label{figure:dependence graph U_2}
	\end{figure}
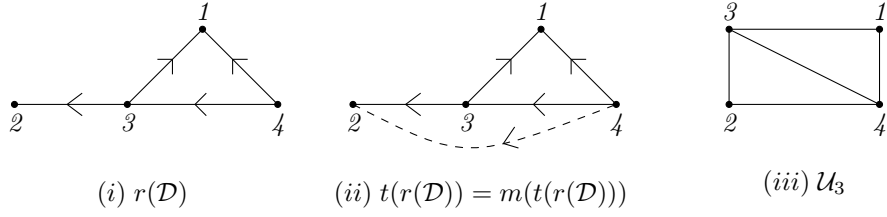
\end{ex}

Note that Theorem~\ref{thrm:alt-defs}~(4) gives a computationally efficient method for identifying the unconditional dependence graph of a given DAG.
Characterizations~(1) and~(2) play a fundamental role in the transformational characterization of unconditional equivalence of DAGs given in \citep[Theorem~9]{pgm22}.

Theorem~\ref{thrm:alt-defs} also yields other interesting observations about special families of DAGs in relation to their marginal independence structure.
Consider the DAGs $\D = ([n], E)$ in which $i<j$ whenever $i\rightarrow j\in E$.
These are precisely the DAGs on $n$ nodes whose adjacency matrix is upper triangular when the rows and columns are ordered as $1,\ldots, n$ from left-to-right (top-to-bottom).
For such DAGs, which are commonly considered when studying algebraic properties of DAG models \citep{misra2022directed}, it can be shown that $\U$ is uniquely determined by the DAG $\D$.
Furthermore, the entries of $\U$ capture structural information about $\D$.

\begin{lemma}
	\label{lemma:bijection-A_D-to-U}
	Let \(\mathbb{A}\) be the set of upper triangular adjacency matrices of DAGs and define \(\mathbb{Y} \coloneqq \{U_\D : A_\D \in \mathbb{A}\}\), where $U_\D=T^\top T$ and \(T = \sum\limits_{p = 0}^{n -1} A^p_\D\).
	Then the map \(f\) where \(f(A_\D)= U_\D\) is a bijection, and \(a(U_\D)\) is the adjacency matrix of \(\U^\D\).
\end{lemma}

\begin{proof}
	The property $a(U_\D)=A_{\U^{\D}}$ was shown in Theorem~\ref{thrm:alt-defs} for any DAG $D$ in the UEC of $\U^\D$.
	As seen in the proof of Theorem~\ref{thrm:alt-defs}, the $(i,j)$ entry of $A_{\D}^p$ is the number of directed paths of length $p$ from $v_i$ to $v_j$ in $\D$.
	From this, we get that the entry $t_{i,j}$ of $T$ counts the total number of directed paths from $v_i$ to $v_j$ in $\D$.
	Therefore, the entries $t_{i,j}\in T$ and $a_{i,j}\in A_\D$ satisfy \(t_{i,i+1}=a_{i,i+1}\) for every $i \in [n-1]$.
	Furthermore, the following recurrence relation holds: \[ t_{i,j}=a_{i,j}+\sum_{i<m<j} t_{i,m}a_{m,j}; \] i.e., the directed paths from $v_i$ to $v_j$ either have length $1$ (or 0, if $i=j$) or length $l>1$; and a directed path of length $l$ can be seen as a directed edge from $v_m \in \pa(v_j)$ to $v_j$, together with a directed path of length $l-1$ from $v_i$ to $v_m$.
	Therefore, $t_{i,j}$ only depends on the entries $t_{i,m}$ of $T$ for which $m-i<j-i$, and the entries of $A_\D$.
	Since the values of $t_{i,i+1}$ are given by $A_\D$, each entry $t_{i,j}$ for $j>i+1$ is also uniquely determined by the entries of $A_\D$.
	So far, we saw that different matrices $T$ correspond to graphs with different upper triangular adjacency matrices.

	Now, $T$ is an upper triangular matrix with nonzero entries in the diagonal.
	It follows that $T$ is invertible and $Tx\neq \mathbf{0}$ for any nonzero vector $x$ in $\mathbb{R}^n$.
	Thus, \[x^\top U_\D x=x^\top T^\top Tx=||Tx||^2 > 0, \text{ for all } x \in \mathbb{R}^n\setminus \{\mathbf{0}\}.
	\]
	This implies that $U_\D$ is positive definite as well.
	But every symmetric positive definite matrix \(A\) has a unique Cholesky decomposition \(A = LL^\top\), where \(L\) is a lower triangular matrix.
	Therefore, the matrix $U_\D$ is unique for each DAG $\D$.
\end{proof}

\begin{lemma}
	The entry \(u_{i,j}\) of $U_\D=[u_{i,j}]^n_{i,j=1}$ counts the number of colliderless walks between \(v_i\) and \(v_j\) in the DAG \(\D\).
\end{lemma}
\begin{proof}
	As discussed in the proof of Theorem~\ref{thrm:alt-defs}, $t_{i,j}$ counts the number of directed paths from $v_i$ to $v_j$ in $\D$.
	Notice now that a trek in $\D$ between $v_i$ and $v_j$ is a pair of directed paths from some node $v_k$, where one path is directed toward $v_i$ and the other path is directed toward $v_j$.
	Since $u_{i,j} = \sum_{k=1}^nt_{k,i}t_{k,j}$ then $u_{i,j}$ counts the number of ordered pairs $(s_1,s_2)$ where $s_1$ is a directed path from $v_k$ to $v_i$ and $s_2$ is a directed path from $v_k$ to $v_j$ in $\D$ for $k\in[n]$.
	However, combining any two directed paths does not necessarily form a trek as there could be repetition of vertices and edges.
	Hence, $u_{i,j}$ gives us the number of colliderless walks between $v_i$ and $v_j$.
\end{proof}

\subsection{Undirected graphs that are unconditional dependence graphs}
\label{subsec: characterizing UDGs}
The unconditional dependence graph $\U^\D$ of a DAG $\D$ is a representative of the UEC of DAGs $\{\U^\D\}$.
Not all undirected graphs represent a UEC; i.e., the set $\{\U\}$ is empty for certain undirected graphs $\U$.
On the other hand, every undirected graph represents some marginal independence model.
Namely, given nodes $[n]$, we can specify a collection $S$ of pairs of elements of $[n]$.
The marginal independence model $\mathcal{M}(S)$ then consists of all joint distributions $(X_1,\ldots, X_n)$ satisfying $X_i \independent X_j$ whenever $\{i,j\}\in S$.
The graph $\U = ([n], [n]\times[n]\setminus S)$ is the unconditional dependence graph representing the model $\mathcal{M}(S)$.
Since $S$ can be arbitrary, any undirected graph can be viewed as the unconditional dependence graph of some marginal independence model.
General marginal independence models are studied in \citep{boege2022marginal}.
We are interested in those marginal independence models that come from DAG models; i.e., the unconditional dependence graphs $\U = \U^\D$ for some DAG $\D$.
We call the undirected graphs that represent a nonempty UEC of DAG \emph{UEC-representatives}.

\begin{defn}
	\label{def:UEC-rep}
	An undirected graph $\U$ for which $\U = \U^\D$ for some DAG $\D$ is called a \emph{UEC-representative}.
\end{defn}

In this subsection, we prove that UEC-representatives are exactly those undirected graphs whose \textit{intersection} and \textit{independence} numbers are equal.
Although computing the intersection and independence number is NP-hard in general, it can be done in polynomial time for UEC-representatives.
We begin by defining the minimum edge clique cover, which is essential for computing the intersection number of a graph.

\begin{defn}
	Let $\U$ be an undirected graph.
	\begin{enumerate}
		\item A subset of vertices of $\U$ is a \emph{clique} if every pair of vertices in the subset is adjacent in $\U$.

		\item An edge clique cover $\E$ of $\U$, i.e., a  collection of cliques where every edge of $\U$ is contained in at least one clique in $\E$, is called a \emph{minimal edge clique cover} of $\U$ if no proper subset of $\E$ satisfies this property \citep{roberts1985applications}.

		\item A \emph{minimum edge clique cover} $\E$ of $\U$ is a minimal edge clique cover of minimum cardinality.
		      The size of $\E$ is called the \emph{intersection number} of $\U$, denoted by $\delta(\U)$.

		\item A set ${I}\subseteq V^\U$ is said to be \textit{independent} if all vertices in ${I}$ are pairwise nonadjacent.
		      The \textit{independence number} of $\U$, denoted by $\alpha(\U)$, is the maximum cardinality of an independent set of $\U$.
	\end{enumerate}
\end{defn}

We illustrate these definitions with an example.

\begin{ex}
	\label{ex:edge clique cover}
	The graph $\U$ in Figure~\ref{figure:2 DAGs} (iii) has independence number $\alpha(\U)=2$, since $\{1,2\}$ is the maximum independent set, and intersection number $\delta(\U)=2$, corresponding to the minimum edge clique cover $\{\{1,3,4\},\{2,3,4\}\}$.
	The graph $\U_2$ in Figure~\ref{figure:intersection number different than independence number} (i) has $\alpha(\U_2)=2$ but $\delta(\U_2)=3$, since $\{\{1,2,3\},\{3,4\},\{4,5,6\}\}$ is the minimum edge clique cover for $\U_2$.
	Note that minimum edge clique covers are not necessarily unique for undirected graphs.
	For example, the edge clique covers \begin{equation*}
		\begin{split}
			&\{\{1,2,3\}, \{3,4,5\}, \{5,6,7\}, \{7,8,1\}, \{\underline{1},3,7\}\} \quad \mbox{and}\\ &\{\{1,2,3\}, \{3,4,5\}, \{5,6,7\}, \{7,8,1\}, \{\underline{5},3,7\}\}
		\end{split}
	\end{equation*} are both minimum over the graph $\U_3$ shown in Figure~\ref{figure:intersection number different than independence number} (ii).
	Maximum independent sets of a graph are also not necessarily unique.
	For instance, the sets $\{1,4\}$ and $\{1,5\}$ are both maximum independent sets of $\U_2$.

	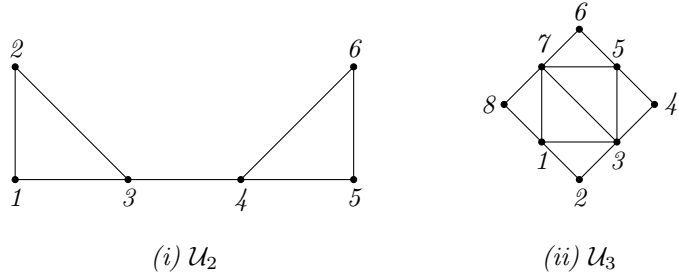
\begin{figure}
		\begin{tikzpicture}

			\filldraw[black]
			(2,0) circle [radius=.04] node [below]
				{1}
			(3.5,0) circle [radius=.04] node [below]
				{3}
			(2,1.5) circle [radius=.04] node [above]
				{2}
			(5,0) circle [radius=.04] node [below]
				{4}
			(6.5,0) circle [radius=.04] node [below]
				{5}
			(6.5,1.5) circle [radius=.04] node [above]
				{6}
			(4.25,-.75) circle [radius=0] node [below] {(i) $\U_2$};
			\draw
			(3.5,0)--(5,0)
			(2,0)--(2,1.5)
			(2,0)--(3.5,0)
			(2,1.5)--(3.5,0)
			(5,0)--(6.5,1.5)
			(6.5,0)--(6.5,1.5)
			(5,0)--(6.5,0);

			\filldraw[black]
			(9,0.5) circle [radius=.04] node [below]
				{1}
			(10.5,1) circle [radius=.04] node [right]
				{4}
			(9.5,2) circle [radius=.04] node [above]
				{6}
			(8.5,1) circle [radius=.04] node [left]
				{8}
			(9.5,0) circle [radius=.04] node [below]
				{2}
			(10,0.5) circle [radius=.04] node [below]
				{3}
			(9,1.5) circle [radius=.04] node [above]
				{7}
			(10,1.5) circle [radius=.04] node [above]
				{5}
			(9.5,-.75) circle [radius=0] node [below] {(ii) $\U_3$};
			\draw
			(9,0.5)--(10,0.5)
			(9.5,2)--(10,1.5)
			(9.5,2)--(9,1.5)
			(8.5,1)--(9,0.5)
			(8.5,1)--(9,1.5)
			(10.5,1)--(10,0.5)
			(10.5,1)--(10,1.5)
			(9,0.5)--(9.5,0)
			(9.5,0)--(10,0.5)
			(9,0.5)--(9,1.5)
			(10,0.5)--(10,1.5)
			(9,1.5)--(10,1.5)
			(9,1.5)--(10,0.5);

		\end{tikzpicture}
		\caption{Examples of undirected graphs with intersection number that exceeds their independence number.}\label{figure:intersection number different than independence number}
	\end{figure}
\end{ex}

In Example~\ref{ex:edge clique cover}, the intersection and independence number of $\U$ are equal.
Furthermore, Examples~\ref{ex:unconditional dependence} and~\ref{ex:equivalent definitions} show that the UEC of $\U$ is nonempty.
This is an important connection, which we state in the next theorem.

\begin{thrm}
	\label{thrm:alpha-equals-delta}
	An undirected graph $\U$ is a UEC-representative if and only if \(\alpha(\U) = \delta(\U)\).
\end{thrm}
\begin{proof}
	In \cite[Lemma~2]{pgm22} it is shown that a minimal edge clique cover of $\U$ is given by considering a DAG $\G$ in the corresponding UEC, and taking the cliques corresponding to each of the source nodes in $\G$ together with their descendants.
	It is also observed in this same lemma that the source nodes in $\G$ form a maximum independent set in $\U$.
	(Note that if $c$ is the total number of the cliques specified above then any set of $m$ nodes with $m\geq c$ must contain at least two adjacent nodes.)
	Additionally, one can observe that $\alpha(U) \leq \delta(U)$ for any undirected graph $U$.
	This is due to the fact that all nodes in a maximum independent set must necessarily be in disjoint cliques in a minimum edge clique cover.
	Hence, since the minimal edge clique cover identified for $\U$ in \cite[Lemma~2]{pgm22} has one clique for each node in a maximum independent set for $\U$, it follows that this cover is in fact minimum.
	In particular, we have that if $\U$ is a UEC-representative then $\alpha(\U) = \delta(\U)$.

	Let \(\alpha(\U) = \delta(\U)\).
	Then, there is a bijection between the nodes in a maximum independent set and the cliques in a minimum edge clique cover of $\U$.
	Given a maximum independent set $M$, let $\E^\U=\{C_m: m\in M\}$ be the corresponding minimum edge clique cover of $\U$.
	Moreover, the nodes in $M$ are necessarily in distinct cliques in $\U$.
	Now, consider the DAG \(\D = \{V^\U, m\rightarrow v : m \in M\text{ and } v \in C_m\}\), that is, the graph where there is an edge from a vertex in the maximum independent set of \(\U\) to each of the other vertices in its corresponding clique in \(\E^\U\).
	Using characterization 4 of Theorem~\ref{thrm:alt-defs}, it follows that \(\U^\D = \U\), so the set of DAGs equivalent to \(\U\) is nonempty.
\end{proof}

Theorem~\ref{thrm:alpha-equals-delta} will play an important role in Section~\ref{sec:grobner} where we use the fact (observed in \citep[Lemma~2]{pgm22}) that a maximum independent set in a UEC-representative $\U$ is given by the collection of source nodes for any given DAG in the UEC $\{\U\}$.
Specifically, Theorem~\ref{thrm:alpha-equals-delta} allows us to develop a sufficient statistic for UECs whose fiber is searchable via a Markov basis given by a Gr\"obner basis for the corresponding toric ideal.
We will also make use of the following facts.

\begin{lemma}
	\label{lem: unique minimum edge clique cover}
	A UEC-representative $\U$ has a unique minimum edge clique cover.
\end{lemma}

\begin{proof}
	Let $\E^\U=\{C_1,\ldots,C_k\}$ be a minimum edge clique cover of a UEC-representative $\U$.
	Since $\alpha(\U)=\delta(\U)$, by Theorem~\ref{thrm:alpha-equals-delta}, for a maximum independent set $T=\{t_1,\ldots,t_k\}$ of $\U$ we must have that for every $t_i\in T$ there is a unique clique $C_{t_i}\in \{C_1,\ldots,C_k\}$ such that $t_i\in C_{t_i}$ and $t_i\notin C_j$ for all $j\neq t_i$.
	It follows that $C_{t_i} = \ngh_\U[t_i]$.
	Otherwise, there would be an edge $t_i \mathdash v$ where $t_i\in C_{t_i}$ and $v\in C_j$ for $j\neq t_i$.
	However, since $t_i \notin C_j$ for all $j\neq t_i$, this edge would not lie in any clique in $\E^\U$, contradicting the assumption that it is an edge clique cover.
	Let $\E_1^\U = \{C_1^\prime, \ldots, C_k^\prime\}$ be a second minimum edge clique cover of $\U$.
	It follows analogously, that each $C_i^\prime$ contains a unique element of $T$, so we can index $\E_1^\U$ as $\E_1^\U = \{C_{t_1}^\prime, \ldots, C_{t_k}^\prime\}$.
	However, by applying the same argument as before, we get that $C_{t_i} = \ngh_\U[t_i] = C_{t_i}^\prime$.
	Hence, $\U$ has a unique minimum edge clique cover.
\end{proof}

\begin{defn}
	\label{defn:uniqueECC}
	Given a UEC-representative $\U$, we let $\E^\U$ denote its unique minimum edge clique cover.
\end{defn}

As seen in the proof of Lemma~\ref{lem: unique minimum edge clique cover}, the minimum edge clique cover of a UEC-representative $\U$ can be identified as $\E^\U = \{\ngh_\U[m] : m\in M\}$ for any maximum independent set $M$ in $\U$.
While the family of UEC-representatives for DAGs are new, they are in bijection with classic combinatorial objects.

\begin{defn}
	\label{defn:min-cover}
	A \textit{minimal cover} of a finite set $S$ is a set $\{S_1,\ldots,S_k\}$ of nonempty subsets of $S$ such that for $T\subseteq[k]$ we have $\cup_{i\in T}S_i=S$ if and only if $T=[k]$.
\end{defn}

For example, let $S$ be the set $\{1,2\}$.
There are five covers of $S$, namely $\{\{1\},\{2\}\}$, $\{\{1,2\}\}$, $\{\{1\},\{1,2\}\}$, $\{\{2\},\{1,2\}\}$, and $\{\{1\},\{2\},\{1,2\}\}$, but only $\{\{1\},\{2\}\}$ and $\{\{1,2\}\}$ are minimal.
A formula for the number $a(n)$ of minimal covers of $[n]$ can be derived in terms of the Stirling Numbers of the Second Kind \citep[A046165]{oeis}.
These numbers also enumerate the number of nonempty UECs of DAGs on $n$ nodes.

\begin{prop}
	\label{prop:min-covers}
	There is a bijection between the minimal covers of $[n]$ and the UEC-representatives on $n$ nodes.
\end{prop}

\begin{proof}
	For a cover $\{S_1,\ldots,S_k\}$ of the set $[n]$, we construct the undirected graph $\U$ with vertex set $[n]$ such that $\{v, w\}\in E^\U$ if and only if $v,w\in S_i$ for some $i\in [k]$.
	Now let $C_i$ be the complete graph with vertex set $S_i$, for $i\in [k]$.
	Then the cliques $\{C_1,\ldots,C_k\}$ form an edge clique cover of $\U$.

	In fact, this correspondence between cover sets and cliques gives a bijection between the covers of $[n]$ and the edge clique covers of graphs on $n$ nodes.
	We will now show that minimal covers of $[n]$ correspond to edge clique covers that are minimal, i.e., no clique has each of its elements contained in two cliques, and vice versa.
	Indeed, a cover is minimal if and only if there does not exist a set $S_i\in \{S_1,\ldots,S_k\}$ such that for every $s\in S_i$ there is some $j_s\in [k]\setminus\{i\}$ such that $s\in S_{j_s}$.
	By construction, the latter happens if and only if there does not exist $C_i\in\{C_1,\ldots,C_k\}$ with the property that for every $c\in C_i$ there is some $j_c \in [k]\setminus\{i\}$ such that $c\in C_{j_c}$.
	The edge clique cover $\{C_1,\ldots,C_k\}$ defines an undirected graph $\U$ such that $\E^\U\coloneqq \{C_1,\ldots,C_k\}$ is its minimum edge clique cover, by Lemma~\ref{lem: unique minimum edge clique cover}.
	This gives a 1-1 correspondence between minimal covers of $[n]$ and undirected graphs on the node set $[n]$ with nonempty UEC.
\end{proof}

It follows from Proposition~\ref{prop:min-covers} that the number of UEC-representatives on $n$ nodes for the first few positive integers $n$ are those presented in Table~\ref{table: number of uecs}.
\begin{table}[]
	\centering
	\begin{tabular}{c|c}
		$n$ & Number of UECs \(|\mathbb{U}^n|\) \\\hline
		1   & 1                                 \\
		2   & 2                                 \\
		3   & 8                                 \\
		4   & 49                                \\
		5   & 462                               \\
		6   & 6\,424                            \\
		7   & 129\,425                          \\
		8   & 3\,731\,508                       \\
		9   & 152\,424\,420                     \\
		10  & 8\,780\,782\,707                  \\
		11  & 710\,389\,021\,036.
	\end{tabular}
	\caption{The number of UECs on $n$ nodes (or equivalently the number of UEC-representatives on $n$ nodes).}
	\label{table: number of uecs}
\end{table}

The results in this section can be used to develop search algorithms for estimating the marginal independence structure of a DAG model from data.
The algorithms we present will rely on the underlying algebraic structure of the space of UEC-representatives, the relevant aspects of which are described in the next section.

\section{Gr\"obner Bases for Moving between UECs}
\label{sec:grobner}

In Section~\ref{sec:uncond-depend}, we introduced a family of undirected graphs, called UEC-representatives, as representations of unconditional equivalence classes (UECs) of DAGs and characterized exactly which undirected graphs are in this family.
The resulting characterization (Theorem~\ref{thrm:alpha-equals-delta}) when combined with Theorem~\ref{thrm:alt-defs}~(2) implies that a UEC containing a DAG $\D$ has a maximum independent set consisting of the set of source nodes in $\D$ and maximal cliques given by the sets of descendants of these nodes in $\D$.
This provides a natural sufficient statistic for a UEC-representative: namely a maximum independent set $I$ and the number of nodes in $I$ that are adjacent to each node $j\in V\setminus I$.

In relation to the problem of learning the undirected dependence graph of a DAG model from data, we would like a set of moves that allow us to traverse the space of all UEC-representatives and select the optimal one.
This sufficient statistic, through the use of algebraic techniques, leads to such a set of moves.
First, we identify each UEC-representative with a monomial.
We then identify a Gr\"obner basis for the toric ideal arising from a monomial map that maps such monomials to the corresponding sufficient statistic of the UEC-representative described above.
The Fundamental Theorem of Markov Bases \citep{diaconis1998algebraic, petrovic2017survey} then implies that we can use this Gr\"obner basis to traverse the set of UEC-representatives with the same sufficient statistic.

In Section~\ref{sec: traversing}, we prove that a specified set of combinatorial moves allow us to move between families of UEC-representatives with different sufficient statistics.
We then prove that the combinatorial moves along with the moves obtained from the Gr\"obner basis are enough to connect the entire space of UEC-representatives.

\subsection{Monomial representation of undirected graphs with nonempty UEC}
\label{subsec: monomial representation}
For any positive integer $n$, consider the indeterminates of the form $x_{i|A}$ where $i$ can be any value in $[n]$ and $A$ is any element in the power set of $[n]\setminus i$.
We define the ring homomorphism on these sets of indeterminates as follows:
\begin{equation}
	\label{eqn:phi}
	\begin{split}
		\phi &: \mathbb{K}[x_{i|A} : {i}\in [n], A\in 2^{[n]\setminus i}] \longrightarrow \mathbb{K}[y_i,\beta_j : i,j \in [n]]; \\ \phi &: x_{i|A} \longmapsto y_i\prod_{j\in A} \beta_{j}.
	\end{split}
\end{equation}

Each indeterminate can be seen as a clique of size $i\cup A$ with $i$ as the source node.
By Theorem~\ref{thrm:alpha-equals-delta}, we know that an undirected graph is a UEC-representative if and only if its intersection number and independence number are equal.
Further, we have seen in the proof of Theorem~\ref{thrm:alpha-equals-delta} that there exists a bijection between the source nodes of a DAG and the cliques of a minimum edge clique cover of its UEC-representative.
In particular, the source nodes in the DAG form a maximum independent set, and each node in it corresponds to the clique formed by its descendants in the DAG (see also the proof of Lemma~\ref{lem: unique minimum edge clique cover}).
Using this property, we can represent each UEC-representative as a product of indeterminates as follows: \[ \U= x_{i_1|A_1} x_{i_2|A_2} \cdots x_{i_k|A_k}, \] where $k$ is the number of cliques in the minimum edge clique cover $\E^\U$ of $\U$ (i.e., by Theorem~\ref{thrm:alpha-equals-delta}, $k=\alpha(\U)=\delta(\U)$), and $i_j$ are the source nodes of each of the cliques.
It follows that the graph $\U = ([n],E)$ with monomial representation $x_{i_1|A_1}\cdots x_{i_k|A_k}$ has the sufficient statistic given by the exponent of the monomial $\phi(x_{i_1|A_1}\cdots x_{i_k|A_k})$, namely, $\left((i_1,\ldots,i_k),(\assn(j) : j\in [n])\right)$, where \(\assn(j) = \left|\left\{i\in\{i_1,\ldots, i_k\}: i - j \in E\right\}\right|\).
Observe that the monomial representation $x_{i_1|A_1}\cdots x_{i_k|A_k}$ is not necessarily unique: there can be multiple options for selecting a source node in a clique, when the UEC contains DAGs with different source nodes.

By Theorem~\ref{thrm:alpha-equals-delta}, we know that every UEC-representative has a monomial representation of the form $x_{i_1|A_1} x_{i_2|A_2} \cdots x_{i_k|A_k}$, where $i_j \notin A_{j'}$ for any $j,j' \in [n]$.
As the monomial representation is not unique, let us for now assume that a monomial representation has been chosen for each undirected UEC-representative.

\begin{ex}
	\label{ex:Same image}
	Let $\U$ and $\U'$ be the two graphs shown in Figure \ref{figure:same image}.
	If we choose the monomial representations of $\U$ and $\U'$ as $x_{1|23}x_{4|35}$ and $x_{1|35}x_{4|23}$ respectively, then computing the image under $\phi$ gives us \[ \phi(x_{1|23}x_{4|35})=\phi(x_{1|3}x_{4|235})=y_1y_4\beta_2\beta_3^2\beta_5.
	\]
\end{ex}
\begin{figure}
	\begin{tikzpicture}[scale=.8]
		\filldraw[black]
		(0,1.5) circle [radius=.04] node [above] {1}
		(1,0) circle [radius=.04] node [below] {2}
		(2,1.5) circle [radius=.04] node [above] {3}
		(3,0) circle [radius=.04] node [below] {5}
		(4,1.5) circle [radius=.04] node [above] {4}
		(2,-.5) circle [radius=0] node [below] {$\U$};
		\draw
		(0,1.5)--(4,1.5)
		(0,1.5)--(1,0)
		(1,0)--(2,1.5)
		(2,1.5)--(3,0)
		(3,0)--(4,1.5);

		\filldraw[black]
		(6,1.5) circle [radius=.04] node [above] {1}
		(7,0) circle [radius=.04] node [below] {2}
		(8,1.5) circle [radius=.04] node [above] {3}
		(9,0) circle [radius=.04] node [below] {5}
		(10,1.5) circle [radius=.04] node [above] {4}
		(8,-.5) circle [radius=0] node [below] {$\U'$};

		\draw
		(6,1.5)--(10,1.5)
		(7,0)--(10,1.5)
		(7,0)--(8,1.5)
		(6,1.5)--(9,0)
		(8,1.5)--(9,0);
	\end{tikzpicture}
	\caption{Two graphs having the same image under a specific monomial representation.}\label{figure:same image}
\end{figure}
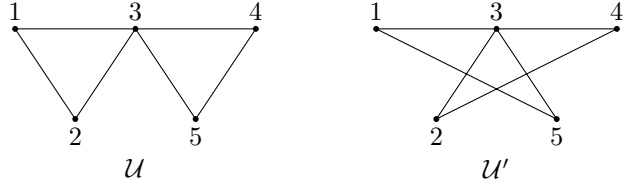

\subsection{A reduced Gr\"obner basis for the toric ideal $\ker(\phi)$.}
\label{subsec: grobner basis}
As a first step towards moves that allow us to explore the space of all UEC-representatives, we provide a (reduced) Gr\"obner basis for the toric ideal $\ker(\phi)$.
To clarify why such a set of polynomials allows us to explore the space of UEC-representatives with the same sufficient statistic, we first state some fundamental definitions.

\begin{defn}
	Let $K$ be any $d \times n$ matrix with integer entries.
	Then the \textit{toric ideal} corresponding to $K$ is defined as \[ I_K=\langle x^u-x^v : u-v \in \ker(K) \rangle.
	\]
	Here, $u, v$ are the vectors $(u_1,u_2,\ldots, u_n)$ and $(v_1,v_2,\ldots, v_n)$ in $\mathbb{Z}^n_{\geq 0}$, and $x^u, x^v$ are monomials $x_1^{u_1}x_2^{u_2}\cdots x_{n}^{u_n}$ and $x_1^{v_1}x_2^{v_2}\cdots x_{n}^{v_n}$, respectively.
\end{defn}
\begin{defn}
	\citep[Chapter 5]{sturmfels1996grobner}
	For any matrix $K=\{k_1,\ldots, k_n\}$ with columns $k_i \in \mathbb{N}^d\setminus \{0\}$ and any vector $b\in \mathbb{N}^d$, the \textit{fiber} of $K$ over $b$ denoted by $K^{-1}(b)$ is defined as the set $K^{-1}(b)=\{u\in \mathbb{N}^n : Ku=b \}$.
\end{defn}

Observe that as each entry of $K$ is a positive integer, the fiber of $K$ over any vector $b$ is a finite set.
Now, let $\mathcal{F}$ be any finite subset of $\ker(K)$.
Then the graph denoted by $K^{-1}(b)_{\mathcal{F}}$ is defined as follows: the graph's nodes are the elements in $K^{-1}(b)$, and two nodes $u$ and $u'$ are connected by an edge if $u-u' \in \mathcal{F}$ or $u'-u \in \mathcal{F}$.

\begin{thrm}
	\citep[Theorem~5.3]{sturmfels1996grobner}
	\label{thrm:connectivity}
	Let $\mathcal{F}\in ker(K)$.
	The graphs $K^{-1}(b)_{\mathcal{F}}$ are connected for all $b\in \mathbb{N}K$ if and only if the set $\{x^{v^+}-x^{v^-} : v\in \mathcal{F}\}$ generates the toric ideal $I_K$.
\end{thrm}

To apply Theorem~\ref{thrm:connectivity}, we analyze the kernel of the map $\phi$.
As $\phi$ is a monomial map, it is represented by a matrix $M$, and hence we can construct the associated toric ideal $I_M$, which we denote by $I_\phi$ to show that it arises from the map $\phi$.
The matrix corresponding to $\phi$ has $2n$ rows and $n\cdot 2^{n-1}$ columns.
Now, let $\mathcal{F}$ be the set of quadratic binomials of the form \(x_{i|A} \cdot x_{j|B} - x_{i|A \cup C}\cdot x_{j|B\setminus C}\), where $C$ is any subset of $B$ with $C\cap A = \emptyset$.
Each binomial in $\mathcal{F}$ lies in the kernel of $\phi$ since
\begin{eqnarray*}
	\phi(x_{i|A} \cdot x_{j|B})&=& y_i y_j \prod_{i'\in A}\beta_{i'} \prod_{j'\in B}\beta_{j'} \\ &=& y_i y_j \prod_{i' \in A\cup C} \beta_{i'} \prod_{j' \in B\setminus C} \beta_{j'} \\ &=& \phi(x_{i|A\cup C}\cdot x_{j|B\setminus C}).
\end{eqnarray*}

Let $\mathbb{U}^n$ be the set of UEC-representatives with $n$ vertices.
Our first aim is to find a Gr\"obner basis (i.e., a generating set) of $I_\phi$ so that we can apply Theorem~\ref{thrm:connectivity} and move between any two graphs (i.e., nodes of $K^{-1}(b)_{\mathcal{F}}$ for any fixed $b$) within the same fiber.
As we have already fixed the representation for each graph, we divide the space $\mathbb{U}^n$ into $\mathbb{U}^n_{b_1},\ldots,\mathbb{U}^n_{b_s}$ based on their fibers, i.e., two graphs lie in $\mathbb{U}^n_{b_j}$ if their image under $\phi$ is the monomial $x^{b_j}=y_1^{b_{j1}}\cdots y_n^{b_{jn}}\beta_1^{b_{j(n+1)}}\cdots \beta_n^{b_{jn}}$.
For example, using the two monomials from Example \ref{ex:Same image}, then the two respective graphs lie in the same fiber, where $b$ is the vector corresponding to $y_1y_4\beta_2\beta_3^2\beta_5$ .

In order to get a reduced Gröbner basis for $I_{\phi}$, we first consider the monomial map with extended domain given as follows:
\begin{equation}
	\begin{split}
		\label{eqnarray:phi'} \phi^\prime : \mathbb{K}[x_{i|A} : {i}\in [n], A\in 2^{[n]}] &\longrightarrow \mathbb{K}[y_i,\beta_j : i,j \in [n]] \\x_{i|A} &\longmapsto y_i\prod_{j\in A} \beta_{j}.
	\end{split}
\end{equation}
In other words, we add the variables of the form $x_{i|A}$ with $i \in A$ to the domain.
The motivation behind extending the domain is to identify a lattice structure associated to the ideal for which pre-existing results can be applied.
Before getting into the results, we first state some basic terminology and results that will be used here.
For all undefined terminology relating to Gr\"obner bases, we refer the reader to \citep{sturmfels1996grobner}.

Let $\mathbb{K}$ be a field and $L = (e_L, \prec_L)$ be a finite lattice, i.e., a partially ordered set with finite ground set $e_L$ and relation $\prec_L$ in which every pair of elements $a, b\in e_L$ have a meet $a\wedge b$ and and join $a \vee b$ with respect to $\prec_L$.
Let $\mathbb{K}[L]=\mathbb{K}[{x_a : a\in e_L}]$ denote the polynomial ring in $|e_L|$ variables over $\mathbb{K}$.
For any two variables $a,b\in e_L$, the binomial $f_{a,b} \in \mathbb{K}[L]$ is given by : \[ f_{a,b} = x_ax_b - x_{a\wedge b}x_{a\vee b}.
\]

Observe that the binomial $f_{a,b}$ is zero if $a$ and $b$ are comparable in $L$.
Now, let \[ G_L= \{f_{a,b} : a,b \text { are incomparable in } L \}, \] i.e., the set of all the binomials of the form $f_{a,b}$ where $a,b \in e_L$ are incomparable, and let $I_{L}\subsetneq \mathbb{K}[L]$ be the ideal generated by the elements of $G_L$.
A monomial order $<$ on $\mathbb{K}[L]$ is called \textit{compatible} if, for all $a,b\in e_L$ for which $a$ and $b$ are incomparable, we have that the initial term $\mathrm{in}_<(f_{a,b})$ of $f_{a,b}$ is $x_ax_b$.

Now, let $P=\{p_1,\ldots,p_n\}$ be a finite poset and $L = J(P)$ be the finite distributive lattice which consists of all poset ideals of $P$, ordered by inclusion.
Let $S=\mathbb{K}[x_1,\ldots, x_n,t]$ be the polynomial ring in $n+1$ variables over the field $\mathbb{K}$.
We define a surjective ring homomorphism $\pi$ as
\begin{eqnarray*}
	\pi : \mathbb{K}[L] &\rightarrow& S \\ \pi(x_\alpha) &\mapsto& (\prod_{i\in \alpha} x_i) t,
\end{eqnarray*}
for any $\alpha \in L$, where $\alpha$ is a poset ideal of $P$ (and hence can be viewed as a subset of the elements of $e_L$).
We then have the following result of \citet{binomial-ideals}:

\begin{thrm}[{\citep[Combining Theorems~6.17,~6.19, and Corollary~6.20]{binomial-ideals}}] \label{thrm:Grobner basis of lattice} Let $L$ be a finite lattice with a compatible monomial order $<$ on $\mathbb{K}[L]$.
	If $L$ is distributive, then $I_L$ is equal to $\ker(\pi)$ and $G_L$ forms a Gröbner basis for $\ker(\pi)$.
\end{thrm}

To fit this result into our framework, we take $J(P)$ to be the collection of all indeterminates of the form $x_{v|A}$, with $v\in [n]$ and $A\in 2^{[n]}$.
As $t$ is a homogenizing factor, it does not change the kernel of $\pi$, so we remove it from the map, giving us
\begin{eqnarray*}
	\mathbb{K}[L]&=&\mathbb{K}[x_{i|A} : i\in [n], A\in 2^{[n]}] \text{ and} \\ S&=& \mathbb{K}[y_i,\beta_j : i,j \in [n]].
\end{eqnarray*}
Here $x_\alpha$ corresponds to $x_{i|A}$ and $\prod_{i\in \alpha}x_i$ is equal to $y_i\prod_{j\in A}\beta_j$.
The description of $L$ as a distributive lattice is as follows: consider the chain lattice $C_n$, with ground set $[n]$ and the natural ordering, and the Boolean lattice $B_n$ with ground set the subsets of $[n]$ ordered by inclusion.
Now, we define the distributive lattice $C_n\times B_n$, the \emph{direct product} of $C_n$ and $B_n$ (see \citep{stanley2011enumerative} for a definition), and let $L$ be the lattice of the indeterminates $x_{i|A}$, where $x_{i|A}<x_{j|B}$ in $L$ whenever $(i,A)<(j,B)$ in $C_n\times B_n$, for $i,j\in[n], A,B\in 2^{[n]}$.
The lattices $L$ and $C_n\times B_n$ are (canonically) isomorphic.
In particular, we have that
\begin{eqnarray*}
	x_{i|A}\wedge x_{j|B}&=& x_{j|A\cap B}, \quad \text{and} \\ x_{i|A}\vee x_{j|B}&=& x_{i|A\cup B}
\end{eqnarray*}
when $i<j$.
Now that we have a distributive lattice, we need a compatible monomial order for $\mathbb{K}[L]$.
In order to construct such a monomial order, we first extend the ordering on $L$ to a total order, such that: \[ x_{i|A} < x_{j|B} \text{ if either } \]
\begin{enumerate}
	\item $i<j$, or
	\item $i=j$ and $|A| < |B|$, or
	\item $i=j$, $|A|=|B|$ and $\text{min}\{a\mid a\in A\setminus B\}<\text{min}\{b\mid b\in B\setminus A\}$.
\end{enumerate}
Using the above total order on the variables $x_{i|A}$, we can write each monomial in $\mathbb{K}[L]$ as $x^u$, where $u$ is a vector of size $n2^n$.
This allows us to define a graded reverse lexicographic ordering for $\mathbb{K}[L]$ as follows: \[ x^u < x^v \text{ if either} \]
\begin{enumerate}
	\item $\sum_{i=1}^{n2^n}u_i < \sum_{i=1}^{n2^n}v_i$ or
	\item if $\sum_{i=1}^{n2^n}u_i = \sum_{i=1}^{n2^n}v_i$ and the rightmost nonzero entry of $({v_1-u_1},\ldots, {v_{n2^n}-u_{n2^n}})$ is negative.
\end{enumerate}

\begin{lemma}
	The graded reverse lexicographic order on $\mathbb{K}[L]$ defined above is compatible with the partial ordering of $L$.
\end{lemma}
\begin{proof}
	Let $a=x_{i|A}$ and $b=x_{j|B}$ be two incomparable elements in $e_L$ and consider the binomial $f_{a,b}=x^u-x^v$ where $x^u=x_{i|A}x_{j|B}$ and $x^v=x_{i|A\cap B}x_{j|A\cup B}$.
	We then have two cases: In the first case, assume that $i\neq j$.
	Let us also assume that $i<j$.
	As $x_{i|A}$ and $x_{j|B}$ are incomparable, we know that $A\not\subseteq B$.
	Now, comparing the two monomials in $f_{a,b}$, we see that $x_{j|A\cup B}$ corresponds to the rightmost nonzero entry between $u$ and $v$.
	Thus, the rightmost nonzero entry of $(u_1-v_1,\ldots, u_{n2^n}-v_{n2^n})$ is negative, giving us that $\mathrm{in}_{<_{\mathrm{grvlex}}(f_{a,b})}=x_ax_b$.

	In the second case, suppose that $i=j$.
	In this case, incomparability implies that neither $A \subseteq B$ nor $B\subseteq A$.
	Thus, $x_{i|A\cup B}$ corresponds to the rightmost nonzero entry between $u$ and $v$, giving us $\mathrm{in}_{<_{\mathrm{grvlex}}(f_{a,b})}=x_ax_b$.
\end{proof}

\begin{ex}
	For $n=3$, we assign the numbers from $1$ to $24$ to the indeterminates as
	\begin{eqnarray*}
		&&x_{1|\{\emptyset\}} < x_{1|\{1\}} < x_{1|\{2\}} < x_{1|\{3\}} < x_{1|\{1,2\}} < x_{1|\{1,3\}} < x_{1|\{2,3\}} < x_{1|\{1,2,3\}} < \\ &&x_{2|\{\emptyset\}} < x_{2|\{1\}} < x_{2|\{2\}} < x_{2|\{3\}} < x_{2|\{1,2\}} < x_{2|\{1,3\}} < x_{2|\{2,3\}} < x_{2|\{1,2,3\}} < \\ &&x_{3|\{\emptyset\}} < x_{3|\{1\}} < x_{3|\{2\}} < x_{3|\{3\}} < x_{3|\{1,2\}} < x_{3|\{1,3\}} < x_{3|\{2,3\}} < x_{3|\{1,2,3\}}.
	\end{eqnarray*}
	The indeterminates $a=x_{1|\{2,3\}}$ and $b=x_{2|\{1,2\}}$, are incomparable in $L$.
	So, computing $f_{a,b}$ gives us \(f_{a,b}=x^u-x^v=x_{1|\{2,3\}}x_{2|\{1,2\}}-x_{1|\{2\}}x_{2|\{1,2,3\}}\), where $u=e_7+e_{13}$ and $v=e_3+e_{16}$.
	Now, \(u-v=e_7+e_{13}-e_3-e_{16}\) has the rightmost nonzero component $-1$, indeed giving us \(\mathrm{in}_{<_{\mathrm{grvlex}}(f_{a,b})}=x_{1|\{2,3\}}x_{2|\{1,2\}}\).

\end{ex}
This brings us to the following result.
\begin{thrm}
	\label{thrm:reduced GB}
	Let $\phi^\prime$ be the monomial map as defined in~\eqref{eqnarray:phi'} and $L$ be the corresponding distributive lattice.
	Then the set
	\begin{eqnarray*}
		G_L=&\{ x_{i|A}x_{j|B}-x_{\min\{i,j\}|A \cap B}x_{\max\{i,j\}|A\cup B}: \\ &x_{i|A} \text{ and } x_{j|B} \text{ are incomparable in }L \}
	\end{eqnarray*}
	forms a reduced Gröbner basis for $I_{\phi^\prime}$.
\end{thrm}
\begin{proof}
	The fact that $G_L$ forms a Gröbner basis follows from Theorem~\ref{thrm:Grobner basis of lattice}.
	In order to prove that $G_L$ is reduced, we need to show that for any two arbitrary distinct elements $g$ and $g'$ in $G_L$ (with $g\neq g'$), none of the terms of $g$ are divisible by the initial term of $g'$.
	If
	\begin{eqnarray*}
		&g&=g_1-g_2= x_{i_1|A_1}x_{i_2|A_2}-x_{i_1|A_1\cap A_2}x_{i_2|A_2\cup A_2} \text{ and} \\ &g'&=g'_1-g'_2=x_{j_1|B_1}x_{j_2|B_2}-x_{j_1|B_1\cap B_2}x_{j_2|B_2\cup B_2}
	\end{eqnarray*}
	with $\mathrm{in}_<(g')$ dividing either of the two terms of $g$, then either $g'_1=g_1$ or $g'_1=g_2$.
	If $g'_1=g_1$, then we necessarily have $g_2=g'_2$ and hence $g=g'$, which is a contradiction.
	Further, if $g'_1=g_2$, then we have \(x_{j_1|B_1}x_{j_2|B_2}=x_{i_1|A_1\cap A_2}x_{i_2|A_1\cup A_2}\).
	This gives us that \(j_1=i_1, j_2=i_2, B_1=A_1\cap A_2 \text{ and } B_2=A_1\cup A_2\).
	Computing $B_1\cap B_2$ and $B_1\cup B_2$ in terms of $A_1$ and $A_2$ gives us
	\begin{eqnarray*}
		B_1 \cap B_2&=& (A_1\cap A_2)\cap (A_1\cup A_2)=A_1\cap A_2=B_1, \\ B_1 \cup B_2&=& (A_1\cap A_2)\cup (A_1\cup A_2)=A_1\cup A_2=B_2,
	\end{eqnarray*}
	which contradicts the incomparability of $x_{j_1|B}$ and $x_{j_2|B}$.
	We conclude that $G_L$ is a reduced Gröbner basis for $I_{\phi^\prime}$.
\end{proof}

Given this Gr\"obner basis for $I_{\phi^\prime}$, i.e., in the extended domain, we identify a generating set for $I_{\phi}$ in the polynomial ring without the extra variables.
We know that $I_{\phi}$ is equal to $I_{\phi'}\cap \mathbb{K}[x_{i|A}:i\notin A]$ as $\phi$ and $\phi'$ are identical on $\mathbb{K}[x_{i|A}:i\notin A]$.
However, the monomial order that we defined on the extended domain is not an elimination order for the $x_{i|A}$ with $i\in A$ since we have binomials in the reduced Gr\"obner basis of the form $x_{i|A}x_{j|B}-x_{i|A\cap B}x_{j|A\cup B}$, where $i\notin A, j\notin B$ but $j\in A$.
Hence, the initial term $x_{i|A}x_{j|B}$ will appear in the ring generated by the restricted subset of indeterminates, but not the entire binomial.
Thus, we cannot take the intersection of the reduced Gr\"obner basis for $I_{\phi^\prime}$ with $\mathbb{K}[x_{i|A}:i\notin A]$ to get a Gr\"obner basis for the ideal $I_\phi$.
To handle this, we construct a new generating set for $I_{\phi^\prime}\cap \mathbb{K}[x_{i|A}:i\notin A]$ by using the reduced Gr\"obner basis on the extended domain.

\begin{thrm}
	\label{thrm:generating set}
	Let $\phi^\prime$ be the monomial map as defined above and $L$ be the corresponding distributive lattice.
	Then the set
	\begin{eqnarray*}
		&T=\{x_{i|A}x_{j|B}-x_{i|A\cap B}x_{j|A\cup B} \in G_L \text{ with } i\notin A, j\notin A\cup B \} \cup \\ &\{ x_{i|A}x_{j|B}-x_{i|(A\cap B)\cup j}x_{j|(A\setminus j)\cup B}: x_{i|A} \text{ and } x_{j|B} \text{ are incomparable} \\ &\text{in }L \text{ with } i\notin A, j\notin B, j\in A \}
	\end{eqnarray*}
	is a generating set for $I_\phi = I_{\phi^\prime}\cap \mathbb{K}[x_{i|A}:i\notin A]$.
\end{thrm}
\begin{proof}
	We first show that $T$ is a generating set for $I_{\phi^\prime}\cap \mathbb{K}[x_{i|A}:i\notin A]$.
	As the matrix corresponding to the map $\phi^\prime$ on $\mathbb{K}[x_{i|A}:i\notin A]$ (i.e., the map $\phi$ defined \eqref{eqn:phi}) has the vector $e_1+e_2+\cdots +e_{n2^n}$ in its row space, $I_{\phi^\prime}\cap \mathbb{K}[x_{i|A}:i\notin A]$ is also a toric ideal generated by homogeneous binomials.
	This is because the existence of $e_1+ \cdots + e_{n2^n}$ in the row space implies that the coordinate sum of all vectors $u\in \ker(M)$ is zero \citep[Lemma~4.14]{sturmfels1996grobner}.
	Let $g=g_1-g_2$ be any arbitrary binomial in $I_{\phi^\prime}\cap \mathbb{K}[x_{i|A}:i\notin A]$.
	As $g$ is also present in $I_{\phi^\prime}$, by Theorem~\ref{thrm:connectivity}, we can move from $g_1$ to $g_2$ by applying the moves (binomials) in $G_L$.
	But in order to prove that $T$ is a generating set, we need to show that we can move from $g_1$ to $g_2$ by applying the moves in $T$.

	We first classify the binomials in $G_L$ in terms of the extra variables.
	We have four possible cases: $x_{i|A}x_{j|B}-x_{i|A\cap B}x_{j|A\cup B}$ with some variable $x_{d|D}, d\in D$, dividing
	\begin{enumerate}
		\item both $x_{i|A}x_{j|B}$ and $x_{i|A\cap B}x_{j|A\cup B}$,
		\item $x_{i|A}x_{j|B}$ but not $x_{i|A\cap B}x_{j|A\cup B}$,
		\item $x_{i|A\cap B}x_{j|A\cup B}$ but not $x_{i|A}x_{j|B}$ or
		\item neither $x_{i|A}x_{j|B}$ nor $x_{i|A\cap B}x_{j|A\cup B}$.
	\end{enumerate}
	It is clear that the moves of type $1$ are not used while moving from $g_1$ to $g_2$ in $I_{\phi^\prime}$.
	Observe that we do not have any binomial of type $1$ in the set $T$.
	Similarly, binomials of type $4$ could be used while moving, and hence are included in $T$.

	Now, any binomial in type $2$ needs to have $i\in A$ and $j \notin A\cup B$.
	As the monomial $x_{i|A}x_{j|B}$ of the type $2$ binomials is divisible by some variable of the form $x_{d|D}, d\in D$, the only reason this binomial could be used is if we wanted to have the factor $x_{i|A\cap B}x_{j|A\cup B}$ in some intermediate step.
	But this could be achieved by applying a type $4$ binomial of the form $x_{i|A\setminus \{i\}}x_{j|B\cup \{i\}}-x_{i|A\cap B}x_{j|A\cup B}$.

	Lastly, any binomial in type $3$ has $i\notin A, j\notin B$ with $j\in A$.
	Now, binomials of this form will be used if
	\begin{enumerate}
		\item we either want to move from a factor $x_{i|A}x_{j|B}$ to another factor $x_{i|A'}x_{j|B'}$, where both $x_{i|A}x_{j|B}$ and $x_{i|A'}x_{j|B'}$ are initial terms of type $3$ binomials with the same non-initial term $x_{i|A\cap B}x_{j|A\cup B}$, or
		\item if we want to apply a sequence of moves to reach an intermediate monomial having a factor of the form $x_{j|A'\cup B'}$ with $j\in A'$ (i.e., the intermediate monomial does not lie in $\mathbb{K}[x_{i|A}: i \in A]$), before we move to a monomial with the factor $x_{j|A'\cup B' \cap C}$ with $j \notin C$.
	\end{enumerate}
	For the first case, as $G_L$ is a reduced Gr\"obner basis, the initial terms of each binomial are unique, and hence we cannot apply a similar technique as for type $2$ binomials.
	To mitigate this issue, we introduce a new set of binomials of the form
	\begin{eqnarray*}
		&\{ x_{i|A}x_{j|B}-x_{i|(A\cap B)\cup j}x_{j|(A\setminus j)\cup B}: x_{i|A} \text{ and } x_{j|B} \text{ are incomparable in } L \\ &\text{ with } i\notin A, j\notin B, j\in A \}.
	\end{eqnarray*}
	Observe that the non-initial terms of this set are initial terms in $G_L$ as we have the binomials of the form $x_{i|(A\cap B)\cup j}x_{j|(A\setminus j)\cup B}-x_{i|A\cap B}x_{j|A\cup B}$ in $G_L$.
	Thus, these binomials could be used to move from one initial term of $G_L$ to some other without going through the intermediate factor of $x_{i|A\cap B}x_{j|A\cup B}$ with $j\in A$.

	For the second case, we can replace the first move $x_{i|A}x_{j|B}-x_{i|A\cap B}x_{j|A\cup B}$ (which takes us to the monomial having the factor $x_{j|A\cup B}$) with $x_{i|A}x_{j|B}-x_{i|(A\cap B)\cup j}x_{j|(A\setminus j)\cup B}$, so that we have the factor $x_{j|(A\setminus j)\cup B}$, keeping the monomial inside $\mathbb{K}[x_{i|A}:i\notin A]$.
	From here, we can apply the same sequence of moves to reach $x_{j|A'\cup B' \cap C}$, implying that the intermediate monomials always lie inside $\mathbb{K}[x_{i|A}:i\notin A]$.

	So, we can replace the binomials in $G_L$ with the binomials in $T$ to reach from $g_1$ to $g_2$ in $\mathbb{K}[x_{i|A}:i\notin A]$.
	As $g$ was an arbitrary binomial in $I_{\phi^\prime}\cap \mathbb{K}[x_{i|A}:i\notin A]$, by Theorem~\ref{thrm:connectivity} we can conclude that $T$ forms a generating set for $I_{\phi^\prime}\cap \mathbb{K}[x_{i|A}:i\notin A]$.
\end{proof}

\begin{thrm}
	\label{thrm:rgb for less variables}
	Let $T$ be the set as defined in Theorem~\ref{thrm:generating set}.
	Then $T$ is a reduced Gr\"obner basis for $I_\phi = I_{\phi^\prime}\cap \mathbb{K}[x_{i|A}:i\notin A]$.
\end{thrm}
\begin{proof}
	As $T$ is a union of two sets of binomials, we refer to the binomials in the first set as $f_i$ and that in the second set as $g_i$.
	To prove that $T$ is a Gr\"obner basis, we first look at the possible $S$-polynomials of any pair of binomials in $T$.
	For any two polynomials $f$ and $g$, the $S$-polynomial $S(f,g)$ is defined as \[ S(f,g)=(LCM(im(f),im(g))/in(f))f-(LCM(im(f),im(g))/in(g))g, \] where $im(f)$ and $in(f)$ are the initial monomial and initial term of $f$ respectively.
	By Buchberger's criterion, the set $T$ is a Gr\"obner basis if each $S$-polynomial gives remainder zero on application of the division algorithm \citep[Theorem~2.28]{hassett}.
	So, we have the following cases to consider:
	\begin{enumerate}
		\item $S(f_i,f_j)$: Observe that both $f_i$ and $f_j$ are also present in $G_L$, which is a Gr\"obner basis.
		      So, applying the division algorithm gives us remainder zero in $G_L$.
		      Further, $S(f_i,f_j)$ does not have any variables of the form $x_{i|A}$ with $i\in A$.
		      Thus, the division algorithm only uses the binomials of the form $f_i$ in $G_L$, and hence would have remainder zero in $T$ as well.
		\item $S(f_i,g_j)$ or $S(g_i,g_j)$ where the initial terms do not have a common factor:
		      Let $f_1=f_{11}f_{12}-f_{21}f_{22}$ and $g_1=g_{11}g_{12}-g_{21}g_{22}$, where $f_{ij}$ and $g_{ij}$ are degree one monomials with $f_{11},f_{12}\neq g_{11},g_{12}$.
		      Computing the $S$-polynomial gives us
		      \begin{eqnarray*}
			      S(f_1,g_1)&=&-g_{11}g_{12}f_{21}f_{22}+f_{11}f_{12}g_{21}g_{22} \\ &=&g_{11}g_{12}f_1-f_{11}f_{12}g_1,
		      \end{eqnarray*}
		      which gives the remainder zero in the second step of the division algorithm.
		\item $S(f_i,g_j)$ or $S(g_i,g_j)$ where the initial terms have exactly one common factor:
		      In this case, we have $f_1=f_{11}f_{12}-f_{21}f_{22}$ and $g_1=f_{11}g_{12}-g_{21}g_{22}$ with $f_{12}\neq g_{12}$.
		      Computing the $S$-polynomial gives us
		      \begin{eqnarray*}
			      S(f_1,g_1)=-g_{12}f_{21}f_{22}+f_{12}g_{21}g_{22},
		      \end{eqnarray*}
		      which lies in $I_{\phi^\prime}\cap \mathbb{K}[x_{i|A}:i\notin A]$.
		      As $T$ is a generating set, by Theorem~\ref{thrm:connectivity} we can reach from $-g_{12}f_{21}f_{22}$ to $f_{12}g_{21}g_{22}$ by applying a set of binomials in $T$.
		      Thus, we show that we can apply the same set of binomials in the division algorithm and get zero as the remainder.
		      As we can use the binomials of $T$ in either direction while moving from $-g_{12}f_{21}f_{22}$ to $f_{12}g_{21}g_{22}$, we need to show that we indeed apply them in the forward direction (i.e., using the initial terms) so that they could be applied in the division algorithm as well.
		      First, we have that any degree $3$ binomial in $I_\phi\cap \mathbb{K}[x_{i|A}:i\notin A]$ has at least one term which is divisible by the product of an incomparable pair of variables.
		      This is because, if we have a degree $3$ binomial $\overline{f}=\overline{f_1}-\overline{f_2}=x_{i|A}x_{j|B}x_{k|C}-x_{i|A_1}x_{j|B_1}x_{k|C_1}$ where $i\leq j \leq k$ and $A\subseteq B \subseteq C$, then $A_1,B_1,$ and $C_1$ cannot satisfy the condition $A_1\subseteq B_1 \subseteq C_1$ unless the binomial is identically zero.
		      In order to prove this, let us assume that $a\in A\setminus A_1$ with $A_1\subseteq B_1 \subseteq C_1$.
		      As $\overline{f}$ lies in $I_\phi\cap \mathbb{K}[x_{i|A}:i\notin A]$, we know that $A\cup B \cup C=A_1\cup B_1\cup C_1$ as multisets.
		      This implies that three copies of $a$ lie in $A\cup B \cup C$ but at most two copies of $a$ can lie in $A_1\cup B_1\cup C_1$, which is a contradiction.
		      We use the same argument by reversing the roles of $A$ and $A_1$ to conclude that $A=A_1, B=B_1$ and $C=C_1$.

		      Secondly, the initial term of such a degree $3$ binomial $\overline{f}$ is always $\overline{f_2}$ as $x_{k|C}$ corresponds to the rightmost component in the ordering among the variables in $\overline{f_1}$ and $\overline{f_2}$.
		      Furthermore, there cannot exist any binomial of the form \[ \overline{g}=\overline{g_1}-\overline{g_2}=x_{i|A}x_{j|B}x_{k|C}-x_{i|A_1}x_{j|(B_1\cap C_1)\cup k}x_{k|(B_1\setminus k)\cup C_1} \] where $x_{j|(B_1\cap C_1)\cup k}x_{k|(B_1\setminus k)\cup C_1}$ is the only incomparable pair.
		      This is because $k$ lies in $(B_1\cap C_1)\cup k$.
		      So, if $A\subseteq B \subseteq C$, then $k$ has to lie in $C$ as well, which is a contradiction.
		      Finally, observe that if we have a binomial of the form $\overline{g}$ where $x_{j|(B_1\cap C_1)\cup k}x_{k|(B_1\setminus k)\cup C_1}$ is incomparable, then either $x_{i|A_1}$ is incomparable with $x_{j|(B_1\cap C_1)\cup k}$ (or $x_{k|(B_1\setminus k)\cup C_1}$) or $\overline{g_1}$ is the initial term of $g$.
		      This follows from the fact that if $x_{i|A_1}$ is comparable with either of the two terms, then $x_{k|(B_1\setminus k)\cup C_1}$ corresponds to the rightmost component among all the terms in $\overline{g_1}$ and $\overline{g_2}$ (as $A_1\setminus k$ has to be a subset of $(B_1\setminus k)\cup C_1$).
		      Thus, we can apply the same set of binomials used for moving between the two terms of $S(f_1,g_1)$ in the division algorithm to get zero as a remainder.
		\item $S(f_i,g_j)$ or $S(g_i,g_j)$ where the initial terms are equal: In this case, we show that such pairs do not exist in $T$.
		      We know that for any $g_i$, the initial term is of the form $x_{i|A}x_{j|B}$ with $j\in A$.
		      If this is also the initial term of some $f_j$, then the non initial term of $f_j$ would be $x_{i|A\cap B}x_{j|A\cup B}$ where $j\in A\cup B$, which is a contradiction.
		      Similarly, if some $g_i$ and $g_j$ have the same initial term, then $g_i$ and $g_j$ have to equal as they would also have the same non initial term.
	\end{enumerate}
	Thus, we can conclude that $T$ is a Gr\"obner basis.
	In order to show that $T$ is reduced, observe that the set of all initial terms of $T$ is a subset of the initial terms of $G_L$.
	Furthermore, the union of the set of initial and non initial terms of $T$ is also a subset of the corresponding union of $G_L$.
	But as $G_L$ is a reduced Gr\"obner basis, it follows that $T$ is reduced as well.
\end{proof}

In Theorem~\ref{thrm:rgb for less variables}, we proved that $T$ forms a reduced Gr\"obner basis for $I_{\phi}$.
However, $T$ is a subset of $\mathcal{F}$ as for any incomparable pair $x_{i|A}$ and $x_{j|B}$ with $i<j$, we have
\begin{eqnarray*}
	x_{i|A}x_{j|B}-x_{i|A\cap B}x_{j|A\cup B}&=& x_{j|B}x_{i|A}-x_{j|A\cup B}x_{i|A\cap B} \\ &=& x_{j|B}x_{i|A}-x_{j|B\cup (A\setminus B)}x_{i|A\setminus (A\setminus B)},
\end{eqnarray*}
and for the second set of binomials
\begin{eqnarray*}
	x_{i|A}x_{j|B}-x_{i|(A\cap B)\cup j}x_{j|(A\setminus j)\cup B}&=& x_{j|B}x_{i|A}-x_{j|(A\setminus j)\cup B}x_{i|(A\cap B)\cup j} \\ &=& x_{j|B}x_{i|A}-x_{j|B\cup (A\setminus j)}x_{i|(A\setminus A\setminus (B\cup j)},
\end{eqnarray*}
where $j\in B\setminus A$.
Thus, $\mathcal{F}$ is generates $I_{\phi}$ and Theorem~\ref{thrm:connectivity} lets us move within the fiber $\mathbb{U}^n_{b}$ of UEC-representatives on vertex set $[n]$ with sufficient statistic $b$.

\begin{coro}
	\label{coro: completeness of within fiber moves}
	Let $\mathbb{U}^n_{b}$ be the set of UEC-representatives with vertices $[n]$ which have the same image $x^b$ under the map $\phi$.
	Then we can move between any two graphs in $\mathbb{U}^n_b$ by using the binomials in $\mathcal{F}$.
\end{coro}
\begin{proof}
	The result follows from Theorem~\ref{thrm:connectivity} and the fact that $\mathcal{F}$ is a generating set of $I_\phi = I_{\phi^\prime}\cap \mathbb{K}[x_{i|A}:i\notin A]$ (Theorem~\ref{thrm:generating set}).
\end{proof}

Below is an example illustrating how we can use Corollary~\ref{coro: completeness of within fiber moves} to move between two graphs in the same fiber.

\begin{ex}
	Let $\U=x_{1|2}x_{3|24}x_{5|4}$ and $\U'=x_{1|4}x_{3|2}x_{5|24}$ be the two graphs seen in Figure \ref{figure:moves}.
	In order to move from $\U$ to $\U'$, we first ``apply'' the binomial $f_1= x_{1|2}x_{5|4}-x_{1|\emptyset}x_{5|24}$ to $\U$.
	By ``apply'' the binomial, we mean that we use the fact that the binomial $f_1$ has a term dividing the monomial representing $\U$, i.e., $x_{1|2}x_{5|4}$, to replace this portion of the monomial for $\U$ with the other term in the binomial $f_1$, i.e., $x_{1|\emptyset}x_{5 |24}$.
	This transforms the monomial representing $\U$ into a new monomial $x_{1|\emptyset}x_{3|24}x_{5|24}$ which represents a different graph, namely, the intermediate graph in Figure~\ref{figure:moves}.
	In this way, we can apply binomials to monomial representations to move between graphs representing different nonempty UECs.
	As one of the monomials of $f_1$ divides the representation of $\U$, i.e., $x_{1|2}x_{5|4}$, applying $f_1$ on $\U$ means that we can replace $x_{1|2}x_{5|4}$ with $x_{1|\emptyset}x_{5|24}$.
	We can then apply the binomial $f_2=x_{1|\emptyset}x_{3|24}-x_{1|4}x_{3|2}$ on the intermediate graph with representation $x_{1|\emptyset}x_{3|24}x_{5|24}$ to reach $\U'$.
	At the monomial level, these moves can be seen as follows: \[ x_{1|2}x_{3|24}x_{5|4}\rightarrow x_{1|\emptyset}x_{3|24}x_{5|24}\rightarrow x_{1|4}x_{3|2}x_{5|24}.
	\]

	\begin{figure}
		\begin{tikzpicture}[scale=.7]
			\filldraw[black]
			(-2,1.5) circle [radius=.04] node [above] {1}
			(-1,0) circle [radius=.04] node [below] {2}
			(0,1.5) circle [radius=.04] node [above] {3}
			(1,0) circle [radius=.04] node [below] {4}
			(2,1.5) circle [radius=.04] node [above] {5}
			(0,-.5) circle [radius=0] node [below] {$\U$};
			\draw
			(-2,1.5)--(-1,0)
			(-1,0)--(1,0)
			(-1,0)--(0,1.5)
			(0,1.5)--(1,0)
			(1,0)--(2,1.5);

			\filldraw[black]
			(3.5,1.5) circle [radius=.04] node [above] {1}
			(4.5,0) circle [radius=.04] node [below] {2}
			(5.5,1.5) circle [radius=.04] node [above] {3}
			(6.5,0) circle [radius=.04] node [below] {4}
			(7.5,1.5) circle [radius=.04] node [above] {5}
			(5.5,-.5) circle [radius=0] node [below] {Intermediate graph};

			\draw
			(4.5,0)--(7.5,1.5)
			(4.5,0)--(6.5,0)
			(4.5,0)--(5.5,1.5)
			(5.5,1.5)--(6.5,0)
			(6.5,0)--(7.5,1.5);

			\filldraw[black]
			(9,1.5) circle [radius=.04] node [above] {1}
			(10,0) circle [radius=.04] node [below] {2}
			(11,1.5) circle [radius=.04] node [above] {3}
			(12,0) circle [radius=.04] node [below] {4}
			(13,1.5) circle [radius=.04] node [above] {5}
			(11,-.5) circle [radius=0] node [below] {$\U'$};

			\draw
			(9,1.5)--(12,0)
			(10,0)--(11,1.5)
			(10,0)--(12,0)
			(10,0)--(13,1.5)
			(12,0)--(13,1.5);
		\end{tikzpicture}
		\caption{Moving within fiber from $\U$ to $\U'$}\label{figure:moves}
	\end{figure}
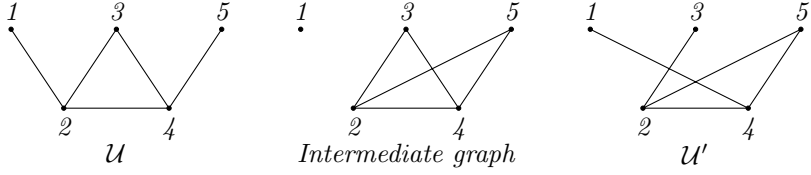
\end{ex}

\section{Traversing the Space of UECs}
\label{sec: traversing}

The set of moves given by the binomials $\mathcal{F}$ in Corollary~\ref{coro: completeness of within fiber moves} allows us to move between monomial representations of UEC-representatives whenever these monomials lie in the same fiber.
We see also from these binomials that such moves hold certain constant features of the monomial representation of the graph; for instance, the specified source nodes $i_1,\ldots, i_k$ as well as the intersection and union of any two of the sets $A_{i_j}$ and $A_{i_l}$.
To explore the entire space of UEC-representatives, we need to relax these restrictions on our moves.
We provide the necessary relaxation by introducing new moves in two steps: In subsection~\ref{subsec: generalized fibers} we introduce two moves that allow us to move between UEC-representatives with the same intersection number (and hence independence number according to Theorem~\ref{thrm:alpha-equals-delta}).
Then, in subsection~\ref{subsec: combinatorial}, we introduce two additional moves that allow us to move between UEC-representatives with different intersection numbers.
Along the way, we note how each of these four additional moves can also be represented via a (not necessarily homogeneous) binomial.
Finally, we prove in subsection~\ref{subsec: connecting} that, collectively, these moves allow us to connect the space of all UEC-representatives.

\subsection{Generalized fibers and out-of-fiber moves}
\label{subsec: generalized fibers}
We now introduce a framework and associated set of moves that allows us to expand the connectivity of our fibers to the union over all fibers containing monomial representations of UEC-representatives having the same intersection number.

As seen in the previous subsection, we can move between any two graphs lying in the same fiber by using the binomials in $\mathcal{F}$.
However, the structure of any given fiber is strictly based on the representation that we choose for each graph.
For example, if we consider the graph $\U=(V,E)= ([4],\{\{1,2\},\{1,3\},\{2,3\}\})$ and choose its monomial representation as $x_{1|\{2,3\}}x_{4|\{\emptyset\}}$, then we cannot move to the graph $\U'=([4],\{\{1,3\},\{1,4\},\{3,4\}\})$ (with representation $x_{2|\emptyset}x_{4|\{1,3\}}$) by using any binomial from $\mathcal{F}$ because the two monomials lie in different fibers.
On the other hand, if we use the representation $x_{2|\{1,3\}}x_{4|\emptyset}$ for $\U$, then we can use the binomial $x_{2|\{1,3\}}x_{4|\emptyset}-x_{2|\emptyset}x_{4|\{1,3\}}$ in $\mathcal{F}$ to move from $\U$ to $\U'$ since the two monomials now lie in the same fiber.
In order to avoid this restriction caused by fixing a graph representation, we define the notion of \textit{generalized fiber} as follows:

\begin{defn}
	\label{def:generalized fiber} For any UEC-representative $\U$, the \textit{generalized fiber} of $\U$ is the collection of all UEC-representatives $\U'$ for which there exists a monomial representation of $\U$ and a monomial representation of $\U^\prime$ such that the images of the two monomials are equal under $\phi$.
\end{defn}

The advantage of having this definition is that now we can put the graphs $\U$ and $\U'$ (mentioned above) in the same generalized fiber even though there exist representations of $\U$ and $\U'$ which put them in different fibers.
Further, we can also extend the connectivity of fibers to include connectivity of generalized fibers using the same arguments.
Thus, we have the following Corollary:

\begin{coro}
	For any UEC-representative $\U$, we can move between any two graphs in the generalized fiber of $\U$ by using the binomials in $\mathcal{F}$.
\end{coro}

Observe that the generalized fibers partition the graphs in $\mathbb{U}^n$.
However, the intersection number of each graph in a generalized fiber is the same.
Thus, we formed a coarser partition of the graphs in $\mathbb{U}^n$ in terms of their intersection number:
\begin{eqnarray*}
	\mathbb{U}^n&=&\mathbb{U}^n_1\cup \mathbb{U}^n_2 \cup \cdots \cup \mathbb{U}^n_n, \text{ and} \\ \mathbb{U}^n_i&=& \mathbb{U}^n_{{\U_{i_1}}}\cup \mathbb{U}^n_{\U_{i_2}} \cup \cdots \mathbb{U}^n_{\U_{i_s}},
\end{eqnarray*}
where each $\mathbb{U}^n_i$ is a collection of graphs in $\mathbb{U}^n$ with intersection number $i$ and $\mathbb{U}^n_{\U_{i_j}}$ is a generalized fiber containing the graph $\U_{i_j}$.
Our goal for the remainder of this subsection is to construct a set of moves that allows us to explore each $\mathbb{U}^n_i$ by moving between different generalized fibers $\mathbb{U}^n_{\U_{i_j}}$ in $\mathbb{U}^n_i$.

In order to move between graphs lying in different generalized fibers, we define two new moves which we call \textit{out-of-fiber} moves.
These moves allow us to move between any two graphs that lie in different fibers but have the same intersection number.

\begin{defn}
	\textit Let $x_{i_1|A_1}x_{i_2|A_2}\cdots x_{i_k|A_k}$ be a monomial representation of a UEC-representative $\U$.
	If $c$ is any arbitrary element in $A_1\setminus A_2$, then the \textit{out-of-fiber-add} is defined as the move which takes $x_{i_1|A_1}x_{i_2|A_2}\cdots x_{i_k|A_k}$ to $x_{i_1|A_1}x_{i_2|A_2\cup c}\cdots x_{i_k|A_k}$.
	In other words, \textit{out-of-fiber-add} can be represented as a quadratic binomial of the form \[ x_{i_1|A_1}x_{i_2|A_2}-x_{i_1|A_1}x_{i_2|A_2\cup c} \] where $c\in A_1\setminus A_2$.
	Similarly, we define the move \textit{out-of-fiber-delete} as follows: if $c'\in A_1\cap A_2$, then \textit{out-of-fiber-delete} takes $x_{i_1|A_1}x_{i_2|A_2}\cdots x_{i_k|A_k}$ to $x_{i_1|A_1}x_{i_2|A_2\setminus c'}\cdots x_{i_k|A_k}$.
	In terms of quadratic binomials, \textit{out-of-fiber-delete} can be seen as \[ x_{i_1|A_1}x_{i_2|A_2}-x_{i_1|A_1}x_{i_2|A_2\setminus c'} \] where $c'\in A_1\cap A_2$.
\end{defn}

Notice that \textit{within-fiber} moves (i.e., those moves coming from the Gr\"obner basis in subsection~\ref{subsec: grobner basis}) always preserve the union and intersection of the sets $A_i$ and $A_j$, whereas the out-of-fiber moves explicitly change them.
Specifically, for any $d\in i_2\cup A_2$, \textit{out-of-fiber-add} adds edges of the form $c\mathdash d$ and \textit{out-of-fiber-delete} deletes edges of the form $c'\mathdash d$.
Furthermore, the new monomial obtained from these two moves preserves the property that $i_j\notin A_j'$ for any $j,j' \in [n]$.
Thus, we know that we do not move to any graph representing an empty UEC, as desired.
It is also important to note that an out-of-fiber-add (out-of-fiber-delete) can be undone with the corresponding out-of-fiber-delete (out-of-fiber-add).

\begin{prop}
	\label{prop:out-of-fiber-inverse}
	The out-of-fiber-add and out-of-fiber-delete moves are inverses of each other.
\end{prop}
\begin{proof}
	This follows from the definition of the two moves.
\end{proof}
\

Now that we have defined these two moves, we show how to use them to move between graphs in different generalized fibers having the same intersection number.
As $\mathbb{U}^n_1$ contains exactly one element, we first focus on $\mathbb{U}^n_2$.
Note that each element in $\mathbb{U}^n_2$ is of the form $x_{i_1|A_1}x_{i_2|A_2}$.
Observe that even though a graph in $\mathbb{U}^n_2$ can have multiple representations, the elements in $A_1\cap A_2$ always remain the same in every representation.
As noted above, a key difference between the within-fiber moves and the out-of-fiber moves is that any within-fiber move on $x_{i_1|A_1}x_{i_2|A_2}$ preserves $A_1\cap A_2$ whereas the out-of-fiber-add and -delete increase and decrease the cardinality of $A_1 \cap A_2$, respectively.
This is useful for traversing $\mathbb{U}^n_2$ and in general $\mathbb{U}^n_i$.

\begin{lemma}
	\label{lemma:connectivity-intersection number 2}
	Let $\U$ and $\U'$ be two graphs in $\mathbb{U}^n_2$ lying in different generalized fibers.
	Then there exists a set of within-fiber, out-of-fiber-add and out-of-fiber-delete moves which connects $\U$ with $\U'$.
\end{lemma}
\begin{proof}
	Let $\U$ and $\U'$ have monomial representations $x_{i_1|A_1}x_{i_2|A_2}$ and $x_{j_1|B_1}x_{j_2|B_2}$, respectively.
	If there does not exist any representation of $\U'$ which has $i_1$ (similarly $i_2$) as source node, then we know that $i_1$ (similarly $i_2$) lies in $B_1\cap B_2$.
	Thus, we can use two out-of-fiber-delete moves,
	\begin{eqnarray*}
		&x_{j_1|B_1}x_{j_2|B_2}-x_{j_1|B_1\setminus i_1}x_{j_2|B_2} \text{ and}\\ &x_{j_1|B_1\setminus i_1}x_{j_2|B_2}-x_{j_1|B_1\setminus i_1}x_{j_2|B_2\setminus i_2},
	\end{eqnarray*}
	to reach the graph with representation $x_{j_1|B_1\setminus i_1}x_{j_2|B_2\setminus i_2}$.
	Now, as $i_1$ and $i_2$ are no longer in the intersection, we can change the representation to $x_{i_1|B_1\cup j_1}x_{i_2|B_2\cup j_2}$.
	Now, let $a_1$ lie in $A_1\cap A_2$ but not in $(B_1\cup j_1)\cap (B_2\cup j_2)$.
	Without loss of generality, we can assume that $a_1$ lies in $B_1\cup j_1$.
	We can then use the out-of-fiber-add move \[ x_{i_1|B_1\cup j_1}x_{i_2|B_2\cup j_2}-x_{i_1|B_1\cup j_1}x_{i_2|B_2\cup j_2\cup a_1}, \] which brings $a_1$ in the intersection of $B_1\cup j_1$ and $B_2\cup j_2\cup a_1$.
	Similarly, if $b_1$ lies in $B_1\cap B_2$ but not in $A_1\cap A_2$, we apply the corresponding out-of-fiber-delete move to remove $b_1$ from the intersection.

	Continuing this process, we reach some $x_{i_1|B_1'}x_{i_2|B_2'}$ such that $A_1\cap A_2$ is equal to $B_1'\cap B_2'$.
	Once we are at this stage, we know that $x_{i_1|B_1'}x_{i_2|B_2'}$ lies in the same generalized fiber as $x_{i_1|A_1}x_{i_2|A_2}$, and hence can be reached by using the within-fiber moves from $\mathcal{F}$.
\end{proof}

We use a similar idea as above for showing the connectivity of graphs lying in any arbitrary $\mathbb{U}^n_m$ for $m\geq 3$.

\begin{lemma}
	\label{lem: completeness of out-of-fiber moves}
	Let $\U$ and $\U'$ be two graphs in $\mathbb{U}^n_m$ with $m\geq 3$, lying in different generalized fibers.
	Then there exists a set of within fiber, out-of-fiber-add and delete moves which connects $\U$ with $\U'$.
\end{lemma}
\begin{proof}
	Let $\U$ and $\U'$ have monomial representations $x_{i_1|A_1}x_{i_2|A_2}\cdots x_{i_m|A_m}$ and $x_{j_1|B_1}x_{j_2|B_2}\cdots x_{j_m|B_m}$, respectively.
	Without loss of generality, we can assume that $i_1$ lies in $B_1$.
	If there does not exist any representation of $\U'$ having $i_1$ as a source node, then we know that $i_1$ lies in at least one other $B_k$, i.e., in $B_1\cap B_k$ for some $k\in\{2,\ldots,m\}$.
	We then apply the out-of-fiber-delete move $x_{j_1|B_1}x_{j_k|B_k}-x_{j_1|B_1}x_{j_k|B_k\setminus \{i_1\}}$ to remove $i_1$ from the intersection.
	By repeating this process we can make sure that $i_1$ does not lie in any other $B_k$, and hence we can change the representation of $x_{j_1|B_1}$ to $x_{i_1|B_1\cup j_1\setminus \{i_1\}}$.

	We can iterate this process over $i_1,\ldots, i_m$.
	If, at some point in this iteration, we have $x_{i_k|B_{k'}}$ in an intermediate step with $i_{l}\in B_{k'}$ and not in any other $B_k$, then we apply a \textit{within-fiber} move of the form $x_{j_l|B_l}x_{i_k|B_{k'}}-x_{j_l|B_l\cup i_l}x_{i_k|B_{k'}\setminus \{i_l\}}$ and change the representation of $x_{j_l|B_l\cup i_l}$ to $x_{i_l|B_l\cup j_l}$ to get $i_l$ as a source node.
	Continuing this process, we arrive at an intermediate graph $\U_1$ whose monomial representation is of the form $x_{i_1|B_1'}x_{i_2|B_2'}\cdots x_{i_m|B_m'}$, i.e., a monomial representation of a graph having the same source nodes as $\U$.
	In order to reach $\U$ from $\U_1$, we pick any two cliques, say $x_{i_1|B_1'}x_{i_2|B_2'}$ and apply out-of-fiber-add and -delete moves to make $B_1'\cap B_2'$ equal to $A_1\cap A_2$ as in the proof of Lemma~\ref{lemma:connectivity-intersection number 2}.
	Repeating this step at most $m \choose 2$ times, we can move to a graph $\U_2$ whose monomial representation is $x_{i_1|\overline{B}_1}x_{i_2|\overline{B}_2}\cdots x_{i_m|\overline{B}_m}$ where $A_k\cap A_{k'}$ is equal to $\overline{B}_k\cap \overline{B}_{k'}$ for all $k,k'$.
	However, this implies that $\U$ and $\U_2$ lie in the same generalized fiber and hence can be connected by a sequence of within-fiber moves.
\end{proof}

We explain the above process with an example.

\begin{ex}
	Let $\U$ and $\U'$ be two UEC-representatives in $\mathbb{U}^5_3$ having monomial representations $x_{1|3}x_{2|35}x_{4|35}$ and $x_{1|2}x_{3|25}x_{4|2}$, respectively.
	As $3$ cannot be a source node in any representation of $\U$, we first apply some out-of-fiber-delete moves on $\U$ so that we can change the representation and use $3$ as a source node.
	We apply the \textit{out-of-fiber-delete} binomial $x_{1|3}x_{2|35}-x_{1|\emptyset}x_{2|35}$ on $\U$ to reach the graph $\U_1$ with representation $x_{1|\emptyset}x_{2|35}x_{4|35}$.
	We then apply the within-fiber binomial $x_{2|35}x_{4|35}-x_{2|35}x_{4|5}$ on $\U_1$ to reach the graph $\U_2$ with representation $x_{1|\emptyset}x_{2|35}x_{4|5}$.
	Now we can change the representation $\U_2$ to $x_{1|\emptyset}x_{3|25}x_{4|5}$ and make $3$ a source node.
	We then apply the \textit{within fiber} move $x_{1|\emptyset}x_{3|25}-x_{1|2}x_{3|5}$ on $\U_2$ to reach $\U_3$ with representation $x_{1|2}x_{3|5}x_{4|5}$.
	Observe that in this step, the first terms of $\U_3$ and $\U'$ are equal, i.e., $x_{1|2}$.

	Now, we apply the \textit{out-of-fiber-delete} move $x_{3|5}x_{4|5}-x_{3|5}x_{4|\emptyset}$ on $\U_3$ to arrive at $\U_4$ with representation $x_{1|2}x_{3|5}x_{4|\emptyset}$.
	From here, the only step required is to bring $2$ into the second and third term.
	Thus, we apply the \textit{out-of-fiber-add} move $x_{1|2}x_{3|5}-x_{1|2}x_{3|25}$ to reach $\U_5$ with representation $x_{1|2}x_{3|25}x_{4|\emptyset}$.
	Finally, we apply another \textit{out-of-fiber-add} move $x_{1|2}x_{4|\emptyset}-x_{1|2}x_{4|2}$ to reach $\U'$.
\end{ex}

\subsection{Moving between UECs within different intersection numbers}
\label{subsec: combinatorial}

So far, we have the binomials in $\mathcal{F}$ which we use to move within any generalized fiber and then the out-of-fiber moves which we use to explore any $\mathbb{U}^n_i$ for $1\leq i \leq n$, i.e., the set of all graphs representing nonempty UECs that have intersection number $i$.
However, neither the binomials in $\mathcal{F}$ (or more generally in $ker(\phi)$) nor the out-of-fiber moves allow us to change the intersection number when applied to any graph representation.
This is due to the fact that the monomial map $\phi$ defined in \eqref{eqn:phi} in subsection~\ref{subsec: monomial representation} is such that the $y_i$ variable preserves the degree of a given monomial in the image of $\phi$.
Hence, in order to move between graphs having different intersection numbers, we need to define a new set of moves.

Thus, we introduce two combinatorial moves called \textit{merge} and \textit{split} that will allow us to move between $\mathbb{U}^n_i$ and $\mathbb{U}^n_{i+1}$.
The two moves are defined as follows:
\begin{defn}
	\label{defn: merge and split} Let $\U$ be any graph in $\mathbb{U}^n_k$.
	If $\U$ has a monomial representation of the form $x_{i_1|A_1}x_{i_2|A_1}x_{i_3|A_3} \cdots x_{i_k|A_k}$, then the \textit{merge} operation `merges' $x_{i_1|A_1}x_{i_2|A_1}$ into a single clique, resulting in the graph with monomial representation $x_{i_1|A_1\cup i_2}x_{i_3|A_3} \cdots x_{i_k|A_k}$.
\end{defn}

The graph resulting from applying \textit{merge} to a graph with intersection number $k$ has intersection number $k-1$.
The move can be seen in terms of a binomial as \[ x_{i_1|A_1}x_{i_2|A_1}-x_{i_1|A_1\cup i_2}.
\]
Similarly, we can increase the intersection number with the following move:
\begin{defn}
	For $\U$ in $\mathbb{U}^n_k$ with monomial representation $x_{i_1|A_1}x_{i_2|A_2} \cdots x_{i_k|A_k}$ such that $c$ lies in $A_1$ and no other $A_j$, the \textit{split} operation `splits' $x_{i_1|A_1}$ into two different cliques, resulting in a graph with monomial representation $x_{i_1|A_1\setminus c}x_{c|A_1\setminus c}x_{i_2|A_2} \cdots x_{i_k|A_k}$.
\end{defn}

Applying the \textit{split} operation to a graph with intersection number $k$ produces a graph with intersection number $k+1$.
In terms of a binomial, the move can be seen as \(x_{i_1|A_1}-x_{i_1|A_1\setminus c}x_{c|A_1\setminus c}\).
Directly from the definitions of \textit{merge} and \textit{split}, we have:
\begin{prop}
	\label{prop:merge split inverse} The operations \textit{merge} and \textit{split} are inverses of each other.
\end{prop}
Also observe that applying \textit{merge} or \textit{split} preserves the property that $i_j\notin A_l$ for any $j,l\in [n]$, ensuring that we stay within the space of UEC-representatives.

\subsection{Connecting the space of all nonempty UECs}
\label{subsec: connecting}
We now show that that the within fiber and out-of-fiber moves combined with the \textit{merge} and \textit{split} are sufficient to traverse the entire space of UEC-representatives on a fixed number of nodes.

\begin{thrm}
	\label{thrm:exploring the entire space}
	Let $\mathbb{U}^n$ be the collection of all UEC-representatives on $n$ nodes.
	Then for any two graphs $\U$ and $\U'$ in $\mathbb{U}^n$, there exists a sequence of within fiber, out-of-fiber, \textit{merge} and/or \textit{split} moves that connects $\U$ and $\U'$.
\end{thrm}

\begin{proof}
	By combining Lemmas \ref{lemma:connectivity-intersection number 2} and \ref{lem: completeness of out-of-fiber moves}, we know that $\mathbb{U}^n_k$ is connected via within fiber and out-of-fiber moves for every $k\in[n]$.
	Hence, it remains to show that if $\U\in \mathbb{U}^n_k$ and $\U^\prime\in\mathbb{U}^n_m$, where $k<m$, then there is a sequence of moves taking us from $\U$ to $\U^\prime$ and a sequence of moves taking us from $\U^\prime$ to $\U$.
	In the former case, we know that the graph $\U^\ast$ consisting of an $(n-k+1)$-clique and $k-1$ disconnected nodes is contained in $\mathbb{U}^n_k$.
	Hence, we can move from $\U$ to $\U^\ast$ via a sequence of within and/or out-of-fiber moves.
	Applying the \textit{split} operation to split the $(n-k+1)$-clique in $\U^\ast$ then produces a graph in $\mathbb{U}^n_{k+1}$.
	We can then iterate this procedure until we arrive at a graph in $\mathbb{U}^n_m$.
	Using within and/or out-of-fiber moves, we can eventually transform our graph to $\U^\prime$.

	For the latter case, we have already shown in Proposition \ref{prop:out-of-fiber-inverse} and \ref{prop:merge split inverse} that the \textit{out-of-fiber-add} and \textit{out-of-fiber-delete} (and similarly \textit{merge} and \textit{split}) are inverses of each other.
	So, we can use the inverses of the set of moves used to reach from $\U$ to $\U'$ to reach from $\U'$ to $\U$.
\end{proof}

\section{DAG-reduction Representatives of UECs}
\label{sec: DAG reductions}
This section is devoted to formally presenting algorithms for the within-fiber, out-of-fiber, merge and split moves developed in Sections~\ref{sec:grobner} and~\ref{sec: traversing} so that they can be implemented with a feasible level of complexity.
We do this by defining a new, more efficient (in terms of time and space complexity) representation of UECs (Section~\ref{sec:dag-reduct-repr}).
We then present pseudocode for moves between UEC-representatives in terms of these new representations, and we prove that these moves are equivalent to the moves developed in Section~\ref{sec: traversing}, allowing for a more efficient implementation.

In particular, by Theorem~\ref{thrm:exploring the entire space}, we know that the moves defined in Sections~\ref{sec:grobner} and~\ref{sec: traversing} could be implemented on UEC-representatives to explore the space of nonempty UECs.
However, a naive implementation using these undirected graphs is needlessly inefficient.
Thus, in Subsection~\ref{sec:dag-reduct-repr}, we define the \emph{DAG-reduction}, a unique encoding of a UEC-representative that efficiently captures the ancestral relations that define the UEC while requiring fewer vertices and edges.
This allows for an efficient implementation of the moves described in Sections~\ref{sec:grobner} and~\ref{sec: traversing}.

To this end, in Subsection~\ref{sec:pseud-its-equiv}, we present pseudocode defining moves \texttt{merge}, \texttt{split}, \texttt{out-del}, \texttt{out-add} and \texttt{within} on DAG-reductions and prove that they are, respectively, equivalent to the moves \textit{merge}, \textit{split}, \textit{out-of-fiber-delete}, \textit{out-of-fiber-add} and \textit{within-fiber} defined on UEC-representatives in Sections~\ref{sec:grobner} and~\ref{sec: traversing}.
In Section~\ref{sec:grues:-markov-chain}, we implement an MCMC algorithm called \texttt{GrUES} using the moves on DAG-reductions.

\subsection{DAG-reductions}
\label{sec:dag-reduct-repr}

In the following, we define a \emph{DAG-reduction}, which can be constructed in \(\mathcal{O}(|V^\U|)\) time and provides an efficient encoding of for UECs for the implementation in Section~\ref{sec:grues:-markov-chain}.
We first introduce some necessary terminology.

For a given DAG $\D$, the \emph{CPDAG} (also known as the \emph{essential graph}) of $\D$ is the partially directed graph having the same vertices as $\D$ and edge set $E$ where a directed edge $i\rightarrow j$ is in $E$ if and only if $i\rightarrow j$ is present in every DAG in the Markov equivalence class of $\D$.
Further, we have that an undirected edge $i\mathdash j$ is in $E$ if and only if $i$ and $j$ are adjacent in every DAG in the Markov equivalence class of $\D$ and there are two DAGs in the equivalence class for which this edge has opposing orientations.
(For more details, see Definition 2.1 of \citep{AMP97}).
Note that two DAGs are Markov equivalent if and only if they have the same CPDAG, meaning that a CPDAG is a graphical representation of a Markov equivalence class.
The chain components of a CPDAG $\G$ are the connected subgraphs of $\G$ consisting of only undirected edges.
For a given CPDAG $\G$, we define a \emph{chain component} of $\G$ \emph{with respect to a vertex $v$} to be the largest collection of vertices that are connected to $v$ by undirected paths in $\G$.
It is denoted by \[\cc_{\G}(v):=\{w\in V^{\G} :v\text{ and } w \text{ are connected by an undirected path in } \G \}.
\]
In the literature on CPDAGs it is standard to view a chain component as a subgraph of the CPDAG, but for the purposes of our implementation it will be beneficial to view the chain components $\cc_{\G}(v)$ as sets of nodes (which in our case correspond to complete subgraphs).
Note that $\cc_{\G}(v) = \{v\}$ whenever there are no undirected edges incident to $v$ in $\G$.
Hence, $\cc_{\G}(v)$ is well-defined and nonempty for every node in $\G$.
In particular, for a CPDAG $\G = (V, E)$, the collection of chain components $\{\cc_{\G}(v) : v \in V\}$ partitions the nodes $V$ of $\G$.

It is possible to associate the CPDAG of a unique Markov equivalence class to a UEC-representative.
This is accomplished by Algorithm~\ref{alg:init-cpdag}.
In the following, we say that a DAG $\D$ is \emph{maximal} in its UEC if any DAG produced by adding an edge to $\D$ is not in the same UEC as $\D$.
As we will see in Corollary~\ref{cor:maximal}, the set of all maximal DAGs in a given UEC forms a unique Markov equivalence class, and its corresponding CPDAG is the output of Algorithm~\ref{alg:init-cpdag}.

\begin{algorithm}[H]
	\Input{Undirected graph \(\U\) representing a nonempty UEC}
	\Output{CPDAG \(\G\) of the Markov equivalence class of maximal DAGs in the UEC represented by $\U$}
	\BlankLine
	$\Pi \coloneqq \{\pi = v \mathdash v' \mathdash v'' : \pi \mbox{ an induced path in } \U \}$, i.e., $(v,v'') \notin E^{\U}$\;
	$\G \coloneqq \U$ with all $\pi\in \Pi$ oriented as v-structures \(v \rightarrow v' \leftarrow v''\)\;
	$\G \coloneqq \G$ with all bidirected edges \(v \leftrightarrow w\) removed\;
	\Return $\G$
	\caption{\(\mathtt{init\_CPDAG}(\U)\)}
	\label{alg:init-cpdag}
\end{algorithm}

The fact that Algorithm~\ref{alg:init-cpdag} has the desired output requires proof, which will be given in Theorem~\ref{thrm:init}.
However, we first use this algorithm to obtain the central object of study in this section.
Using the output of Algorithm~\ref{alg:init-cpdag}, we can define the DAG-reduction of a UEC-representative.

\begin{defn}
	\label{defn:dag-reduction}
	The \textit{DAG-reduction} of a UEC-representative \(\U\) is defined as \(\D^\U \coloneqq (V, E)\), where \(V \coloneqq \cc_\G(V^\G)\), \[ E \coloneqq \{\mathbf{v} \rightarrow \mathbf{w} : \mathbf{v}, \mathbf{w} \in V \text{ and there is a } v \rightarrow w \in E^\G \text{ such that } v\in \mathbf{v}, w \in \mathbf{w}\}, \] and \(\G \coloneqq \mathtt{init\_CPDAG}(\U)\) using Algorithm~\ref{alg:init-cpdag}.
\end{defn}

Note that the sets $\mathbf{v}$ and $\mathbf{w}$ are chain components of the CPDAG \(\G = \mathtt{init\_CPDAG}(\U)\) in the definition of $E$ in Definition~\ref{defn:dag-reduction}.
Hence, they are sets of vertices of $\U$, allowing us to consider vertices $v\in \mathbf{v}$ and $w\in\mathbf{w}$.
We illustrate Algorithm~\ref{alg:init-cpdag} and the construction of a DAG-reduction of a UEC-representative in the following example.

\begin{ex}
	\label{ex:reduced DAG}
	\begin{figure}
		\begin{tikzpicture}[scale=.45]
			\filldraw[black]
			(-3,0) circle [radius=0] node [below] {}
			(-2,1.5) circle [radius=.04] node [above] {1}
			(-1,0) circle [radius=.04] node [below] {2}
			(0,1.5) circle [radius=.04] node [above] {4}
			(1,0) circle [radius=.04] node [below] {3}
			(2,1.5) circle [radius=.04] node [above] {5}
			(3,0) circle [radius=.04] node [below] {6}
			(0,-.7) circle [radius=0] node [below] {$\U$};
			\draw
			(-2,1.5)--(-1,0)
			(-1,0)--(3,0)
			(-1,0)--(0,1.5)
			(0,1.5)--(1,0)
			(1,0)--(2,1.5)
			(2,1.5)--(3,0);

			\filldraw[black]
			(5,1.5) circle [radius=.04] node [above] {1}
			(6,0) circle [radius=.04] node [below] {2}
			(7,1.5) circle [radius=.04] node [above] {4}
			(8,0) circle [radius=.04] node [below] {3}
			(9,1.5) circle [radius=.04] node [above] {5}
			(10,0) circle [radius=.04] node [below] {6}
			(7.4,-.7) circle [radius=0] node [below] {$\mathcal{H}$};
			\draw
			(5,1.5)--(6,0)
			(6,0)--(8,0)
			(8,0)--(10,0)
			(7,1.5)--(6,0)
			(7,1.5)--(8,0)
			(9,1.5)--(8,0)
			(9,1.5)--(10,0);

			\draw
			(7.6-2,.64)--(7.6-2,.84)
			(7.6-2,.64)--(7.4-2,.64)
			(6.4,.64)--(6.4,.84)
			(6.4,.64)--(6.6,.64)
			(8.4,.64)--(8.4,.84)
			(8.4,.64)--(8.6,.64)
			(7.6,.64)--(7.6,.84)
			(7.6,.64)--(7.4,.64)
			(6.6,0)--(6.75,.15)
			(6.6,0)--(6.75,-.15)
			(7.4,0)--(7.25,.15)
			(7.4,0)--(7.25,-.15)
			(8.6,0)--(8.75,.15)
			(8.6,0)--(8.75,-.15);

			\filldraw[black]
			(-3+14,0+4) circle [radius=0] node [below] {}
			(-2+14,-2.5+4) circle [radius=.04] node [above] {1}
			(-1+14,-4+4) circle [radius=.04] node [below] {2}
			(0+14,-2.5+4) circle [radius=.04] node [above] {4}
			(1+14,-4+4) circle [radius=.04] node [below] {3}
			(2+14,-2.5+4) circle [radius=.04] node [above] {5}
			(3+14,-4+4) circle [radius=.04] node [below] {6}
			(0+14,-4.9+4) circle [radius=0] node [below] {CPDAG};
			\draw
			(-2+14,-2.5+4)--(-1+14,-4+4)
			(-1+14,-4+4)--(-1+14,-4+4)
			(1+14,-4+4)--(3+14,-4+4)
			(-1+14,-4+4)--(0+14,-2.5+4)
			(0+14,-2.5+4)--(1+14,-4+4)
			(1+14,-4+4)--(2+14,-2.5+4)
			(2+14,-2.5+4)--(3+14,-4+4);

			\draw
			(7.6-2+7,.64)--(7.6-2+7,.84)
			(7.6-2+7,.64)--(7.4-2+7,.64)
			(-.6+14,-3.36+4)--(-.6+14,-3.16+4)
			(-.6+14,-3.36+4)--(-.4+14,-3.36+4)
			(1.4+14,-3.36+4)--(1.4+14,-3.16+4)
			(1.4+14,-3.36+4)--(1.6+14,-3.36+4)
			(0.6+14,-3.36+4)--(0.6+14,-3.16+4)
			(0.6+14,-3.36+4)--(0.4+14,-3.36+4)
			(1.6+14,-4+4)--(1.75+14,-3.85+4)
			(1.6+14,-4+4)--(1.75+14,-4.15+4);

			\filldraw[black]
			(5+14,-2.5+4) circle [radius=.04] node [above] {\{1\}}
			(7+14,-2.5+4) circle [radius=.04] node [above] {\{3\}}
			(9+14,-2.5+4) circle [radius=.04] node [above] {\{5,6\}}
			(6+14,-4+4) circle [radius=.04] node [below] {\{2\}}
			(8+14,-4+4) circle [radius=.04] node [below] {\{4\}}
			(7+14,-4.9+4) circle [radius=0] node [below] {DAG reduction};

			\draw
			(5+14,-2.5+4)--(6+14,-4+4)
			(6+14,-4+4)--(7+14,-2.5+4)
			(7+14,-2.5+4)--(8+14,-4+4)
			(8+14,-4+4)--(9+14,-2.5+4);

			\draw
			(5.55+14,-3.3+4)--(5.55+14,-3.1+4)
			(5.55+14,-3.3+4)--(5.35+14,-3.3+4)
			(6.45+14,-3.3+4)--(6.45+14,-3.1+4)
			(6.45+14,-3.3+4)--(6.65+14,-3.3+4)
			(7.55+14,-3.3+4)--(7.55+14,-3.1+4)
			(7.55+14,-3.3+4)--(7.35+14,-3.3+4)
			(8.45+14,-3.3+4)--(8.45+14,-3.1+4)
			(8.45+14,-3.3+4)--(8.65+14,-3.3+4);
		\end{tikzpicture}

		\caption{A UEC-representative $\U\in \mathbb{U}_3^6$ and its corresponding CPDAG and DAG-reduction.
			The graph $\mathcal{H}$ is the output of Line $2$ of Algorithm~\ref{alg:init-cpdag} applied to $\U$.
		}\label{figure:U-CPDAP-RDAG}
	\end{figure}
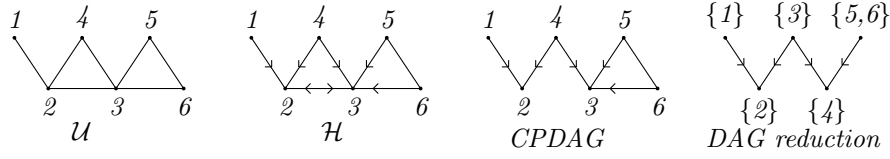
	Let $\U$ be the \emph{UEC}-representative as seen in Figure \ref{figure:U-CPDAP-RDAG}.
	We first orient all the induced $2$-paths in $\U$ as v-structures.
	For instance, $4\mathdash 3 \mathdash 6$ is an induced $2$-path in $\U$ as $4$ and $6$ are not adjacent.
	Hence, the edges are directed as $4\rightarrow 3\leftarrow 6$.
	Observe that the edge $2 \mathdash 3$ gets bidirected as the induced $2$-paths $2 \mathdash 3 \mathdash 5$ and $3 \mathdash 2 \mathdash 1$ are directed as $2 \rightarrow 3 \leftarrow 6$ and $3 \rightarrow 2 \leftarrow 1$, respectively.
	Orienting all the induced $2$-paths as v-structures and removing all the bidirected edges (which is just $2\leftrightarrow 3$ in this example) gives us the corresponding CPDAG of $\U$.

	Observe that $5$ and $6$ are the only two vertices in $\G$ that are connected by an undirected edge.
	This gives us that the only chain component of $\G$ that is not a singleton is $\{5,6\}$.
	We collapse the vertices $5$ and $6$ in $\G$ to form the chain component $\{5,6\}$, thereby obtaining the corresponding DAG-reduction $\D^{\U}$.

\end{ex}

As seen in Example \ref{ex:reduced DAG}, it is essential to identify the induced 2-paths in $\U$ by using its clique structure.
The following lemma characterizes all v-structures of $\D$ by using the minimum edge clique cover of $\U$.
Recall that we use the notation $\{\U\}$ to denote the UEC of DAGs with UEC-representative $\U$.
In the following, we will use the skeleton of a maximal DAG in $\{\U\}$.
Let $\widehat\U$ denote the undirected graph with the same vertices as $\U$ that does not contain the edge $v \mathdash w$ if and only if there exist cliques $C_1,C_2,C_3\in\E^\U$ such that $v,w\in C_1$, $v\in C_2$, $w\notin C_2$, $w\in C_3$ and $v\notin C_3$.
It is shown in the proof of \cite[Lemma~6]{pgm22} that a maximal DAG $\D$ in $\{\U\}$ has skeleton $\widehat\U$.
We also make use of edges that are \emph{implied by transitivity} or \emph{partially weakly covered}, defined in the paragraph preceding Example~\ref{ex:unconditional dependence}.

\begin{lemma}
	\label{lem:v-structures}
	Let $\U = (V,E)$ be a UEC-representative, and consider a maximal DAG $\D\in \{\U\}$.
	Then a path $( v, x, w)$ is a v-structure in $\D$ if and only if there exist cliques $C_1, C_2 \in \E^\U$ such that $v,x\in C_1$, $w\notin C_1$, $w,x\in C_2$, $v\notin C_2$, and any clique in $\E^\U$ containing $v$ or $w$ also contains $x$.
\end{lemma}

\begin{proof}
	Recall from Lemma~\ref{lem: unique minimum edge clique cover} and its proof that $\U$ has the unique minimum edge clique cover $\E^\U = \{\ngh_\U[m] : m\in M\}$, for any maximum independent set $M$ of $\U$.

	For the `only if' direction, consider nodes $v,x,w$ for which there are $C_1, C_2 \in \E^\U$ such that $v,x\in C_1$, $w\notin C_1$, $w,x\in C_2$, $v\notin C_2$, and any clique containing $v$ or $w$ also contains $x$.
	As $v$ and $x$ (and similarly $w$ and $x$) are adjacent in $\U$, by Theorem~\ref{thrm:alt-defs} there exists a trek between $v$ and $x$ (and similarly $w$ and $x$) in $\D$.

	If $v$ is the source node in the trek between $v$ and $x$, then $v\rightarrow x$ is implied by transitivity in $\D$ \citep[Section 2]{pgm22}.
	Hence the maximality of $\D$ in $\{\U\}$, together with \citep[Lemma~4]{pgm22}, implies that $v\rightarrow x\in E^\D$.
	If $v$ is not a source node in the trek between $v$ and $x$ then $\pa_\D(v)\neq \emptyset$.
	Moreover, as every clique in $\E^\U$ containing $v$ also contains $x$, any clique in $\E^\U$ containing $v$ and a parent of $v$ would also contain $x$.
	So, by Theorem~\ref{thrm:alt-defs}~(2), we have $\ma_\D(\pa(v))\subseteq \ma_\D(\pa(x)\setminus \{v\})$.
	We conclude that $\{v,x\}$ is partially weakly covered (see \citep[Definition~3]{pgm22}).
	Hence by maximality of $\D$ and \cite[Lemma~4]{pgm22}, it follows that $v\rightarrow x$ lies in $E^\D$.
	By symmetry, it also follows that the edge $w\rightarrow x$ is in $E^\D$.

	Since $\D$ is maximal in $\{\U\}$, it has skeleton \(\widehat\U\), and hence there is no edge between $v$ and $w$ in $\D$.
	Otherwise, $v\mathdash w$ would be an edge in \(\widehat\U\) such that $v,x\in C_1$ but $w\notin C_1$, $w,x\in C_2$ but $v\notin C_2$, and a clique $C_3\in \E^\U$ containing the edge $v\mathdash w$, contradicting the construction of \(\widehat\U\).
	Hence, $v\rightarrow x \leftarrow w$ is a v-structure in $\D$.

	Conversely, if $v\rightarrow x \leftarrow w$ is a v-structure in $\D$, then $v\in \an_\D(x)\cap \an_\D(v)$ and $w\in \an_\D(x)\cap \an_\D(w)$.
	Hence by (1) in Theorem~\ref{thrm:alt-defs}, the edges $v\mathdash x$ and $w\mathdash x$ are in $\U$.
	Moreover, since $v\rightarrow x$ in $\D$, any trek between a node $y$ of $\D$ and $v$ can be extended to a trek between $y$ and $x$.
	Hence, by Theorem~\ref{thrm:alt-defs}, if $y\mathdash v$ is in $E^\U$ then $y\mathdash x$ is in $E^\U$.
	Thus, any clique in $\E^\U$ containing $v $ has to also contain $x$.
	By symmetry, we also have that every clique in $\E^\U$ containing $w$ has to contain $x$.
	Lastly, we need to show that there exists a clique in $\E^\U$ containing $v$ and $x$ but not $w$, as well as a clique in $\E^\U$ containing $w$ and $x$ but not $v$.
	If $v$ and $w$ are not adjacent in $\U$, then we have the desired clique structure since any clique containing $v$ has to also contain $x$ but not $w$, and any clique containing $w$ has to contain $x$ but not $v$.
	In case $v\mathdash w$ in $\U$, we can look at the skeleton \(\widehat\U\) of $\D$.
	Since $v\rightarrow x \leftarrow w$ is a v-structure in $\D$, the vertices $v$ and $w$ cannot be adjacent in \(\widehat\U\).
	Hence, there exist cliques $C_1, C_2\in \E^\U$ such that $v$ (and hence $x$) lies in $C_1$ but $w\notin C_1$, and $w$ (and hence $x$) lies in $C_2$ but $v\notin C_2$.
\end{proof}

Lemma~\ref{lem:v-structures} describes the v-structures in any maximal DAG in a given UEC.
We can now show that such maximal DAGs form a Markov equivalence class.

\begin{coro}
	\label{cor:maximal}
	Suppose that $\U = (V,E)$ is a UEC-representative.
	Then the maximal DAGs in $\{\U\}$ form a Markov equivalence class.
\end{coro}
\begin{proof}
	It follows from the proof of \citep[Lemma~6]{pgm22} that any two maximal DAGs in the UEC represented by $\U$ have the same skeleton, $\widehat{\U}$.
	In Lemma~\ref{lem:v-structures}, it was shown that the v-structures of a maximal DAG $\D$ in $\{\U\}$ are independent of the choice of $\D$.
	Therefore, any two maximal DAGs have the same skeleton and v-structures, which means that they are Markov equivalent \citep{VP90}.
\end{proof}

In order to use the properties of CPDAGs to our advantage, we now show that the output of Algorithm~\ref{alg:init-cpdag} when applied to a UEC-representative is indeed \emph{a} CPDAG and then that it is specifically \emph{the} CPDAG uniquely representing the Markov equivalence class of the maximal DAGs in the UEC.
We further show that there is a one-to-one correspondence between any UEC-representative $\U$ and the CPDAG \(\mathtt{init\_CPDAG}(\U)\).
This will enable us to use versions of the moves identified for UECs in Section~\ref{sec: traversing} that operate on the level of DAG-reductions, which will allow for a more efficient implementation.

\begin{lemma}
	\label{lemma:init-cpdag}
	Let $\U$ be a UEC-representative.
	Then the output of $\mathtt{init\_CPDAG}(\U)$, denoted \(\G^{\mathtt{init}}\), is a CPDAG.
\end{lemma}
\begin{proof}
	It suffices to show that \(\G^{\mathtt{init}}\) satisfies conditions (i)--(iv) of \citep[Theorem~4.1]{AMP97}, which we do in the following paragraphs.

	We start with condition (iii), assuming for a proof by contradiction that \(\G^{\mathtt{init}}\) contains the induced subgraph \(v \rightarrow u \mathdash w.
	\)
	It must be that $v$ and $w$ are adjacent in $\U$, otherwise the induced path $(v,u,w)$ (by Line~2 of Algorithm~\ref{alg:init-cpdag}) would necessitate the contradicting edge \(u \leftarrow w\) in \(\G^{\mathtt{init}}\).
	Hence, because \(v,w\) are adjacent in \(\U\) but not \(\G^{\mathtt{init}}\), we must have the edge \(v \leftrightarrow w\) after Line~2 of Algorithm~\ref{alg:init-cpdag}.
	This implies there exists an \(x\) such that \(x\) and \(w\) are not adjacent in \(\U\) and \((x,v, w)\) is an induced 2-path in \(\U\).
	However, to avoid the induced 2-path \(x, v, u\) in \(\U\) (which would contradict \(v \rightarrow u\) after Line~2) we must also have that \(x\) and \(u\) are adjacent in \(\U\).
	Furthermore, to avoid the induced 2-path \(x, u, w\) in \(\U\) (which would contradict \(u \mathdash w\) after Line~2), we must have that \(x\) and \(w\) are adjacent, resulting in a contradiction.
	Thus, \(\G^{\mathtt{init}}\) cannot contain the induced subgraph \(v \rightarrow u \mathdash w \).

	For condition (i), it suffices to show that \(\G^{\mathtt{init}}\) contains no partially directed cycles.
	For a proof by contradiction, assume that it does contain a partially directed cycle.
	By condition (iii), it suffices to consider three cases here, 3-cycles of the form (1) \(v \mathdash u \mathdash w\) with \(v \rightarrow w\), or (2) \(v \rightarrow u \mathdash w\) with \(v \leftarrow w\) or (3) fully directed \(k\)-cycles (possibly with \(k > 3\)).
	For (1), we will show the stronger property that if \(v\mathdash u, u\mathdash w\) are in \(\G^{\mathtt{init}}\), then $v \mathdash w$ is in \(\G^{\mathtt{init}}\).
	Since $v \mathdash u$ does not get directed in Line~2 of Algorithm~\ref{alg:init-cpdag}, it must be that $v \mathdash w \in E^\U$.
	In order to prove the claim by contradiction, let us assume that $v\mathdash w$ gets bidirected in Line~2 of Algorithm~\ref{alg:init-cpdag} and hence removed in Line 3.
	Then, there must exist $x,z$ such that $x\mathdash v$ where $x$ is not adjacent to $w$, and $w\mathdash z$ where $v$ is not adjacent to $z$, in $\U$, i.e., $(x,v,w), (v,w,z)$ are induced paths.
	Moreover, it cannot be that $(x,v,u)$ forms an induced path in $\U$ because then Line~2 of Algorithm~\ref{alg:init-cpdag} would direct $v\rightarrow u$; hence $x\mathdash u$ in $\U$.
	Similarly, the path $(x,u,w)$ cannot be induced in $\U$, again because $u\mathdash w$ would get directed in Line~2 of Algorithm~\ref{alg:init-cpdag}.
	But $x$ and $w$ cannot be adjacent in $\U$, which leads to a contradiction.
	For (2), there must exist a vertex $x$ (resp.
	$z$) such that $(v,u,x)$ (resp. $(z,v,w)$) forms an induced path in $\U$ so that the edge $v\rightarrow u$ (resp. $v\leftarrow w $) gets directed in Line~2 of Algorithm~\ref{alg:init-cpdag}.
	Additionally, since $u\mathdash w$ does not get directed in Line~2 of Algorithm~\ref{alg:init-cpdag}, it must be that $x$ and $w$ are adjacent in $\U$, i.e., $(x,u,w)$ is not an induced path.
	Similarly, $x\mathdash w\mathdash v$ cannot be an induced path in $\U$ because Line~2 of Algorithm~\ref{alg:init-cpdag} would introduce direction $v\rightarrow w$; hence $v\mathdash x$ in $\U$, contradicting that $(v,u,x)$ is an induced path.
	For (3), assume that $v\rightarrow u \rightarrow w$ and $w\rightarrow v$ in $G^{\mathtt{init}}$.
	Then, arguing as in the previous lines, there exists $x$ such that $(x,v,w)$ is an induced path in $\U$ (so that $v\leftarrow w$ gets directed), but $(x,v,u)$ is not an induced path in $\U$ (because Line~2 of Algorithm~\ref{alg:init-cpdag} would direct $v\mathdash u$ in the wrong direction), hence $x\mathdash u$ is an edge in $\U$.
	Since $x$ is not adjacent to $w$ in $\U$, $(x,u,w)$ is an induced path, directing the edge between $u$ and $w$ in the wrong direction in Line~2 of Algorithm~\ref{alg:init-cpdag}, a contradiction.

	For condition (ii), it suffices to show that the subgraph induced by considering only undirected edges of \(\G^{\mathtt{init}}\) is chordal.
	As in the proof of condition (i), (1), every undirected 2-path in \(\G^{\mathtt{init}}\) is part of an undirected 3-cycle, and hence there are no $k$-cycles for \(k >3\).

	Lastly, for condition (iv), we have that all directed edges in \(\G^{\mathtt{init}}\) are strongly protected (namely, they are in configuration (b) of \citep[Definition~3.3]{AMP97}), following from the fact that Line~2 of Algorithm~\ref{alg:init-cpdag} only directs edges into v-structures and the fact that Line~3 maintains this property while removing bidirected edges.
	That is, for every induced subgraph of the form \(x \rightarrow v \leftrightarrow w\) after Line~2, there exists an \(m\) such that \(x \rightarrow v \leftarrow m\) is preserved as an induced subgraph after Line~3.
	Namely, this \(m\) is part of a maximum independent set that defines \(\E^\U\), the unique minimum edge clique cover of \(\U\) (see Theorem~\ref{thrm:alpha-equals-delta}, Lemma~\ref{lem: unique minimum edge clique cover}, and Definition~\ref{defn:uniqueECC} along with the sentence following it).
	Specifically, observe that \((x,v,w)\) is an induced path in \(\U\), meaning we have cliques \(C_1, C_2 \in \E^\U\) such that \(x,v \in C_1\) but \(w \not\in C_1\) and \(v,w \in C_2\) but \(x \not\in C_2\).
	Hence, there must be an \(m \in C_2\) that is part of a maximum independent set in \(\U\), giving us that \((x,v,m)\) is an induced path in \(\U\) and that there exists no induced path \((u,m,u')\) in \(\U\).
	Thus, \(x \rightarrow v \leftarrow m\) remains as an induced subgraph after Line~3.
\end{proof}

\begin{thrm}
	\label{thrm:init}
	Let $\U$ be a UEC-representative.
	Then the output of $\mathtt{init\_CPDAG}(\U)$, denoted \(\G^{\mathtt{init}}\), is the unique CPDAG corresponding to the maximal DAGs in the UEC $\{\U\}$.
\end{thrm}

\begin{proof}
	Having shown in Lemma~\ref{lemma:init-cpdag} that \(\G^{\mathtt{init}}\) is indeed a CPDAG, it suffices here to show that \(\G^{\mathtt{init}}\) has the same skeleton and v-structures as (and is thus identical to) the unique CPDAG \(\G^{\U}\) representing the MEC of maximal DAGs in \({\U}\).
	Note that \(\G^{\U}\) is well-defined and unique, by Corollary~\ref{cor:maximal}, and that according to Definition~\ref{defn:uniqueECC}, we make use of the unique minimum edge clique cover $\E^\U=\{C_m: m\in M\}$ of $\U$, where $M$ is a maximum independent set of $\U$ and $C_m:=\text{ne}_\U[m]$ is the closed neighborhood of \(m\) in \(\U\).

	CPDAGs must have the same skeleton as DAGs to which they are Markov equivalent, so we know that \(\G^{\U}\) has skeleton \(\widehat\U\).
	Furthermore, \(\widehat\U\) and the skeleton of \(\G^{\mathtt{init}}\) are both subgraphs of \(\U\), i.e., they are both constructed from \(\U\) by removing certain edges.
	Hence, it suffices to show that for each edge \(v \mathdash w \in E^{\U}\), we have that \(v \mathdash w\) is not in \(\widehat\U\) if and only if it is not in the skeleton of \(\G^{\mathtt{init}}\).

	To see this, suppose that $v$ and $w$ are adjacent in $\U$ but not in the skeleton of \(\G^{\mathtt{init}}\).
	This means that the edge between $v$ and $w$ was bidirected in Line~$2$ of Algorithm~\ref{alg:init-cpdag}.
	Hence there are two induced $2$-paths $(s,v,w )$ and $( v,w,t )$ in $\U$.
	Since all edges must be in some clique of $\E^\U$, it follows that $v,w\in C_1$, $s,v\in C_2$, $w\notin C_2$, $w,t\in C_3$ and $v\notin C_3$, for some $C_1,C_2,C_3\in\E^\U$.
	It follows that $v\mathdash w$ gets removed from $\U$ when producing \(\widehat\U\).

	Conversely, $\widehat{\U}$ is produced by removing any edge in $\U$ for which there exist $C_1,C_2,C_3\in\E^\U$ such that $v,w\in C_1$, $v\in C_2$, $w\notin C_2$, $w\in C_3$ and $v\notin C_3$.
	It follows that there exist $s\in C_2$, not adjacent to $w$, and $t\in C_3$, not adjacent to $v$.
	Thus, there are two induced $2$-paths $( s,v,w )$ and $( v,w,t )$ in $\U$, which implies that the edge $v\mathdash w$ is removed in Line~$3$ of Algorithm~\ref{alg:init-cpdag} while producing \(\G^{\mathrm{init}}\).

	CPDAGs must also have the same v-structures as the DAGs to which they are Markov equivalent.
	This means that Lemma~\ref{lem:v-structures} characterizes the v-structures of \(\G^{\U}\).
	Hence, it suffices to show that a path $( v, x, w)$ in $\G^{\mathrm{init}}$ is a v-structure if and only if there exist cliques $C_1, C_2 \in \E^\U$ such that $v,x\in C_1$, $w\notin C_1$, $w,x\in C_2$, $v\notin C_2$, and every clique in $\E^\U$ containing $v$ or $w$ also contains $x$.
	The `if' direction follows from Line~2 of Algorithm~\ref{alg:init-cpdag}: in case \(v\) and \(w\) are not adjacent in \(\U\), then the (induced) path is directed into a v-structure, and the property that every clique containing \(v\) or \(w\) also contains \(x\) ensures that the v-structure remains after Line~3; in case \(v\) and \(w\) are adjacent in \(\U\), then the argument used in the proof of condition (iv) in Lemma~\ref{lemma:init-cpdag} can be applied---that is, there must be an \(m\) in \(C_1\) but not \(C_2\) and an \(m'\) in \(C_2\) but not \(C_1\) such that there are induced paths \((v,x,m')\) and \((w,x,m)\) in \(\U\) which (along with the path \((v,x,w)\)) become v-structures in \(\G^{\mathrm{init}}\).
	The `only if' direction follows from Line~3: every v-structure created by Line~2 corresponds to a path $(v, x, w)$ for which there exist cliques $C_1, C_2 \in \E^\U$ such that $v,x\in C_1$, $w\notin C_1$, $w,x\in C_2$, and $v\notin C_2$; then, if there exists a clique in $\E^\U$ containing \(v\) (resp.~\(w\)) but not containing \(x\), the edge between \(v\) (resp.~\(w\)) is bidirected and removed by Lines~2 and 3, so that \((v, x, w)\) is not a path in \(\G^{\mathrm{init}}\).
\end{proof}

The following corollary is immediate from Theorem~\ref{thrm:init}.

\begin{coro}
	\label{coro:bijection-U-G}
	There is a bijection between the set of UEC-representatives on \(n\) nodes \(\mathbb{U}^n\) and the set of CPDAGs \[ \mathbb{G}^n = \{\mathtt{init\_CPDAG}(\U) \mid {\U \in \mathbb{U}^n}\}; \] Namely, the map \[ f: \mathbb{U}^n \longrightarrow \mathbb{G}^n; \qquad f: \U \longmapsto \mathtt{init\_CPDAG}(\U) \] is a bijection.
\end{coro}

The bijection identified in Corollary~\ref{coro:bijection-U-G} naturally extends to a bijection between UEC-representatives and DAG-reductions.

\begin{lemma}
	\label{lemma:bijection between U and D}
	Let $\mathbb{D}^n$ denote the space of all the DAG-reductions obtained from the undirected graphs in $\mathbb{U}^n$.
	There is a bijection between the spaces $\mathbb{U}^n$ and $\mathbb{D}^n$.
\end{lemma}

\begin{proof}
	With the help of the bijection identified in Corollary~\ref{coro:bijection-U-G}, it suffices to show that the procedure for constructing \(\D^\U\) from \(\G = \mathtt{init\_CPDAG}(\U)\) in Definition~\ref{defn:dag-reduction} is injective.
	(Note that the surjection follows from the definitions of $\mathbb{U}^n$, $\mathbb{G}^n$ and $\mathbb{D}^n$ since \(\mathbb{D}^n\) is defined to be the image of \(\mathbb{G}^n\), the space of all CPDAGs of the Markov equivalence classes of maximal DAGs in the UECs with UEC-representatives in $\mathbb{U}^n$.)
	Injectivity follows from the fact that different Markov equivalence classes have different CPDAGs.
	These CPDAGs differ in either their chain components or their directed edges.
	If two distinct UECs $\{\U\}, \{\U^\prime\}$ have associated Markov equivalence classes of maximal DAGs with CPDAGs $\G = \mathtt{init\_CPDAG}(\U)$ and $\G^\prime = \mathtt{init\_CPDAG}(\U^\prime)$ having different chain components, then the DAG-reductions $\D$ and $\D^\prime$ will have different vertex sets and hence be distinct.

	Now consider the case that $\D$ and $\D^\prime$ have the same vertex sets.
	The characterization of CPDAGs of Markov equivalence classes in \citep[Theorem~4.1]{AMP97} implies that the CPDAGs, $\G$ and $\G^\prime$ for the UECs $\{\U\}$ and $\{U^\prime\}$ must differ in their sets of directed edges.
	For the sake of contradiction, we suppose that $\G$ and $\G^\prime$ are distinct but that $\D = \D^\prime$.
	It follows that there must be two distinct vertices $\mathbf{v}, \mathbf{w}$ of $\D = \D^\prime$ for which there exist $v,v^\prime\in\mathbf{v}$ and $w,w^\prime\in\mathbf{w}$ such that $v\rightarrow w\in E^\G$ and $v^\prime\rightarrow w^\prime\in E^{\G^\prime}$.
	By Definition~\ref{defn:dag-reduction} this would allow $\G$ and $\G^\prime$ to have different directed edges while having $\D = \D^\prime$.
	However, we claim that $v^\prime\rightarrow w^\prime\in E^\G$ and $v\rightarrow w\in E^{\G^\prime}$.
	If not, we would have in $\G$ either the induced path $v\rightarrow w \mathdash w^\prime$ or the triangle $v\rightarrow w \mathdash w^\prime \mathdash v$.
	Since $\G$ is the CPDAG of a Markov equivalence class, the former case is not possible by the characterization of CPDAGs in \citep[Theorem~4.1]{AMP97}.
	Hence, we would have to be in the latter case.
	However, this case implies that $w,w^\prime\in cc_\G(v) = \mathbf{v}$.
	Thus, $\mathbf{v} = \mathbf{w}$, a contradiction to the assumption that $\mathbf{v}$ and $\mathbf{w}$ are distinct vertices of $\D = \D^\prime$.
	Therefore, it must be the case that $v^\prime\rightarrow w^\prime\in E^\G$ and $v\rightarrow w\in E^{\G^\prime}$.
	Since $\mathbf{v}$ and $\mathbf{w}$ were chosen arbitrarily, it follows that $\G$ and $\G^\prime$ have the same set of directed edges.
	Since $\D$ and $\D^\prime$ have the same vertex sets, it also follows that $\G$ and $\G^\prime$ have the same chain components.
	Hence, $\G = \G^\prime$, a contradiction.
	Thus, we conclude that injectivity holds.
\end{proof}

Now that we have established a bijection between UEC-representatives and DAG-reductions, we translate the moves for traversing UEC-representatives developed in Sections~\ref{sec:grobner} and~\ref{sec: traversing} to moves on DAG-reductions.
This allows for a more efficient implementation of search algorithms based on these moves, presented in Section~\ref{sec:grues:-markov-chain}.

\subsection{Traversing the space of DAG-reductions.}
\label{sec:pseud-its-equiv}

The bijection in Lemma~\ref{lemma:bijection between U and D} means that exploring the space of nonempty UECs is analogous to exploring the space of their corresponding DAG-reductions.
Thus it is natural to look at the DAG-reduction analogue of the moves defined in Section~\ref{sec: traversing}.
We define the \textit{merge} and \textit{split} operations on DAG-reductions with Algorithms~\ref{alg:reduced-merge} and \ref{alg:reduced-split}.

\begin{algorithm}
	\Input{DAG-reduction \(\D\) of UEC-representative \(\U\) with \(\alpha(\U) = k\)}
	\Output{DAG-reduction \(\D'\)of UEC-representative \(\U'\) with \(\alpha(\U) = k-1\)}
	\BlankLine
	\(\D' \coloneqq \D\)\;
	pick source nodes \(\mathbf{s}, \mathbf{s}' \in V^{\D'}\) such that \(|\mathbf{s}| = |\mathbf{s}'| = 1\) and \(\ch_{\D'}(\mathbf{s}) = \ch_{\D'}(\mathbf{s}')\)\;\label{alg:reduced-merge:pick}
	\(\mathbf{s} \coloneqq \mathbf{s} \cup \mathbf{s}'\)\;
	\For{\(\mathbf{c}' \in \ch_{\D'}(\mathbf{s}) \cap \ch_{\D'}(\mathbf{s}')\) such that \(|\pa_{\D'}(\mathbf{c}')| = 2\)}{
		\(\mathbf{s} \coloneqq \mathbf{s} \cup \mathbf{c}'\)\;
		remove node \(\mathbf{c}'\) from \(V^{\D'}\)\;}
	remove node \(\mathbf{s}'\) from \(V^{\D'}\)\;
	\Return \(\D'\)
	\caption{\(\mathtt{merge}(\D)\)}
	\label{alg:reduced-merge}
\end{algorithm}

\begin{algorithm}
	\Input{DAG-reduction \(\D\) of UEC-representative \(\U\) with \(\alpha(\U) = k\)}
	\Output{DAG-reduction \(\D'\)of UEC-representative \(\U'\) with \(\alpha(\U) = k+1\)}
	\BlankLine
	\(\D' \coloneqq \D\)\;
	pick source node \(\mathbf{s} \in V^{\D'}\) such that \(|\mathbf{s}| \geq 2\)\;\label{alg:reduced-split:pick-source}
	pick \(v, w \in \mathbf{s}\)\;\label{alg:reduced-split:pick-nodes}
	add node \(\mathbf{v} \coloneqq \{v\}\) to \(V^{\D'}\)\; \label{alg:reduced-split:make-node}
	\(\mathbf{s} \coloneqq \mathbf{s}\setminus\{v\}\)\;
	\If{\(|\mathbf{s}| = 1\)}{
		\For{\(\mathbf{c} \in \ch_{\D'}(\mathbf{s})\)}{
			add edge \(\mathbf{v} \rightarrow \mathbf{c}\) to \(E^{\D'}\)\;}
	}
	\Else{add node \(\mathbf{w} \coloneqq \{w\}\) to \(V^{\D'}\)\;\label{alg:reduced-split:add}
		\(\mathbf{s} \coloneqq \mathbf{s}\setminus\{w\}\)\;
		\For{\(\mathbf{c} \in \ch_{\D'}(\mathbf{s})\)}{
			add edge \(\mathbf{w} \rightarrow \mathbf{c}\) to \(E^{\D'}\)\;}
		add edges \(\mathbf{w} \rightarrow \mathbf{s},\ \mathbf{v} \rightarrow \mathbf{s}\) to \(E^{\D'}\)\;}
	\Return \(\D'\)
	\caption{\(\mathtt{split}(\D)\)}
	\label{alg:reduced-split}
\end{algorithm}

\begin{algorithm}
	\Input{DAG-reduction \(\D\) of UEC-representative \(\U\)}
	\Output{DAG-reduction \(\D'\) of UEC-representative \(\U'\) with \(\alpha(\U) = \alpha(\U')\) but different monomial representation}
	\Parameter{\(\mathtt{fiber} \in \{\mathtt{within}, \mathtt{out\_add}, \mathtt{out\_del}\}\), whether the algebraic move is within the fiber, out of the fiber by adding edges to \(\U\), or out of the fiber by deleting edges from \(\U\)}
	\BlankLine
	\(\D' \coloneqq \D\)\;
	\If{\(\mathtt{fiber} = \mathtt{out\_del}\)}{
		pick child node \(\mathbf{t} \in V^{\D'}\)\;
		pick node \(\mathbf{s} \in \ma_{\D'}(\mathbf{t})\)\;
		\(\mathbf{T} \coloneqq \ma_{\D'}(\mathbf{t}) \setminus \{\mathbf{s}\}\)\;}
	\Else{
		pick source nodes \(\mathbf{s}, \mathbf{s}' \in V^{\D'}\) such that \(|\mathbf{s}| \geq 2\) or \(\ch_{\D'}(\mathbf{s})\setminus\ch_{\D'}(\mathbf{s}') \not= \emptyset\)\;
		pick \(\mathbf{t} \in \{\mathbf{s}\} \cup (\ch_{\D'}(\mathbf{s})\setminus\ch_{\D'}(\mathbf{s}'))\) such that \(|\mathbf{t}| \geq 2\) or \(\pa_{\D'}(\mathbf{t}) \not= \emptyset\)\;
		\If{\(\mathtt{fiber} = \mathtt{out\_add}\)}{
			\(\mathbf{T} \coloneqq \{\mathbf{s'}\} \cup \ma_{\D'}(\mathbf{t})\)\;}
		\Else(\tcp*[h]{\(\mathtt{fiber} = \mathtt{within}\)}){\(\mathbf{T} \coloneqq \{\mathbf{s'}\} \cup \ma_{\D'}(\mathbf{t})\setminus\{\mathbf{s}\}\)\;}
	}
	pick \(v \in \mathbf{t}\)\;
	\If{there exists \(\mathbf{t}' \in V^{\D'}\) such that \(\ma_{\D'}(\mathbf{t}') = \mathbf{T}\)}{
		\(\mathbf{t}' \coloneqq \mathbf{t}' \cup \{v\}\)\;
	}
	\Else{
		\(\mathbf{t}' \coloneqq \{v\}\)\;
		\(\mathbf{P} \coloneqq \{\mathbf{p} \in V^{\D'} : \ma_{\D'}(\mathbf{p}) \subsetneq \mathbf{T}\}\)\;
		\(\mathbf{C} \coloneqq \{\mathbf{c} \in V^{\D'} : \ma_{\D'}(\mathbf{c}) \supsetneq \mathbf{T}\}\)\;
		\For{\(\mathbf{p} \in \mathbf{P}\)}{
			add edge \(\mathbf{p} \rightarrow \mathbf{t}'\) to \(E^{\D'}\)\;
		}
		\For{\(\mathbf{c} \in \mathbf{C}\)}{
			add edge \(\mathbf{t}' \rightarrow \mathbf{c}\) to \(E^{\D'}\)\;
		}
	}
	\(\mathbf{t} \coloneqq \mathbf{t} \setminus \{v\}\)\;
	\If{\(\mathbf{t} = \emptyset\)}{
		delete node \(\mathbf{t}\) from \(V^{\D'}\)\;}
	\Return \(\D'\)
	\caption{\(\mathtt{algebraic}(\D)\)}
	\label{alg:reduced-algebraic}
\end{algorithm}

Analogous to the operations \textit{merge} and \textit{split} defined for UEC-representatives, the DAG-reduction versions are inverse operations.

\begin{lemma}
	\label{lemma:merge-split alg inverse}
	The operations $\mathtt{merge}$ and $\mathtt{split}$ introduced in Algorithms \ref{alg:reduced-merge} and \ref{alg:reduced-split} are inverse operations; i.e., for any DAG $\D$, \[ \mathtt{merge}(\mathtt{split}(\D))=\mathtt{split}(\mathtt{merge}(\D))=\D \] for suitable choices of source nodes in Line~\ref{alg:reduced-merge:pick} of Algorithm~\ref{alg:reduced-merge} and \ref{alg:reduced-split}.
\end{lemma}

\begin{proof}
	Let $\D_{\mathtt{split}(\D)}$ be the output of Algorithm~\ref{alg:reduced-split} $\mathtt{split}(\D)$, and consider the node $\mathbf{s}\in V^{\D_{\mathtt{split}(\D)}}$ selected in Line~\ref{alg:reduced-split:pick-source}, the node $\mathbf{v}\in V^{\D_{\mathtt{split}(\D)}}$ constructed in Line~\ref{alg:reduced-split:make-node} and $\mathbf{w}$ that may or may not be a node of $\D_{\mathtt{split}(\D)}$, introduced in Line~\ref{alg:reduced-split:add} of Algorithm~\ref{alg:reduced-split}.
	We will first show that the output of $\mathtt{merge}(\D_{\mathtt{split}(\D)})$ is $\D$, whenever the node $\mathbf{s'}$ we select in Line~\ref{alg:reduced-merge:pick} of Algorithm~\ref{alg:reduced-merge} is equal to $\mathbf{v}\in V^{\D_{\mathtt{split}(\D)}}$ selected in Algorithm~\ref{alg:reduced-split}, and the node $\mathbf{s}$ we select in Algorithm~\ref{alg:reduced-merge} is equal to $\mathbf{w}$, if $\mathbf{w}$ is indeed a node of $\D_{\mathtt{split}(\D)}$, and $\mathbf{s}$ otherwise.
	To show this, we will prove (i) that the each of the chain components $\mathbf{s}, \mathbf{s'}$ identified above consist of a single node in $V^\U$ and (ii) that their children sets in $\D_{\mathtt{split}(\D)}$ coincide.

	Notice that (i) follows by lines 4, 6 and 10 of Algorithm~\ref{alg:reduced-split}.
	For (ii), if $|\mathbf{s}|=1$ then $\ch_{\D_{\mathtt{split}(\D)}}(\mathbf{v})$ is defined as $\ch_{\D_{\mathtt{split}(\D)}}(\mathbf{s})$ in Lines 7 and 8 of Algorithm~\ref{alg:reduced-split}.
	Otherwise, the node $\mathbf{w}$ is added in $V^{\D_{\mathtt{split}(\D)}}$ in Line 10, and in Line 14 we have that $\ch_{\D_{\mathtt{split}(\D)}}(\mathbf{v})= \ch_{\D_{\mathtt{split}(\D)}}(\mathbf{w})=\{\mathbf{s}\}$.
	Therefore, Algorithm~\ref{alg:reduced-merge} can indeed be applied for the selected sets $\mathbf{s},\mathbf{s'}$ described above.

	Let $\D'$ be the output of Algorithm~\ref{alg:reduced-merge} applied on $\D_{\mathtt{split}(\D)}$.
	To see that $\D$ and $\D'$ coincide, notice that $V^{\D_{\mathtt{split}(\D)}}=V^\D \cup \{\mathbf{v}\}$, if $|\mathbf{s}|=1$ in Line 6 of Algorithm~\ref{alg:reduced-split}, and $V^{\D_{\mathtt{split}(\D)}}=V^\D \cup \{\mathbf{v},\mathbf{w}\}$, otherwise.
	In the first case, the node $\mathbf{s'}$ is removed in Line 7 of Algorithm \ref{alg:reduced-merge}.
	Hence $V^{\D_{\mathtt{split}(\D)}}=V^{\D'}$, since no other node gets removed in Line 4 of Algorithm \ref{alg:reduced-merge}.
	Moreover, by Lines 4, 5 of Algorithm~\ref{alg:reduced-split} and Line 3 (Line 4 does not affect $\mathbf{s}$ in this case) of Algorithm~\ref{alg:reduced-merge}, the nodes in the chain component $\mathbf{s}$ are the same in both DAGs.
	In the second case, if $\mathbf{w}\in V^{\D_{\mathtt{split}(\D)}}$, then the node $\mathbf{v}=\mathbf{s'}$ gets removed in Line 7 of Algorithm~\ref{alg:reduced-merge}, while the node $\mathbf{w}$ gets removed in Line 6 of Algorithm \ref{alg:reduced-merge} as a common child of $\mathbf{s}$ and $\mathbf{s'}$.
	All edges added through Algorithm \ref{alg:reduced-split} on $\D$ are attached to one node that gets removed when applying Algorithm \ref{alg:reduced-merge} to $\D_{\mathtt{split}(\D)}$.
	Hence, we have that $E^{\D'}=E^\D$.

	Let us now consider the DAGs arising by applying $\mathtt{merge}$ and $\mathtt{split}$ in the opposite order.
	Let $\D_{\mathtt{merge}(\D)}$ denote the output of Algorithm \ref{alg:reduced-merge} on $\D$ and let $\D'$ be the output of Algorithm \ref{alg:reduced-split} on $\D_{\mathtt{merge}(\D)}$.
	The choice of nodes for Lines 2, 3 in Algorithm \ref{alg:reduced-split} is as follows.
	Firstly, we select $\mathbf{s}$ as in the output $\D_{\mathtt{merge}(\D)}$ of Algorithm \ref{alg:reduced-merge}; Lines 2 and 3 of Algorithm \ref{alg:reduced-merge} ensure that $|\mathbf{s}|\geq 2$.
	Furthermore, we select $\mathbf{v} $ so as to have $\mathbf{s'}=\{v\}$ in $\D$; note that this is possible by Lines 2 and 3 of Algorithm~\ref{alg:reduced-merge}.
	Similarly, if $|\mathbf{s}|>2$ then we pick $w$ so as to have $\mathbf{s}=\{w\}$ in $\D$; otherwise the selection of $w$ is not restricted.
	Comparing the sets $V^\D$ and $V^{\D'}$ as well as $E^{\D}$ and $E^{\D'}$ as before, we get equality between $\D$ and $\D'$.
\end{proof}

Just as in the case of the \textit{within-fiber} move for UEC-representatives, the DAG-reduction version serves as its own inverse.
The DAG-reduction analogues of \texttt{out\_add} and \texttt{out\_del} are inverse operations.

\begin{lemma}
	\label{lemma:within-inverse alg}
	The operation $\mathtt{within}$ introduced in Algorithm~\ref{alg:reduced-algebraic} (namely, \(\mathtt{algebraic}\) applied with parameter \(\mathtt{fiber} \coloneqq \mathtt{within}\)) is its own inverse, for suitable choice of source nodes on the second application, i.e., $\D = \mathtt{within}(\mathtt{within}(\D))$.
\end{lemma}

\begin{proof}
	Let $\D'$ be the output of $\mathtt{within}(\D)$, and consider the corresponding nodes $\mathbf{s}_\D, \mathbf{s'}_\D$ and $\mathbf{t}_\D$ picked in Lines 7 and 8 and the set \(\mathbf{T}_\D\) in Line 12 of Algorithm~\ref{alg:reduced-algebraic}.
	Then \(\D\) is recovered by applying $\mathtt{within}(\D')$ and picking \(\mathbf{s}_{\D'} = \mathbf{t}_{\D'} = \mathbf{t'}_\D\) so that \(\mathbf{T}_{\D'}= \ma(\mathbf{t}_\D)\), i.e., the first application of the algorithm moved \(v \in \mathbf{t}_\D\) to \(\mathbf{t'}_\D=\mathbf{t}_{\D'}\) (possibly creating the node if it did not already exist in \(\D\) and possibly deleting the node \(\mathbf{t}_\D\) if it had cardinality 1) and the second application moved it back (again possibly creating the node if it did not already exist in \(\D'\) possibly deleting the node \(\mathbf{t}_\D'\) if it had cardinality 1).
\end{proof}

\begin{lemma}
	\label{lemma:fiber add-delete alg inverse}
	The operations $\mathtt{out\_add}$ and $\mathtt{out\_del}$ introduced in Algorithm~\ref{alg:reduced-algebraic} (namely, \(\mathtt{algebraic}\) applied with parameter \(\mathtt{fiber}\) respectively set to \(\mathtt{out\_add}\) or \(\mathtt{out\_del}\)) are inverse operations; i.e., for any DAG $\D$, \[ \mathtt{out\_add}(\mathtt{out\_del}(\D))=\mathtt{out\_del}(\mathtt{out\_add}(\D))=\D \] for suitable choices of source nodes in the second application.
\end{lemma}

\begin{proof}
	This proof follows a similar argument to the previous proof, because all three moves share the same basic operations defined in Lines 13--26 and differ only in how they construct the set \(\mathbf{T}\).
	The move \(\D' = \mathtt{out\_del}(\D)\) moves a node \(v\) from its (non-source) chain component \(\mathbf{t}\) into a \(\mathbf{t}'\) whose maximal ancestor set differs from \(\mathbf{t}\) only by \(\mathbf{t}' = \mathbf{t}\setminus\mathbf{s}\) (i.e., the only change in the corresponding undirected graph is to delete edges between \(v\) and \(\mathbf{s}\)).
	The move $\D = \mathtt{out\_add}(\D')$ recovers \(\D\) by (like in the previous proof) picking \(\mathbf{s}_{\D'} = \mathbf{t}_{\D'} = \mathbf{t}'_{\D}\) and \(\mathbf{s}'_{\D'} = \mathbf{s}_{\D}\), which moves \(v\) back to \(\mathbf{t}_{\D}\) (i.e., the only change in the corresponding undirected graph is to add edges between \(v\) and \(\mathbf{s}_{\D}\)).
	\(\D' = \mathtt{out\_del}(\mathtt{out\_add}(\D))\) follows similarly.
\end{proof}

Now that we have defined DAG-reduction analogues of our moves on UEC-representatives, we show that they are indeed equivalent to the corresponding moves defined in Sections~\ref{sec:grobner} and~\ref{sec: traversing}.
In order to do this, we first characterize some structural properties of CPDAGs and DAG-reductions of UEC-representatives in terms of the monomial representations of the UEC-representatives.

\begin{lemma}
	\label{lemma:characterizing the edges of the CPDAG}
	Let $\U$ be a UEC-representative and $x_{i_1|A_1}x_{i_2|A_2}\cdots x_{i_k|A_k}$ be any monomial representation of $\U$.
	Let $A_j'=A_j\cup i_j$ for $j\in [k]$.
	The CPDAG $\G:=\mathtt{init\_CPDAG}(\U)$ of $\U$, has the following properties:
	\begin{enumerate}
		\item An edge $a\mathdash b$ in $\U$ is an undirected edge of $\G$ if and only if $a$ and $b$ are such that $a$ lies in $A_j'$ if and only if $b$ lies in $A_j'$, for any $j'\in [k]$.
		\item An edge $a \mathdash b$ in $\U$ gets bidirected in Algorithm~\ref{alg:init-cpdag} (and hence gets removed) if and only if $a\in A_{j_1}\cap A_{j_2}$ and $b\in A_{j_2}\cap A_{j_3}$, for $j_1, j_2, j_3\in [k]$ such that $j_1\neq j_3$.
		\item An edge $a \mathdash b$ in $\U$ is directed as $a\rightarrow b$ in $\G$ if and only if $\ngh[a]\subsetneq \ngh[b]$.
	\end{enumerate}
\end{lemma}
\begin{proof}
	For (1), the only way that the edge $a\mathdash b$ forms an induced 2-path (or gets bidirected) is if there exists a vertex $c$ which is adjacent to $a$ and not $b$ or vice versa.
	However, as any vertex that is adjacent to $a$ is also adjacent to $b$ (as both $a$ and $b$ lie in same $A_j's$), $a\mathdash b$ remains undirected in $\G$.

	Similarly, if $a \mathdash b$ is undirected in $\G$, then a vertex is adjacent to $a$ if and only if it is also adjacent to $b$.
	This implies that $a$ lies in $A'_j$ if and only if $b$ lies in $A'_j$.

	For (2), if $a\in A_{j_1}\cap A_{j_2}$ and $b\in A_{j_2}\cap A_{j_3}$, then there exist some $c \in A'_{j_1}$ and $d\in A'_{j_3}$ such that $c$ is adjacent to $a$ and not $b$, and $d$ is adjacent to $b$ and not $a$.
	This implies that $\{c\mathdash a \mathdash b\}$ and $\{a\mathdash b \mathdash d\}$ form induced 2-paths, which makes $a\mathdash b$ bidirected.

	For (3), Note that here $a\in \ngh[a]$.
	As $\ngh[a]\subsetneq \ngh[b]$, there exists some $c\in \ngh[b]\setminus \ngh[a]$ such that $c$ is adjacent to $b$ but not $a$.
	This forms the induced 2-path $\{a\mathdash b \mathdash c\}$, which directs the edge $a\mathdash b$ in $\U$ as $a\rightarrow b$ in $\G$.
	On the other hand, there does not exist any $d$ which is adjacent to $a$ and not $b$, which could create an induced 2-path of the form $\{ d\mathdash a \mathdash b\}$.
	Hence, the edge $a\mathdash b$ gets directed as $a\rightarrow b$ but does not get bidirected.
	The argument can be followed in the reverse direction for the `only if' part.
\end{proof}

We use the properties stated in Lemma~\ref{lemma:characterizing the edges of the CPDAG} to express the DAG-reduction of $\U\in \mathbb{U}^n$ as a lattice.
Note that this structure is independent of the choice of monomial representation of $\U$.
In the following corollary, we use the term \textit{level} of a node to refer to the rank of the node in the transitive reduction of the lattice.

\begin{coro}
	\label{coro:reduced DAG lattice}
	Let $\U$ be a UEC-representative and $x_{i_1|A_1}x_{i_2|A_2}\cdots x_{i_k|A_k}$ be any monomial representation of $\U$.
	Also let $A_j'=A_j\cup i_j$ for $j\in [k]$ and $\D$ be the corresponding DAG reduction.
	Then $\D$ is equivalent to a lattice with at most $k$ levels where each level has the following properties:
	\begin{enumerate}

		\item The first level contains $k$ source nodes, each a chain component of the form $\{i_m\cup A_m\setminus \cup_{j\neq m}A_j\}$, for $m\in [k]$.
		\item More generally, the $r$-th level can contain at most $k\choose r$ nodes where each node is a chain component of the form $\{\cap_{l\in L_r}
			      A_l\setminus \cup_{j\in [k]\setminus L_r} A_j\}$, where $L_r\subsetneq [k]$ is any arbitrary set with cardinality $r$.
		\item In particular, the $k$-th level can have at most one node which is the chain component \(\{\cap_{j\in [k]}
		      A_j\}\).
	\end{enumerate}
	Further, there are no edges between any two nodes within the same level and the edges connecting two nodes in different levels are directed downwards.
\end{coro}
\begin{proof}
	Points $1$, $2$ and $3$ follow from Lemma~\ref{lemma:characterizing the edges of the CPDAG}~(1) that two nodes $a$ and $b$ of $\U$ cannot lie in different $A_j's$ if $a \mathdash b$ is undirected in the CPDAG $\G$.
	As the skeleton of $\D$ is a subgraph of $\U$, combining the undirected edges to form the chain components (i.e., the nodes of $\D$) and deleting the bidirected edges in Algorithm~\ref{alg:init-cpdag} leaves us with edges connecting any two nodes $\mathbf{a}$ and $\mathbf{b}$ in different levels of $\D$ whenever $\mathbf{a}$ and $\mathbf{b}$ are of the form $\cap_{l\in L_r}A_l\setminus \cup_{j\in [k]\setminus L_r}A_j$ and $\cap_{m\in L_m}A_m\setminus \cup_{j\in [k]\setminus L_m}A_j$ with $L_r\cap L_m\neq \emptyset$.
	However, if $r<m$, then for any vertex $a$ and $b$ of $\U$ with $a\in \mathbf{a}$ and $b\in \mathbf{b}$, $\ngh[a]\subsetneq \ngh[b]$.
	Thus, the edges are directed downwards in $\D$.
\end{proof}

The advantage of having the properties mentioned in Corollary~\ref{coro:reduced DAG lattice} is that we can now directly construct the DAG-reduction of a UEC-representative $\U$ from any of its monomial representations, without having to apply Algorithm~\ref{alg:init-cpdag} on $\U$.
We explain this construction with an example.

\begin{ex}
	\label{ex:lattice structure}
	In Figure \ref{figure:complete reduced DAG}, we give a complete structure of all the possible nodes of a DAG reduction corresponding to any $\U \in \mathbb{U}^n_3$ with the representation $x_{i_1|A_1}x_{i_2|A_2}x_{i_3|A_3}$.
	Now, let $\U$ be the UEC-representative with monomial representation $x_{i_1|A_1}x_{i_2|A_2}x_{i_3|A_3}=x_{1|2}x_{3|24}x_{5|46}$, as seen in Example \ref{ex:reduced DAG}.
	Figure \ref{figure:U-CPDAP-RDAG} gives us the corresponding CPDAG and DAG reduction of $\U$.
	Observe that the edge $5\mathdash 6$ is undirected in the CPDAG as both $5$ and $6$ lie only in the clique $\{4,5,6\} \supseteq i_3\cup A_3 \setminus A_1\cup A_2$ and in no other clique.
	Similarly, the chain components in the second level of the DAG reduction corresponds to the terms $\{(A_1\cap A_2)\setminus A_3\}$ and $\{(A_2\cap A_3)\cap A_1\}$, respectively.
	As $A_1\cap A_2 \cap A_3 = \emptyset$, we do not have any chain component in the third level.

	\begin{figure}
		\begin{tikzpicture}
			\filldraw[black]
			(4,0) circle [radius=.04] node [below] {\small{$\{\mathbf{A_1\cap A_2 \cap A_3\}}$}}
			(2.5,1.5)circle [radius=.04] node [left] {\small{$\{\mathbf{A_1\cap A_2\setminus A_3}\}$}}
			(4,1.5)circle [radius=.04] node [above] {}
			(6.4,1)circle [radius=0] node [below]
				{\small{$\{\mathbf{A_1\cap A_3\setminus A_2}\}$}}
			(5.5,1.5)circle [radius=.04] node [right] {\small{$\{\mathbf{A_2\cap A_3\setminus A_1}\}$}}
			(2.5,3)circle [radius=.04] node [left]
				{\small{$\{\mathbf{i_1\cup A_1\setminus A_2\cup A_3}\}$}}
			(4,3)circle [radius=.04] node [above] {}
			(4,3.2)circle [radius=0] node [above]
				{\small{$\{\mathbf{i_2\cup A_2\setminus A_1\cup A_3}\}$}}
			(5.5,3)circle [radius=.04] node [right]
				{\small{$\{\mathbf{i_3\cup A_3\setminus A_1\cup A_2}\}$}};

			\draw
			(4,0)--(2.5,1.5)
			(4,0)--(4,1.5)
			(4,0)--(5.5,1.5)
			(4,0)--(2.5,3)
			(4,0) .. controls (3.5,1.5) .. (4,3)
			(4,0)--(5.5,3)
			(2.5,3)--(2.5,1.5)
			(2.5,3)--(4,1.5)
			(4,3)--(2.5,1.5)
			(4,3)--(5.5,1.5)
			(5.5,3)--(4,1.5)
			(5.5,3)--(5.5,1.5);

			\draw[dotted]
			(4,1.5)--(5.5,1);
			\draw
			(2.35,2.2)--(2.5,2)
			(2.5,2)--(2.65,2.2)
			(5.35,2.2)--(5.5,2)
			(5.5,2)--(5.65,2.2)
			(3,1)--(3,1.2)
			(3,1)--(2.8,1)
			(3.85,0.8)--(4,0.6)
			(4,.6)--(4.15,0.8)
			(5,1)--(5,1.2)
			(5,1)--(5.2,1)
			(3,2.5)--(3,2.7)
			(3,2.5)--(2.8,2.5)
			(5,2.5)--(5,2.7)
			(5,2.5)--(5.2,2.5)
			(3.5,2.5)--(3.5,2.7)
			(3.5,2.5)--(3.7,2.5)
			(4.5,2.5)--(4.5,2.7)
			(4.5,2.5)--(4.3,2.5)
			(3.5,1)--(3.5,1.2)
			(3.5,1)--(3.35,1.05)
			(4.5,1)--(4.5,1.2)
			(4.5,1)--(4.65,1.05)
			(3.73,2.2)--(3.73,2.4)
			(3.73,2.2)--(3.85,2.3)
			;
		\end{tikzpicture}
		\caption{A complete description of a DAG reduction of $\U$ in $\mathbb{U}^n_3$.}\label{figure:complete reduced DAG}
	\end{figure}
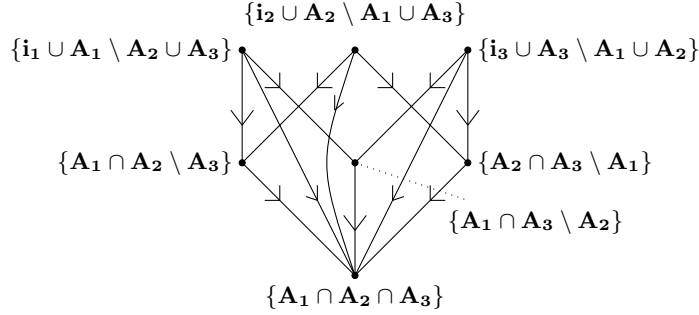
\end{ex}

Now that we have all the necessary tools, we prove the equivalence among the two versions of moves.
Note that as we have already shown in Lemma \ref{lemma:merge-split alg inverse} that $\mathtt{merge}$ and $\mathtt{split}$ are inverse operations, showing the equivalence between \textit{merge} and $\mathtt{merge}$ is enough to show the equivalence between \textit{split} and $\mathtt{split}$.
Similarly, as $\mathtt{out\_add}$ and $\mathtt{out\_del}$ are inverse operations by Lemma \ref{lemma:fiber add-delete alg inverse}, equivalence of \textit{out-of-fiber-delete} and $\mathtt{out\_del}$ implies equivalence of \textit{out-of-fiber-add} and $\mathtt{out\_add}$.

\begin{lemma}
	\label{lemma:merge, split equivalent}
	The moves \textit{merge} and \textit{split} defined on any UEC-representative $\U$ in Definition~\ref{defn: merge and split} are equivalent to the moves defined in Algorithm~\ref{alg:reduced-merge} and \ref{alg:reduced-split} on the corresponding DAG-reduction $\D$ of $\U$.
\end{lemma}
\begin{proof}

	Merge: Let $\U$ be the UEC-representative on which we apply \textit{merge} and $\U'$ be the result.
	Then it is clear from the definition of \textit{merge} that there exist monomial representations of $\U$ and $\U'$ which are of the form $x_{i_1|A_1}x_{i_2|A_1}x_{i_3|A_3}\cdots x_{i_k|A_k}$ and $x_{i_1|A_1\cup i_2}x_{i_3|A_3}\cdots x_{i_k|A_k}$, respectively.
	In order to show that applying \emph{merge} on $\U$ to get $\U'$ is equivalent to applying Algorithm~\ref{alg:reduced-merge} on $\D$ to get $\D'$, we need to show that the chain components of $\D'$ obtained from the monomial representation of $\U'$ by using Corollary \ref{coro:reduced DAG lattice} are equal to the chain components obtained from the output of Algorithm~\ref{alg:reduced-merge}.

	Using the above representations for $\U$ and $\U'$, we note that the chain components of $\D$ containing $i_1$ and $i_2$ are $\{i_1\cup A_1\setminus \cup_{j\in [k]} A_j\}=\{i_1\}$ and $\{i_2\cup A_1\setminus \cup_{j\in [k]} A_j\}=\{i_2\}$, respectively.
	As both chain components have cardinality $1$ and have the same children (i.e., $A_1$), we can pick $\mathbf{s}$ as $\{i_1\}$ and $\mathbf{s'}$ as $\{i_2\}$ in Algorithm~\ref{alg:reduced-merge}.
	Now, $i_2$ does not lie in $A_j$ for any $j \in [k]$.
	Thus, there is no change between the chain components of $\D$ and $\D'$ caused by changing $A_1$ to $A_1\cup i_2$, except the source chain component $\{i_1\}=\{i_1\cup A_1\setminus \cup_{j\neq 1} A_j\}$.
	This chain component gets changed to $\{i_1\cup A_1\cup i_2\setminus \cup_{j\neq 1} A_j\}$ in $\D'$, which is the same as adding the chain component $\{A_1\setminus \cup_{j\neq 1} A_j\}$ of $\D$ to the set $\{i_1\cup i_2\}$.

	It is clear that we delete the chain component $\{i_2\}$ as it moves to the chain component $\{i_1\cup A_1\cup i_2\setminus \cup_{j\neq 1} A_j\}$.
	It follows that the chain component $\mathbf{c'}$ in $\D$ whose parent set is $\{i_1,i_2\}$ (which is $\{A_1\setminus \cup_{j\neq 1}A_j\}$ here) is deleted as well.
	Further, any chain component of $\D$ which is a child of $\{i_1\},\{i_2\}$ and some other source chain component containing $i_r$ (i.e., when the cardinality of the parent set is more than 2) is of the form $\{A_1 \cap A_r\setminus \cup_{j\neq 1,r} A_j\}$.
	These remain unchanged when $A_1$ is updated to $A_1\cup i_2$, which is also in accordance with the algorithm.

	Now, in order to show the equivalence, we also need to show that if $\D$ and $\D'$ are the input and output of Algorithm~\ref{alg:reduced-merge} respectively, then there exist monomial representations of the UECs $\U$ and $\U'$ corresponding to $\D$ and $\D'$ such that they are connected by a single \textit{merge} operation.
	We know from Corollary \ref{coro:reduced DAG lattice} that a chain component $\mathbf{s}$ is a child the of source chain component containing $i_1$ if and only if $\{s\}\subseteq A_1$.
	Thus, picking source chain components $\mathbf{s}$ and $\mathbf{s'}$ from $\D$ such that $\ch(\mathbf{s})=\ch(\mathbf{s'})$ along with $|\mathbf{s}|=|\mathbf{s'}|=1$ implies that there exists a monomial representation of $\U$ of the form $x_{i_1|A_1}x_{i_2|A_1}x_{i_3|A_3}\cdots x_{i_k|A_k}$.
	Now, the source chain component $\mathbf{s}$ getting updated to $\mathbf{s} \cup \mathbf{s'} \cup \mathbf{c'}$ implies that the set $\{i_1\}$ gets changed to $\{i_1\cup i_2\cup A_1\setminus \cup_{j\neq 1} A_j\}$.
	This along with that fact that the chain component $\{A_1\setminus \cup_{j\neq 1} A_j\}$ is no longer present in $\D'$ (and the fact that no other chain components of $\D$ are changed by the Algorithm) implies that there exists a monomial representation of $\U'$ which is of the form $x_{i_1|A_1\cup i_2}x_{i_3|A_3}\cdots x_{i_k|A_k}$.
	Thus, $\U$ and $\U'$ are connected by a single merge operation as desired.

	Split: We have already seen in Proposition \ref{prop:merge split inverse} that the operations \textit{merge} and \textit{split} on UEC-representatives are inverse of each other.
	Similarly, we also proved in Lemma \ref{lemma:merge-split alg inverse} that the $\mathtt{merge}$ and $\mathtt{split}$ operations in Algorithm~\ref{alg:reduced-merge} and $\ref{alg:reduced-split}$ are inverse of each other as well.
	Further, by Lemma \ref{lemma:bijection between U and D}, we know that there exists a unique DAG reduction $\D$ for each UEC-representative $\U$.
	Combining these results gives us that the operation \textit{split} on any UEC-representative $\U$ is equivalent to the $\mathtt{split}$ operation on the corresponding DAG reduction $\D$ if and only the operation \textit{merge} on $\U$ is equivalent to the $\mathtt{merge}$ operation on $\D$.
	Thus, the equivalence for the \textit{split} operation follows from the equivalence for \textit{merge}, which is proved above.
	(Note: Figure \ref{figure:commutative diagram} is a commutative diagram for the preceding argument.)
\end{proof}
\begin{figure}
	\begin{tikzpicture}
		\filldraw[black]
		(2.4,.15) circle [radius=0] node [left] {$\D$}
		(2.4,2.15) circle [radius=0] node [left] {$\U$}
		(4.6,.15) circle [radius=0] node [right] {$\D'$}
		(4.6,2.15) circle [radius=0] node [right] {$\U'$}
		(3.4,.2) circle [radius=0] node [above] {\small{$\mathtt{merge}$}}
		(3.4,2.2) circle [radius=0] node [above] {\small{\textit{merge}}}
		(3.4,0) circle [radius=0] node [below] {\small{$\mathtt{split}$}}
		(3.4,2) circle [radius=0] node [below] {\small{\textit{split}}};

		\draw
		(2.2,0.4)--(2.2,1.8)
		(4.8,0.4)--(4.8,1.8)
		(2.5,0)--(4.5,0)
		(2.5,.3)--(4.5,.3)
		(2.5,2)--(4.5,2)
		(2.5,2.3)--(4.5,2.3);

		\draw
		(2.1,1.7)--(2.2,1.8)
		(2.2,1.8)--(2.3,1.7)
		(4.7,1.7)--(4.8,1.8)
		(4.8,1.8)--(4.9,1.7)
		(2.1,0.5)--(2.2,0.4)
		(2.2,0.4)--(2.3,0.5)
		(4.7,0.5)--(4.8,0.4)
		(4.8,0.4)--(4.9,0.5);

		\draw
		(2.5,0)--(2.6,.1)
		(2.5,0)--(2.6,-.1)
		(4.5,.3)--(4.4,.4)
		(4.5,.3)--(4.4,.2)

		(2.5,2)--(2.6,2.1)
		(2.5,2)--(2.6,1.9)
		(4.5,2.3)--(4.4,2.4)
		(4.5,2.3)--(4.4,2.2);

	\end{tikzpicture}
	\caption{A commutative diagram explaining the bijections between $\U$, $\U'$, $\D$ and $\D'$.}\label{figure:commutative diagram}
\end{figure}
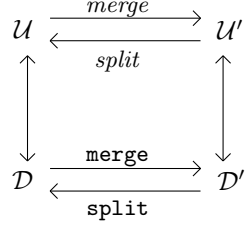

\begin{lemma}
	The within-fiber, out-of-fiber-add, and out-of-fiber-delete moves defined on any UEC-representative $\U$ are equivalent to the moves defined in Algorithm~\ref{alg:reduced-algebraic} on the corresponding DAG-reduction $\D$ of $\U$.
\end{lemma}
\begin{proof}
	Within-fiber: Let $\U$ be any arbitrary UEC-representative and $\U'$ be the graph obtained by applying a \textit{within-fiber} move on $\U$.
	Let $\D$ and $\D'$ be the corresponding DAG-reduction of $\U$ and $\U'$ respectively.
	Then, we first show that applying the \textit{within-fiber} move on $\U$ that gives $\U'$ is same as applying \(\mathtt{within}\) (i.e., Algorithm~\ref{alg:reduced-algebraic} with parameter \(\mathtt{fiber} \coloneqq \mathtt{within}\)) on $\D$ and get $\D'$.

	Let $x_{i_1|A_1}x_{i_2|A_2}\cdots x_{i_k|A_k}$ be a monomial representation of $\U'$.
	We apply the \textit{within-fiber} move $x_{i_1|A_1}x_{i_2|A_2}-x_{i_1|A_1\cup c}x_{i_2|A_2\setminus c}$ to reach $\U'$ with the representation $x_{i_1|A_1\cup c}x_{i_2|A_2\setminus c}\cdots x_{i_k|A_k}$.
	Note that $c$ lies in $A_2\setminus A_1$.
	Now, the chain components in $\D$ containing $i_1$ and $i_2$ are $\{i_1\cup A_1\setminus \cup_{j\neq 1} A_j\}$ and $\{i_2\cup A_2\setminus \cup_{j\neq 2}A_j\}$, respectively.
	So, if $c$ lies only in $A_2$ and no other $A_j$, then we know that the cardinality of the chain component $\{i_2\cup A_2\setminus \cup_{j\neq 2}A_j\}$ is at least $2$.
	Similarly, if $c$ does lie in some other $A_j$ with $j\neq 2$, then $\ch(\{i_2\cup A_2\setminus \cup_{j\neq 2}A_j\})\setminus \ch(\{i_1\cup A_1\setminus \cup_{j\neq 1}A_j\})$ is nonempty (as the chain component $A_1\cap A_j\setminus \cup_{j'\neq 1,j}A_j'$ would lie in that set).
	Thus, we can pick $\mathbf{s}$ and $\mathbf{s'}$ to be the chain components $\{i_2\cup A_2\setminus \cup_{j\neq 2}A_j\}$ and $\{i_1\cup A_1\setminus \cup_{j\neq 1}A_j\}$.

	We now look at all the possible chain components that get changed between $\D$ and $\D'$.
	Let $B\subsetneq [k]\setminus \{2\}$ be the set such that $c$ lies in $A_j$ when $j\in B$.
	This implies that the set $\mathbf{T}=\{\mathbf{s'}\}\cup \ma(\mathbf{t})\setminus \{\mathbf{s}\}$ is precisely $\{\mathbf{s'}\}\cup \{\text{chain components in }\D \text{ containing }i_j \text{ for } j\in B\}$.
	Observe that for any $j'$ with $j'\notin B\cup \{1\}$, any chain component in $\D$ that is a child of the source chain component containing $i_{j'}$ (which is of the form $A_j'\cap_{l\in L}A_{l}\setminus \cup_{l'\in [n]\setminus L\cup j} A_{l'}$ for some set $L$) remains the same when $A_1$ is updated to $A_1\cup c$ and $A_2$ is updated to $A_2\setminus c$.

	Thus, the only two chain components that could get changed in $\D$ are $\mathbf{t}$ (which now can be precisely written as $\{A_2\cap_{j\in B}A_j\setminus \cup_{j'\notin B\cup \{2\}}A_{j'}\}$) and the chain component which is the child of the chain component $\mathbf{s'}$ and the source chain components containing $i_k$ with $k\in B$ (i.e., the chain component $\{A_1\cap_{j\in B}A_j \setminus \cup_{j' \notin \{1\}\cup B}A_{j'}\}$).
	However, the second chain component is exactly the one whose maximal ancestor set is $\mathbf{T}$, and hence gets assigned as $\mathbf{t'}$ in the Algorithm, if it exists.
	When $A_2$ gets updated to $A_2\setminus c$, it is clear that the chain component $\mathbf{t}$ gets updated to $\mathbf{t}\setminus c$ (as $\{(A_2\setminus c)\cap_{j\in B}A_j\setminus \cup_{j'\notin B\cup \{2\}}A_{j'}\}$ is equal to $\{A_2\cap_{j\in B}A_j\setminus \cup_{j'\notin B\cup \{2\}}A_{j'}\}\setminus \{c\}$).
	Similarly, when $A_1$ is updated to $A_1\cup c$, the chain component $\mathbf{t'}$ gets updated to $\mathbf{t'}\cup c$ and hence we add $c$ to $\mathbf{t'}$ or create a new chain component $\mathbf{c}$ depending on whether $\mathbf{t'}$ is empty.

	Out-of-fiber-delete: In this move we go from $\U$ having a representation $x_{i_1|A_1}x_{i_2|A_2}\cdots x_{i_k|A_k}$ to $\U'$ which has a representation $x_{i_1|A_1}x_{i_2|A_2\setminus c}\cdots x_{i_k|A_k}$, where $c$ lies in $A_1\cap A_2$.
	We can pick $\mathbf{t}$ as the chain component containing $c$ as that indeed is a child node (as $c$ lies in at least $A_1\cap A_2$).
	Further, as $c$ lies in $A_1\cap A_2$, the chain component in $\D$ containing $i_2$ is a maximal ancestor of $\mathbf{t}$.
	Thus, we can pick $\mathbf{s}$ as the chain component $\{i_2\cup A_2\setminus \cup_{j\neq 2}A_j\}$.
	Now, let $c$ lie in $A_j$ for all $j\in B$, for some $B\in [k]$.
	One of the chain components that could get changed in $\D$ is the chain component $\{\cap_{j\in B\setminus \{2\}}A_j\setminus \cup_{j' \notin B\setminus \{2\}}A_{j'}\}$.
	But this is precisely the chain component in $\D$ whose maximal ancestor set is $\mathbf{T}$ and hence gets assigned as $\mathbf{t'}$ in the algorithm.
	This gets changed to $\mathbf{t'}\setminus c$ when $A_2$ gets updated to $A_2\setminus c$ and thus, we either add $c$ to $\mathbf{t'}$ or create a new chain component $\mathbf{c}$ depending on whether $\mathbf{t'}$ is empty or not.

	Similarly, the only other chain component that gets changed is $\mathbf{t}$.
	As $c$ lies in $\mathbf{t}$, updating $A_2$ to $A_2\setminus c$ changes $\mathbf{t}$ to $\mathbf{t}\setminus c$ and hence we update the chain component accordingly depending on whether $\mathbf{t}$ contains only $c$ or not, which is again in accordance with the algorithm.

	In order to show the equivalence, we also need to show that if $\D$ and $\D'$ are respectively the input and output of the $\mathtt{algebraic}(\mathcal{D})$ algorithm with $\mathtt{out\_del}$ parameter, then there exist monomial representations of the UEC-representatives $\U$ and $\U'$ corresponding to $\D$ and $\D'$ such that they are connected by a single \textit{out-of-fiber-delete} operation on UEC-representatives.
	Now, the fact that we select $\mathbf{t}$ as a child node implies that there exists a set $B\in[k]$ and a representation $x_{i_1|A_1}x_{i_2|A_2}\cdots x_{i_k|A_k}$ of $\U$ such that $\cap_{j\in B}A_j$ is nonempty.
	As $\mathbf{t}$ gets updated to $\mathbf{t}\setminus v$ in $\D'$, there exists a monomial representation of $\D'$ where $c$ gets removed from at least one $A_j$ where $j \in B$.
	But as $\mathbf{t'}$ contains $v$, $v$ is contained in all but one $A_j$ where $j\in B$ in $\U'$.
	This implies that there exists some $j'\in B$ and a monomial representation of $\U'$ of the form $x_{i_1|A_1}x_{i_2|A_2}\cdots x_{i_{j'}|A_{j'}\setminus v} \cdots x_{i_k|A_k}$, which is indeed connected to $\U$ by a single \textit{out-of-fiber-delete} operation.

	Out-of-fiber-add: As the \textit{out-of-fiber-add} and \textit{out-of-fiber-delete} moves on UEC-representatives are inverses of each other (by Proposition \ref{prop:out-of-fiber-inverse}) and $\mathtt{out\_add}$ and $\mathtt{out\_del}$ operations on DAG reductions are also inverses of each other (by Lemma \ref{lemma:fiber add-delete alg inverse}), we use the same argument as seen in Lemma \ref{lemma:merge, split equivalent} for \textit{split} to conclude that the \textit{out-of-fiber-add} and $\mathtt{out\_add}$ moves are equivalent.
\end{proof}

In summation, we have seen that each DAG-reduction corresponds to a unique UEC-representative.
Moreover, the moves \texttt{merge}, \texttt{split}, \texttt{out\_del}, \texttt{out\_add} and \texttt{within} on DAG-reductions correspond to the moves \textit{merge}, \textit{split}, \textit{out-of-fiber-delete}, \textit{out-of-fiber-add} and \textit{within-fiber} on UEC-representatives.
In the next section, we implement these moves o DAG-reductions to give an MCMC search algorithm for estimating the marginal independence structure of a DAG model from data.

\section{\texttt{GrUES}: Gr\"obner-based Unconditional Equivalence Search}
\label{sec:grues:-markov-chain}

We now combine the results of Sections~\ref{sec:grobner},~\ref{sec: traversing} and~\ref{sec: DAG reductions} to develop \texttt{GrUES} (Algorithm~\ref{alg:grues}), a Markov Chain Monte Carlo (MCMC) method \citep{metropolis-hastings,without_likelihoods} for estimating the marginal independence structure of a DAG model.
Given a random sample \(X \subsetneq \mathbb{R}^n\), \texttt{GrUES} not only estimates the optimal UEC-representative, denoted by \(\widehat{\U}\), but also the posterior distribution \(\pi(\U \mid X)\) of UEC-representatives on \(n\) nodes.
We describe the algorithm in Section~\ref{sec:gr-algor} and apply it to synthetic data in Section~\ref{sec:appl-synth-data}.

\subsection{The \texttt{GrUES} algorithm for estimating a UEC and its posterior}
\label{sec:gr-algor}

\begin{algorithm}
	\SetKwProg{With}{with probability}{ do}{end}
	\Input{\(X\): data set;\\
		\(\U_\mathtt{init}\): initial UEC in the Markov chain;\\
		\texttt{length}: number of sampled UECs in the Markov chain; \\
		\texttt{transitions}: a probability mass function over the five moves;\\
		\texttt{prior}: a probability mass function \(\pi(\U)\);\\
		\texttt{score}: a function evaluating how well a UEC fits the data;}
	\Output{\(\widehat\U_{\mathtt{score}}\): \texttt{score}-optimal UEC;\\
		\texttt{markov\_chain}: UECs sampled from posterior \(\pi(\U \mid X)\);}
	\BlankLine
	\(\widehat\U_{\mathtt{score}} \coloneqq \U_{\mathtt{init}}\)\;
	\(\mathtt{markov\_chain} \coloneqq (\U_{\mathtt{init}})\)\;
	Initialize graph \(\D^\U\), the DAG-reduction of \(\U_{\mathtt{init}}\)\;
	\While{\(|\mathtt{markov\_chain}| < \mathtt{length}\)}{
		pick \(\mathtt{move} \in \{\mathtt{merge},\ \mathtt{split},\ \mathtt{out\_add},\ \mathtt{out\_del},\ \mathtt{within}\}\) according to \(\mathtt{transitions}\)\;
		\(\D^{\U'} \coloneqq \mathtt{move}(\D^\U)\)\;
		\If{\(\mathtt{score}(\U', X) > \mathtt{score}(\widehat\U_{\mathtt{score}}, X)\)}{
			\(\widehat\U_{\mathtt{score}} \coloneqq \U'\)\;
		}
		\(h = \min\big(1, \frac{\pi(X \mid \U')\mathtt{prior}(\U')q(\U', \U)}{\pi(X \mid \U)\mathtt{prior}(\U)q(\U, \U')}\big)\)\;\label{alg:grues:h}
		\With{h}{
			\(\U \coloneqq \U'\)\;
			\(\D^{\U} \coloneqq \D^{\U'}\)\;}
		append \(\U\) to \(\mathtt{markov\_chain}\)\;
	}
	\Return \(\widehat\U_{\mathtt{score}},\ \mathtt{markov\_chain}\)
	\caption{\(\mathtt{GrUES}(X, \U_{\mathtt{init}}, \mathtt{length}, \mathtt{transitions}, \mathtt{prior}, \mathtt{score})\)}
	\label{alg:grues}
\end{algorithm}

Using the moves specified in Section~\ref{sec: DAG reductions}, \texttt{GrUES} randomly traverses the space of UEC-representatives.
In doing so, it keeps track of (i) the optimal UEC-representative \(\widehat\U_{\mathtt{score}}\) encountered, for a given ``\(\mathtt{score}\)'' function (e.g., the BIC score), and (ii) a sequence ``\texttt{markov\_chain}'' of moves recording its progression at each step from the current UEC-representative to a proposed UEC-representative.
(i) and (ii) are the outputs of Algorithm~\ref{alg:grues}.
For (i), the fact that \texttt{GrUES} is doing a random walk allows it to avoid getting stuck in the local optima to which greedy methods may succumb.
We note that \(\widehat\U_{\mathtt{score}}\) is optimal over all UEC-representatives evaluated during Lines~1--8 of the Monte Carlo process and not just the optimum over the Markov chain.
(ii) allows for identification of a maximum a posteriori (MAP) estimate of the data-generating model as an alternative to the \texttt{score}-optimal estimate. 

Besides a data set \(X\), \texttt{GrUES} has input parameters \(\U_{\mathtt{init}}\), \texttt{length}, \texttt{transitions}, \texttt{prior}, and \texttt{score}.
\texttt{GrUES} must also compute likelihoods \(\pi(X \mid \U)\) and the transition kernel \(q\) in Line~\ref{alg:grues:h}.
The following explains each of these in more detail.

\subsubsection*{The \(\U_{\mathtt{init}}\) parameter}
\texttt{GrUES} initializes at \(\U_{\mathtt{init}}\), which can be any UEC-representative.
Here, we use the result of hypothesis tests for independence.
The independence tests also provide a reasonable baseline for comparison to the results of \texttt{GrUES}, because they constitute a constraint-based method for estimating a UEC-representative.
However, independence tests return an undirected graph but not necessarily one that represents a UEC---in these cases, our implementation uses the largest subgraph that is UEC-representative.

\subsubsection*{The \(\mathtt{length}\) parameter}
A longer Markov chain provides a more accurate estimate of the posterior, as well as improving the MAP and \texttt{score}-based estimates, at the expense of longer runtime.
Given the number of UEC-representatives on $n$ nodes (see Table \ref{table: number of uecs}), \texttt{length} is typically much shorter than the size of the state space.
Fortunately, MCMC methods in general are known to provide reasonable estimates even when the Markov chain is much smaller than the search space, especially when given a good initialization and prior.
We provide an empirical evaluation of \texttt{GrUES} while varying the \texttt{length} parameter in Section \ref{sec:comp-map-estim}.

\subsubsection*{The likelihood ratio}
MCMC-based methods such as Algorithm~\ref{alg:grues} rely on the computation of the likelihood ratio \(\frac{\pi(X \mid \U')}{\pi(X \mid \U)}\).
In principle, any method of computing this likelihood ratio or something similar \citep{without_likelihoods} can be used, for example linear regression in the case of continuous data or the method of \citep{boege2022marginal} in the case of discrete data.
In Section~\ref{sec:appl-synth-data}, we run \texttt{GrUES} on linear Gaussian data, so we first fit linear regression models using the respective CPDAG representations of \(\U\) and \(\U'\).
In particular, for a node \(v\) in the CPDAG \(\G = \mathtt{init\_CPDAG}(\U)\), we model the data column corresponding to \(v\) as a linear function of those corresponding to \(\pa_\G(v) \cup \cc_\G(v)\).
Having found the parameters of the MLE, we use them along with the data to evaluate the log-likelihood function---the \(\log\) simplifies the computation and helps avoid underflow without changing the relative ranking of models with respect to the optimization problem.
We then use the \(\log\) of the likelihood ratio, \(\log(\pi(X\mid \U')) - \log(\pi(X\mid \U))\), in place of the likelihood ratio.


\subsubsection*{The \(\mathtt{transitions}\) parameter and transition kernel \(q\)}
The parameter \texttt{transitions} is a probability mass function specified over the five possible moves one could make to transition between UEC-representatives: \texttt{split}, \texttt{merge}, \texttt{out\_add}, \texttt{out\_del} and \texttt{within}.
It affects the transition probability $h$ because changing \texttt{transitions} in turn changes the kernel \(q\).
If \texttt{transitions} is chosen to be a uniform distribution, the moves \texttt{split}, \texttt{merge}, \texttt{out\_add}, and \texttt{out\_del} are each made with probability \(\frac{1}{6}\), and \texttt{within} is made with probability \(\frac{1}{3}\) since it is its own inverse.
By default in Algorithm~\ref{alg:grues}, the transition kernel \(q(u,u^\prime)\) is specified as a hierarchy of uniform distributions (this is done for computational ease, though in principle any well-defined transition kernel can be specified).
To transition from \(\U = u\) (where $\U$ now denotes the random variable taking values $u$, a UEC-representative) only a subset of the complete set of moves may be possible, say $M_u$.
Since $M_u$ may be a strict subset of the moves considered in \texttt{transitions} we assign the probabilities to the moves in $M_u$ by normalizing their probabilities; i.e., for each $m\in M_u$ we take the probability of choosing $m$ to be the conditional probability $P(m \mid m\in M_u)$ where $P(m)$ denotes the probability mass function specified by \texttt{transitions}.
Then any necessary dependencies to perform the selected move are analogously selected uniformly at random.
For example, if the chosen move is merge, say $M = \mathtt{merge}$, then a choice of which cliques to merge must be made, say $N = \{\mathbf{s},\mathbf{s}'\}$, from the set \[ S_u \coloneqq \{\{\mathbf{s},\mathbf{s}'\} : \mathbf{s}, \mathbf{s}' \text{ satisfy Line~\ref{alg:reduced-split:pick-source} of Algorithm~\ref{alg:reduced-merge}}\} \] of all pairs of cliques in \(U = u\) satisfying the conditions for the move \texttt{merge}.
Hence, \[ q(u, u^\prime) = P(M = \mathtt{merge} \mid \U = u)P(N = \{\mathbf{s}, \mathbf{s}'\} \mid M = \mathtt{merge}, \U = u) \] where \begin{equation*}
	\begin{split}
		M \mid \U = u &\sim P(M \mid M\in M_u),\\ N \mid \U = u, M = \mathtt{merge} &\sim \mathrm{Uniform}(S_u).
	\end{split}
\end{equation*}
The corresponding transition matrix for other moves is specified similarly, with their respective dependencies filled in (note that each ``pick'' in the pseudocode of the moves, e.g., in Algorithms~\ref{alg:reduced-merge},~\ref{alg:reduced-split}, and~\ref{alg:reduced-algebraic}, corresponds to such a dependence).

\subsubsection*{The \(\mathtt{prior}\) parameter}
Prior knowledge about the data-generating model can be incorporated into the prior $\pi(\U)$ which is specified via the \texttt{prior} parameter.
As an example, one may have prior knowledge about the possible number of source nodes in the true UEC (i.e., the number of causal sources in the data-generating DAG).
We can capture such prior knowledge or beliefs using the following prior, which we denote as \(\Delta_s^p\) for number of sources \(s\) and scale parameter \(p\).
It is defined by exponentiation of a discretized triangle distribution with support \([n]\), for graphs on \(n\) nodes: using the sequence \(d\) defined by \[ d_i =
	\begin{cases}
		(\frac{2i}{(n+1)s})^p & \text{if } i < s \\ (\frac{2}{n+1})^p & \text{if } i = s\\ (\frac{2((n+1)-i)}{(n+1)(n + 1-s)})^p & \text{if } i > s
	\end{cases}
\] we define \(\Delta_s^p(\U) = \frac{d_i}{\sum d}\), where \(i\) is the number of sources in \(U\).
In Subsection~\ref{sec:appl-synth-data}, we present experiments on synthetic data in which we apply this prior as well as a noninformative (uniform) prior.

\subsubsection*{The \(\mathtt{score}\) parameter}
This parameter lets the user specify a score for each UEC.
The score function has no impact on the Markov chain, posterior or MAP estimates.
It is a score recorded for each UEC-representative queried at any step in the algorithm.
A natural choice for this scoring function is the BIC score, which is commonly used in the causal structure learning literature.
In our application on simulated Gaussian DAG models in Section~\ref{sec:comp-map-estim}, we use the \(\ell_0\)-penalized maximum likelihood (BIC) \citep{geer2013} and the \(\mathtt{nuclear}\)-penalized maximum likelihood, which uses the nuclear norm in its penalization term \citep{hastie2015statistical}.
In both cases, the maximum likelihood is obtained by learning a linear regression model using the CPDAG representative of a given UEC and then penalized according to the respective norm of the adjacency matrix.

\subsection{Applications on synthetic data}
\label{sec:appl-synth-data}

We offer an implementation of our algorithm as a libre/free Python package.
Instructions for installing the package and reproducing all of the following results can be found in the package documentation: \url{https://gues.causal.dev/repro_astat}.
In total, all experiments in the following subsections took about 30 hours of (parallelized) runtime, on an AMD Ryzen 7 PRO 4750G CPU.
This compute time includes, in addition to running GrUES, the time taken to generate the approximately $2\,300$ data sets used in the experiments.
While this gives a rough estimate of how long a single instance of \texttt{GrUES} takes, in Section~\ref{sec:estim-time-compl} we provide a more accurate empirical estimate of how the algorithm scales in complexity as the number of nodes, sample size, and length of Markov chain increase, respectively.

\subsubsection{Learning a posterior}
\label{sec:learn-post-over}

We now visualize a posterior learned by \texttt{GrUES} from a single synthetic data set.
As we saw in Table~\ref{table: number of uecs}, the size of \(\mathbb{U}^n\) grows quickly as \(n\) increases.
Hence, in order to easily visualize \(\pi(\U \mid X)\), we simulated a data set over \(n=3\) nodes, in which case there are 8 possible UECs of DAGs.
To simulate the data, we used the independent and identically distributed exogenous variables \(\varepsilon_1, \varepsilon_2, \varepsilon_3\sim \mathcal{N}(0, 1)\) and the randomly generated weight matrix \[ W =
	\begin{bmatrix}
		0 & 0 & -0.9247 \\ 0 & 0 & 0\\ 0 & 0 & 0
	\end{bmatrix}
\] (corresponding to a DAG) to define the linear Gaussian additive noise model \(\mathbf{X} \coloneqq \mathbf{X}W + \boldsymbol\varepsilon\), where $\boldsymbol\varepsilon =
	\begin{bmatrix}
		\varepsilon_1 & \varepsilon_2 & \varepsilon_3
	\end{bmatrix}
	^\top$ from which we drew a sample of \(1000\) observations.
We then ran \texttt{GrUES}, initialized on the empty graph, with prior \(\pi(\U)=\Delta_2^3\) as described in Section~\ref{sec:gr-algor}, and with uniform transition probabilities so that the transition kernel \(q\) is the product of uniform conditional densities described in Section~\ref{sec:gr-algor}.
\texttt{GrUES} constructed a Markov chain of length 2\,000, the first 1\,000 of which were discarded as burn-in, yielding $1\,000$ samples from the posterior.
The histogram of this estimated posterior is shown in Figure~\ref{fig:posterior:histogram}, which indicates a MAP estimate \(\pi(\U_2 \mid X) = 0.208\), and indeed, \(\U_2\) (shown in Figure~\ref{fig:posterior:Us}) is the true UEC of the DAG corresponding to \(W\).

Also notice that the second most probable model according to the estimated posterior is \(\U_5\), which contains the correct number of sources (i.e., two) and the true edge \(1 \mathdash 3\) (in addition to a single falsely inferred edge).
This model is assigned higher posterior probability than \(\U_3\) and \(\U_4\), which also contain the correct number of sources but do not contain the true edge, indicating that \textrm{GrUES} can correctly make use of the likelihood to capture relevant dependencies while comparing models that are equally likely according to the prior.
However, we see that \(\U_7\) (which contains the correct number of sources but not the true edge) counterintuitively has higher posterior probability than \(\U_6\) (which does contain the true edge, similar to \(\U_5\))---we attribute this to the finite sample.

Finally, notice that we have two extremes in terms of number of source nodes for this small example: on the one hand, the maximum number of possible sources is three (corresponding to the empty graph $\U_1$) and the minimum number is one (corresponding to the complete graph $\U_8$).
The only UEC-representative with three sources is the independence graph which will tend to suffer a low likelihood as it fails to capture any of the relevant dependencies in the model.
On the other hand, $\U_8$ represents the UEC containing all complete DAGs, and it may achieve an unduly high (non-penalized) likelihood due to its ability to overfit the data given that it has the most possible parameters.
We see here that the prior $\Delta_2^3$, which favors graphs with two source nodes, helps to mitigate this problem, as is reflected in the relatively low posterior probability of $\U_8$.
In particular, \texttt{GrUES} assigns the model \(\U_8\) a lower posterior probability than all the models with 2 sources, indicating that \texttt{GrUES} correctly uses the prior to balance the likelihood.

\begin{figure}
	\begin{subfigure}{.54\textwidth}
		\centering
		\includegraphics[width=\linewidth]{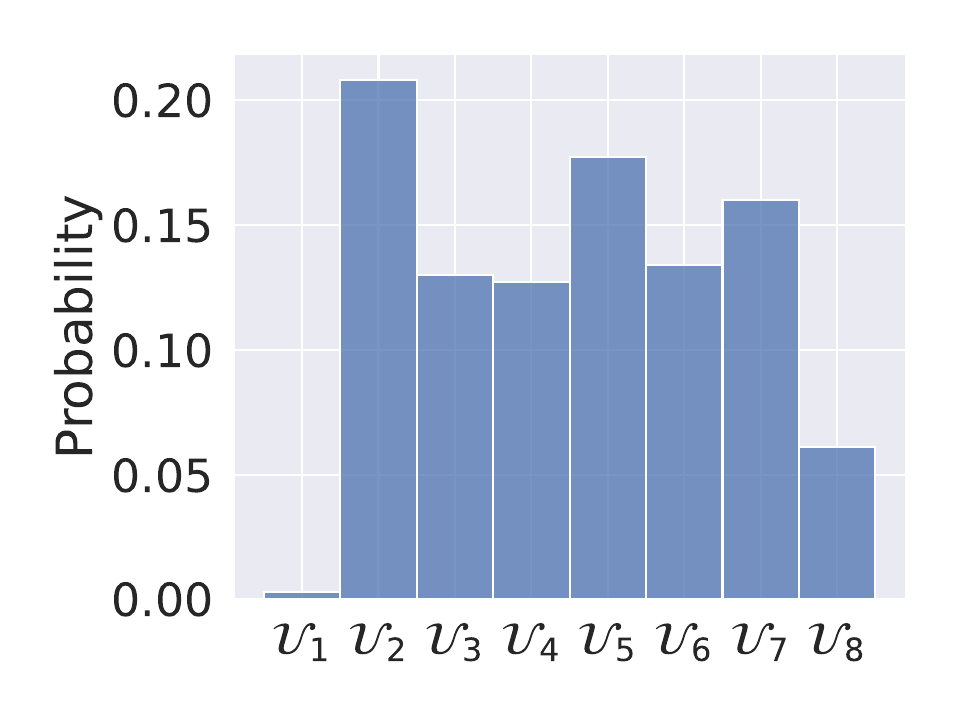}
		\caption{Histogram of \(\pi(\U \mid X)\)}\label{fig:posterior:histogram}
	\end{subfigure}
	\begin{subfigure}{.45\textwidth}
		\begin{subfigure}{\linewidth}
			\begin{tikzpicture}
				\filldraw[black]
				(9,0.7) circle [radius=.04] node [below]
					{3}
				(8.5,1.1) circle [radius=.04] node [left]
					{1}
				(9,1.5) circle [radius=.04] node [above]
					{2}
				(8.5,.5) circle [radius=0] node [below] {\(\U_1\)};

				\filldraw[black]
				(10.5,0.7) circle [radius=.04] node [below]
					{3}
				(10,1.1) circle [radius=.04] node [left]
					{1}
				(10.5,1.5) circle [radius=.04] node [above]
					{2}
				(10,.5) circle [radius=0] node [below] {\(\U_2\)};
				\draw
				(10.5,0.7)--(10,1.1);

				\filldraw[black]
				(12,0.7) circle [radius=.04] node [below]
					{3}
				(11.5,1.1) circle [radius=.04] node [left]
					{1}
				(12,1.5) circle [radius=.04] node [above]
					{2}
				(11.5,.5) circle [radius=0] node [below] {\(\U_3\)};
				\draw
				(11.5,1.1)--(12,1.5);

				\filldraw[black]
				(13.5,0.7) circle [radius=.04] node [below]
					{3}
				(13,1.1) circle [radius=.04] node [left]
					{1}
				(13.5,1.5) circle [radius=.04] node [above]
					{2}
				(13,.5) circle [radius=0] node [below] {\(\U_4\)};
				\draw
				(13.5,1.5)--(13.5,0.7);
			\end{tikzpicture}
		\end{subfigure}
		\begin{subfigure}{\linewidth}
			\begin{tikzpicture}
				\filldraw[black]
				(15,0.5) circle [radius=.04] node [below]
					{3}
				(14.5,0.9) circle [radius=.04] node [left]
					{1}
				(15,1.3) circle [radius=.04] node [above]
					{2}
				(14.5,.3) circle [radius=0] node [below] {\(\U_5\)};
				\draw
				(15,0.5)--(14.5,0.9)
				(15,1.3)--(14.5,0.9);

				\filldraw[black]
				(16.5,0.5) circle [radius=.04] node [below]
					{3}
				(16,0.9) circle [radius=.04] node [left]
					{1}
				(16.5,1.3) circle [radius=.04] node [above]
					{2}
				(16,.3) circle [radius=0] node [below] {\(\U_6\)};
				\draw
				(16.5,0.5)--(16.5,1.3)
				(16.5,0.5)--(16,0.9);

				\filldraw[black]
				(18,0.5) circle [radius=.04] node [below]
					{3}
				(17.5,0.9) circle [radius=.04] node [left]
					{1}
				(18,1.3) circle [radius=.04] node [above]
					{2}
				(17.5,.3) circle [radius=0] node [below] {\(\U_7\)};
				\draw
				(18,0.5)--(18,1.3)
				(18,1.3)--(17.5,0.9);

				\filldraw[black]
				(19.5,0.5) circle [radius=.04] node [below]
					{3}
				(19,0.9) circle [radius=.04] node [left]
					{1}
				(19.5,1.3) circle [radius=.04] node [above]
					{2}
				(19,.3) circle [radius=0] node [below] {\(\U_8\)};
				\draw
				(19.5,0.5)--(19,0.9)
				(19,0.9)--(19.5,1.3)
				(19.5,1.3)--(19.5,0.5);
			\end{tikzpicture}
			\bigskip
		\end{subfigure}
		\caption{All UECs in the set \(\mathbb{U}^3\)}\label{fig:posterior:Us}
	\end{subfigure}
	\caption{A posterior learned with \texttt{GrUES}.}\label{fig:posterior}
\end{figure}

\subsubsection{Comparing MAP estimates, \(\mathtt{score}\)-based estimates, independence tests, and the ground truth}
\label{sec:comp-map-estim}

In order to get a sense of \(\mathtt{GrUES}\)'s performance, we applied it to synthetic data sets whose generating DAGs have varying edge densities and numbers of nodes.
We additionally varied \texttt{GrUES}'s Markov chain length and prior, for a total of four sets of experiments:
\begin{enumerate}
	\item \texttt{prior} is \(\Delta_s^n\), where $n$ and $s$ are, respectively, the true numbers of variables and source nodes in the given model, and \texttt{length} of the estimated Markov chain is greater than the size of the search space \(|\mathbb{U}|\),
	\item \texttt{prior} is uniform and \texttt{length} is greater than \(|\mathbb{U}|\),
	\item \texttt{prior} is \(\Delta_s^n\) and \texttt{length} is smaller than \(|\mathbb{U}|\), and
	\item \texttt{prior} is \(\Delta_s^n\) and \texttt{length} is \textit{much} smaller than \(|\mathbb{U}|\).
\end{enumerate}

\subsubsection*{First experimental setting}

\begin{figure}
	\centering
	\begin{subfigure}{.8\textwidth}
		\includegraphics[width=\linewidth]{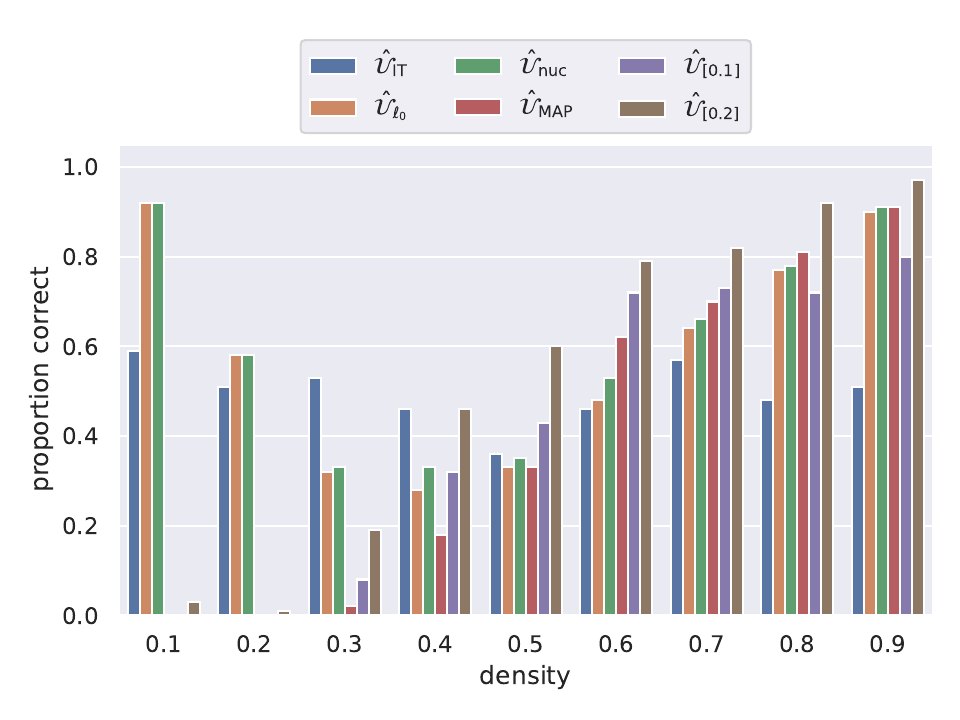}
		\caption{Proportion of experiments where \texttt{GrUES} correctly estimates the true UEC:
			\(\widehat\U_{\mathrm{IT}}\) is the estimate from independence tests;
			\(\widehat\U_{\ell_0}\) is the \(\ell_0\)-penalized maximum likelihood estimate;
			\(\widehat\U_{\mathrm{nuc}}\) is the \(\mathtt{nuclear}\)-penalized maximum likelihood estimate;
			\(\widehat\U_{\mathrm{MAP}}\) is the maximum a posteriori estimate;
			\(\widehat\U_{[0.1]}\) and \(\widehat\U_{[0.2]}\) are, respectively, the \(10\%\) and \(20\%\) HPD credible sets, and the estimate is correct in these cases when the set contains the true UEC.}
		\label{fig:sims:correct}
	\end{subfigure}
	\begin{subfigure}{.8\textwidth}
		\includegraphics[width=\linewidth]{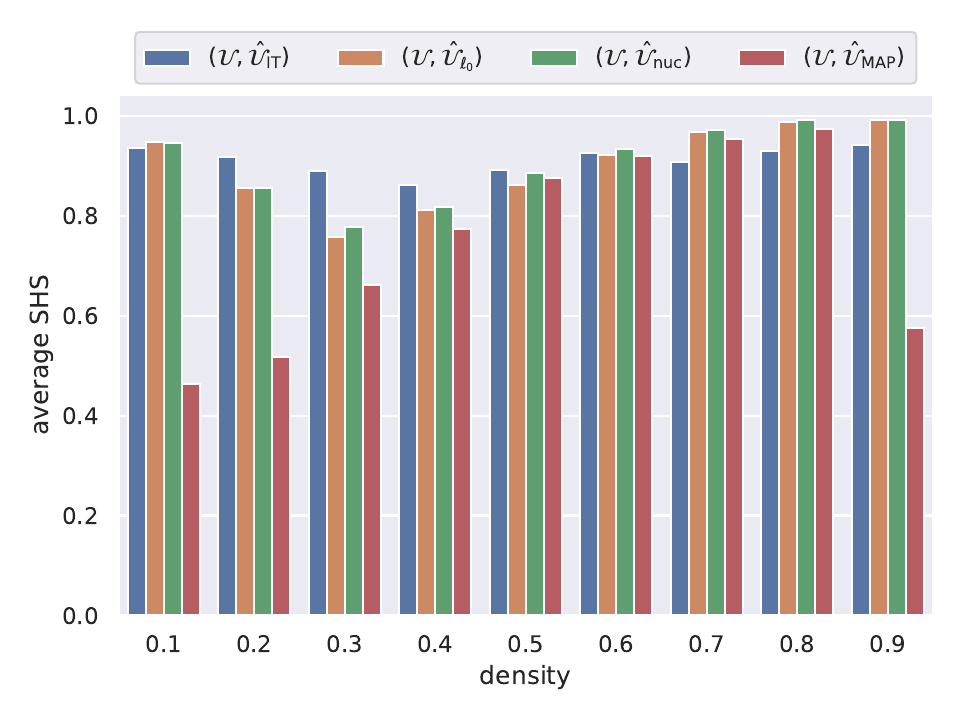}
		\caption{Average SHS (structural Hamming similarity) of the different \texttt{GrUES} estimates compared to the true UEC \(\U\).}\label{fig:sims:accuracy}
	\end{subfigure}
	\caption{First experimental setting, with \texttt{prior} set to \(\Delta_s^{5}\) and Markov chain \(\mathtt{length} = 10\,000 >  |\mathbb{U}^5| = 462\).
		Results averaged over 100 different data sets per density.
		Each data set from DAG model with 5 nodes.
	}
	\label{fig:sims}
\end{figure}

For each edge probability $p \in\{0.1, 0.2, \ldots,0.9\}$, we generated 100 random linear Gaussian DAG models on \(n=5\) nodes, assigning mutually independent standard normal errors and assigning edge weights uniformly at random from \([-1,\ 1]\setminus \{0\}\).
We then applied \(\mathtt{GrUES}\) to each data set (1\,000 observations sampled from each DAG model), using the following: We use partial correlation tests for pairwise independence with significance level $\alpha = 0.05$ using the product-moment correlation coefficient, to yield a graph which we denote \(\widehat\U_{\mathrm{IT}}\).
We then extract the initial graph $\U_{\texttt{init}}$ as the subgraph of $\widehat\U_{\mathrm{IT}}$ as described in Subsection~\ref{sec:gr-algor}.
We use the true number of source nodes \(s\) in each data-generating DAG to set the prior \(\pi(\U)\) to be the function \(\Delta_s^{5}\) described in Section~\ref{sec:gr-algor}, leaving the transition kernel \(q\) to be the product of uniform conditional densities also described in Section~\ref{sec:gr-algor}.
Finally, we use a Markov chain of length 10\,000, which is approximately 21 times greater than the size of the UEC search space \(|\mathbb{U}^5| = 462\) (recall Table~\ref{table: number of uecs}).
The results of \texttt{GrUES} on the synthetic data sets with these choices are shown in Figure~\ref{fig:sims} and Table~\ref{tab:intervals}.

As a baseline, first consider the performance of \(\widehat\U_{\mathrm{IT}}\).
Figure~\ref{fig:sims:correct} ({\color{MidnightBlue} blue bar}) shows that it identifies the true UEC for between \(0.35\) and \(0.6\) of the data-generating DAGs, performing slightly better for sparse DAGs.
Figure~\ref{fig:sims:accuracy} ({\color{MidnightBlue} blue bar}) shows that for a fixed density $p$ the estimator \(\widehat\U_{\mathrm{IT}}\) has average \emph{structural Hamming similarity} (SHS)\footnote{The (standardized) SHS of two graphs is defined as the ratio of the number of (non)edges they have in common divided by the total number of (non)edges.
	For example, an SHS of 1 means the graphs are identical and 0 means they are complements.
} with the true UEC above \(0.85\), regardless of sparsity of the generating DAG.
This disparity makes sense, considering each edge in the UEC is estimated independently, without regard for any kind of global structure (whereas preceding sections demonstrate the importance of global clique structure).

Both penalized MLEs \(\widehat\U_{\ell_0}\) and \(\widehat\U_{\mathrm{nuc}}\) outperform the baseline \(\widehat\U_{\mathrm{IT}}\) at identifying the true UEC for either very sparse or moderately to very dense graphs and underperform it otherwise (Figure~\ref{fig:sims:correct} {\color{BurntOrange} orange bar} and {\color{ForestGreen} green bar}).
The penalized MLEs also generally have a lower average structural Hamming similarity to the true UEC (Figure~\ref{fig:sims:accuracy} {\color{BurntOrange} orange bar} and {\color{ForestGreen} green bar}).
This appears to be a consequence of the observation that while the baseline method \(\widehat\U_{\mathrm{IT}}\) randomly/independently gets some edges wrong, the penalized MLEs \(\widehat\U_{\ell_0}\) and \(\widehat\U_{\mathrm{nuc}}\) make mistakes in a more systematic way by incorrectly removing edges to accommodate the sparsity constraints enforced by the penalization.
This makes these estimators more likely to get edges wrong when there are several sizeable cliques.
Also notice that the penalized likelihood score with \(\mathtt{nuclear}\) penalization consistently outperforms that with \(\ell_0\) penalization.
This makes intuitive sense as the \(\mathtt{nuclear}\) penalization is commonly used when we are searching for clique structures that form communities \citep{hastie2015statistical}.
Here, the communities sought after are those defined as the ancestors of the individual source nodes in the data-generating DAG.

It is also interesting to note that \texttt{GrUES} generally performs worse than independence testing in the range of densities 0.3--0.5.
This is potentially explained by (possibly a combination of) two things: First, with finite data, it is possible that there exists a better-scoring UEC for the given data than the data-generating model, which pulls \texttt{GrUES} away from the true model identified by the independence testing and initialization.
(Note that in the case that independence testing gets the correct UEC, \texttt{GrUES} is initialized at this UEC.)
This problem is overcome by increasing the sample size when possible.
Second, in the case that that independence testing does not identify the true UEC, \texttt{GrUES} initializes at a subgraph of $\widehat\U_{\textrm{IT}}$.
For denser graphs, where independence testing fails more frequently, \texttt{GrUES} appears to find its way to the true model a reasonable percentage of the time.
However, this does not appear to be happening in the density range 0.3--0.4.
One possibility is that within this range, when independence testing does not learn the true UEC, \texttt{GrUES} initializes at a graph that is far away from the true UEC in terms of connectivity of the search space.
In particular, for these densities, the graphs we are initializing at may be isolated away from the true model by several bottlenecks that make it difficult for \texttt{GrUES} to correct for errors in independence testing.
It is possible that the traversal of such bottlenecks are more frequently needed for models in this density range than for densities greater than $0.4$.

\begin{table}
	\centering
	\begin{tabular}{@{} *{10}{c} @{}}
		\headercell{Credible                                                                                                                                                          \\ probability \(t\)} & \multicolumn{9}{c@{}}{DAG edge density \(p\)}\\
		\cmidrule(l){2-10}
		                                       & 0.1       & 0.2        & 0.3       & 0.4      & 0.5       & 0.6        & 0.7         & 0.8       & 0.9                               \\\midrule
		\begin{filecontents*}{tri_n_table.csv}
t,0.1,0.2,0.3,0.4,0.5,0.6,0.7,0.8,0.9
0.1,12,12,11,8,6,5,4,3,2
0.2,33,32,27,21,17,15,13,12,11
0.3,59,56,46,35,28,26,24,23,21
0.4,88,84,68,51,41,38,36,34,32
0.5,123,116,93,68,54,51,48,46,45
0.6,164,153,122,88,69,65,61,59,58
0.7,212,197,155,110,86,80,76,74,72
0.8,270,249,200,144,114,105,100,96,92
0.9,342,313,272,225,200,190,186,180,174
\end{filecontents*}
\csvreader{tri_n_table.csv}{}{\csvcoli & \csvcolii & \csvcoliii & \csvcoliv & \csvcolv & \csvcolvi & \csvcolvii & \csvcolviii & \csvcolix & \csvcolx\\}
	\end{tabular}
	\vspace{-2em}
	\caption{Average size of HPD set \(\widehat\U_{[t]}\) learned on data generated by DAGs on $5$ nodes with different edge densities \(p\).}
	\label{tab:intervals}
\end{table}

Finally, \(\widehat\U_{\mathrm{MAP}}\) (Figure~\ref{fig:sims:correct} {\color{BrickRed} red bar}) outperforms \(\widehat\U_{\mathrm{IT}}\), \(\widehat\U_{\ell_0}\), and \(\widehat\U_{\mathrm{nuc}}\) when the generating DAG has edge density \(p > 0.5\) (and is approximately tied when \(p = 0.5\)).
For example, $\widehat\U_{\textrm{MAP}}$ performs nearly twice as well as \(\widehat\U_{\mathrm{IT}}\) for densities \(p > 0.7\).
We also show results for the proportion of times the $10\%$ and $20\%$ HPD credible sets contain the true UEC.
These sets are denoted $\U_{[0.1]}$ and $\U_{[0.2]}$ and they respectively denote the set of UEC-representatives with highest posterior probability such that the posterior probability of a random UEC-representative lying in this set is $0.1$ and $0.2$, respectively.
Since the set of UECs is discrete, we take the UEC-representatives with highest posterior probability such that their total probability is as close to the specified credible probability ($0.1$ or $0.2$) without exceeding it.
The proportion of times these sets capture the true UEC-representative is shown in Figure~\ref{fig:sims:accuracy} by the {\color{Purple} purple bar} and {\color{Brown} brown bar}, respectively.
These statistics demonstrate that even when the true UEC-representative is not learned by \(\widehat\U_{\mathrm{MAP}}\), it is often among a small set (see Table~\ref{tab:intervals} for sizes of the different HPD credible sets) of UEC-representatives with high probability according to the learned posterior \(\pi(\U \mid X)\).
This suggests that either a MAP estimate or a small HPD credible set learned by \texttt{GrUES} can be used to estimate reasonable sets of possible source nodes to be used in experimental design (or other applications) when the causal system under consideration is expected to have a DAG with edge density at least $0.5$.

\begin{figure}
	\centering
	\includegraphics[width=0.8\linewidth]{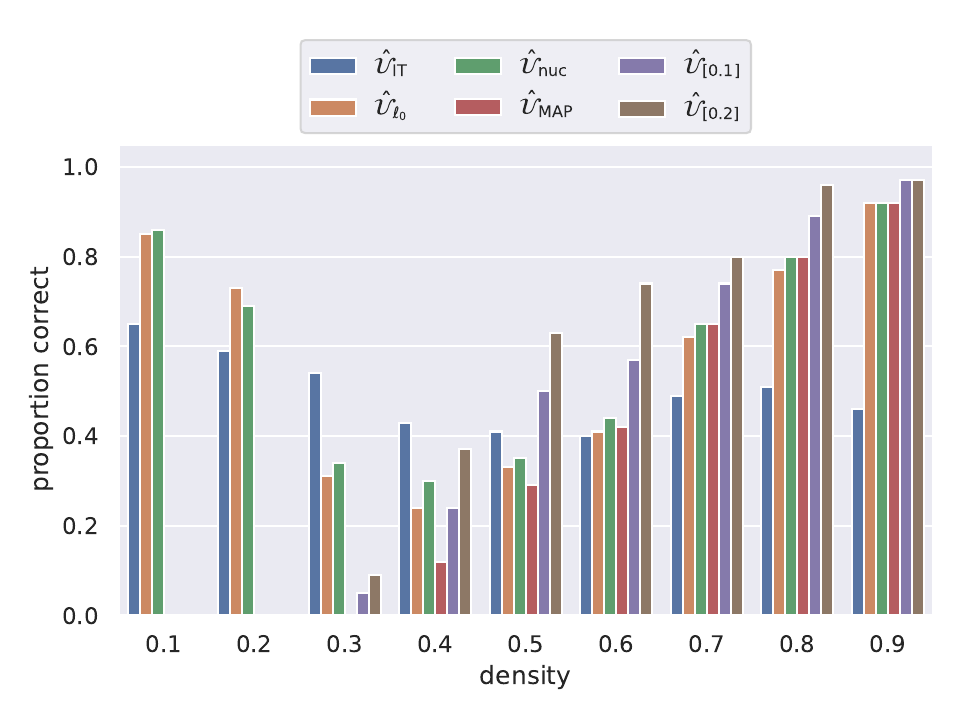}
	\caption{Second experimental setting, with \texttt{prior} left uniform and Markov chain \(\mathtt{length} = 10\,000 >  |\mathbb{U}^5| = 462\).
		Results averaged over 100 different data sets per density, with each data set drawn from a DAG model on 5 nodes.
	}
	\label{fig:uni}
\end{figure}

\subsubsection*{Second experimental setting}
\texttt{GrUES} is given a uniform (noninformative) prior rather than \(\Delta_s^{5}\) as in the first setting.
The results are shown in Figure~\ref{fig:uni}.
Notice that the same general patterns emerge as in the first setting: penalized MLEs performs better than the baseline independence tests for sparse data-generating DAGs, while the MAP estimate performs better than independence testing and penalized MLEs for data-generating DAGs with density \(p> 0.5\).
This shows that \texttt{GrUES} is not overly advantaged by our specific choice of prior.

\begin{figure}
	\centering
	\includegraphics[width=0.8\linewidth]{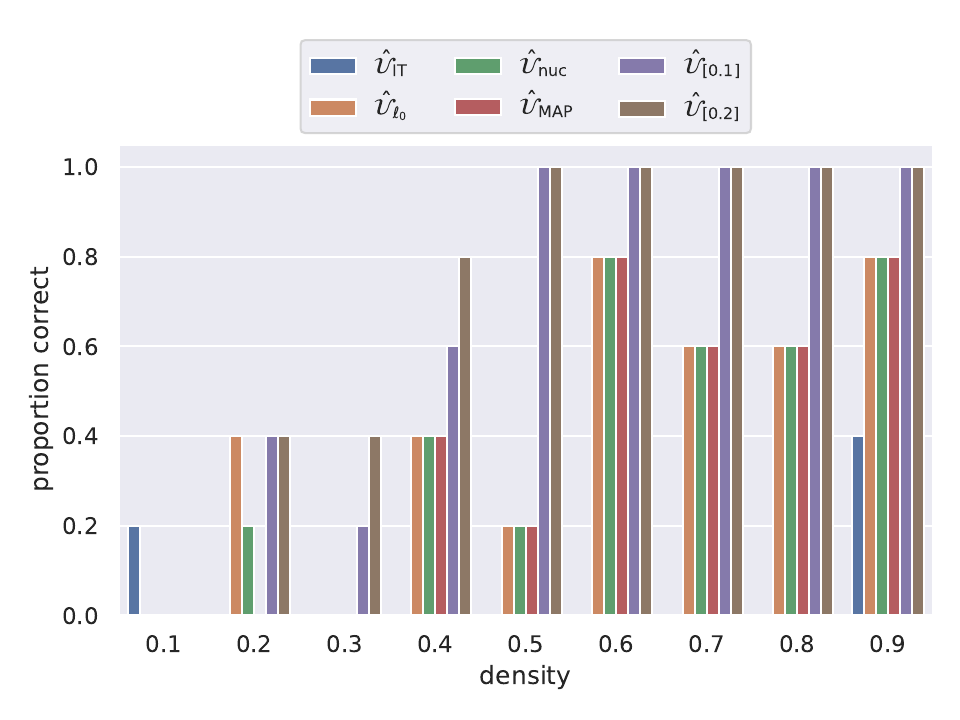}
	\caption{Third experimental setting, with \texttt{prior} set to \(\Delta_s^{8}\) and Markov chain \(\mathtt{length} = 1\,000\,000 < |\mathbb{U}^8| = 3\,731\,508\).
		Results averaged over 5 different data sets per density.
		Each data set from DAG model with 8 nodes.
	}
	\label{fig:n-8}
\end{figure}

\subsubsection*{Third experimental setting}
We again use the \(\Delta_s^{n}\) prior (as in the first setting) but increase the number of nodes in the data-generating DAGs to \({n=8}\).
Hence, the size of the search space increases to \({|\mathbb{U}^8|=3\,731\,508}\).
We also increase the size of the Markov chain to \(1\,000\,000\), only roughly \(\frac{1}{4}\) the size of the search space.
Due to constraints on our computational resources, we reduce the number of data sets per density to 5 (allowing this experimental setting to run all 45 experiments in roughly 15 hours on two CPU cores).
The results are shown in Figure~\ref{fig:n-8}.

Notice that independence testing performs generally quite poorly, correctly estimating an average of less than 0.07 of the true UEC-representatives, while the penalized MLEs (especially with nuclear penalization) do considerably better at most densities.
Furthermore, notice that the MAP estimate is at least as good as the penalized-MLEs for \(p >0.3\).
These results indicate that \texttt{GrUES} can attain accurate penalized-MLE and MAP estimates even when independence testing fails and when the size of the Markov chain is smaller than the search space of UEC-representatives.

\begin{figure}
	\includegraphics[width=0.8\linewidth]{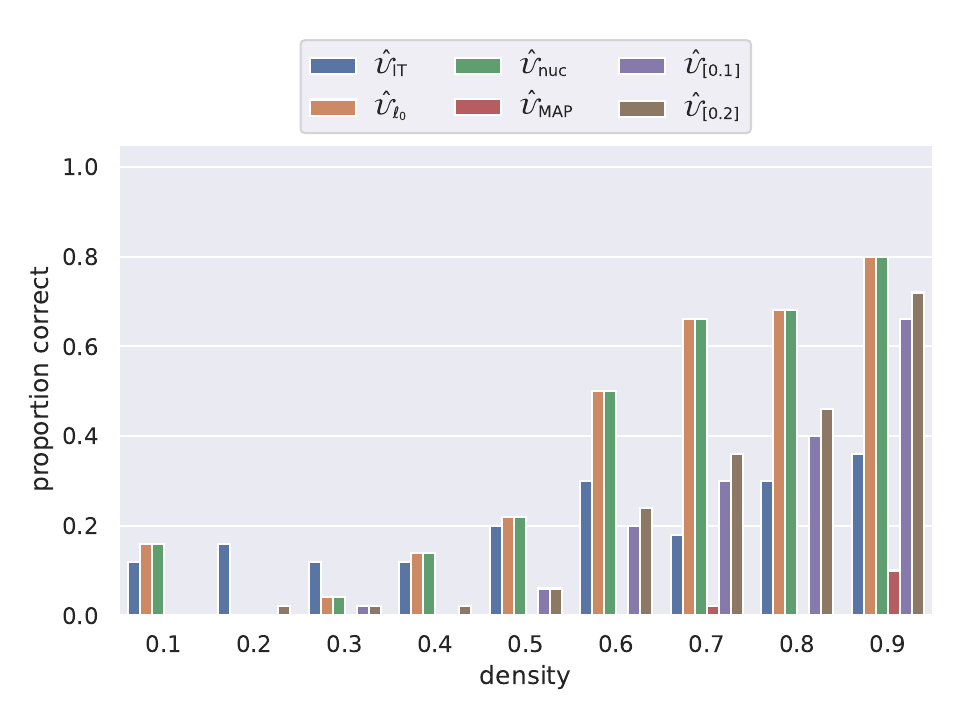}
	\caption{Fourth experimental setting, with \texttt{prior} set to \(\Delta_s^{10}\) and Markov chain \(\mathtt{length} = 100\,000 \ll |\mathbb{U}^{10}| = 8\,780\,782\,707\).
		Results averaged over 50 different data sets per density.
		Each data set from DAG model with 10 nodes.
	}
	\label{fig:n-10}
\end{figure}

\subsubsection*{Fourth experimental setting}
We again use the \(\Delta_s^{n}\) prior and now increase the number of nodes in the data-generating DAGs to \({n=10}\), meaning that the size of the search space increases to \({|\mathbb{U}^{10}|=8\,780\,782\,707}\).
We also decrease the size of the Markov chain to \(100\,000\), making it only roughly \(\frac{1}{100\,000}\) the size of the search space.
The decreased size of the Markov chain allows us to increase the number of data sets per density to 50 without overburdening our computational resources (allowing this experimental setting to run in roughly 15 hours, i.e., about 15 minutes per individual experiment per CPU core).
The results are shown in Figure~\ref{fig:n-10}.

As expected, independence testing performs poorly.
Unlike previous settings, the MAP estimate also performs poorly, indicating that the size of the Markov chain was decreased too much compared to the size of the UEC search space.
Nevertheless, the HPD credible set $\U_{[0.1]}$ consistently outperforms independence testing for densities \(p > 0.6\), and the penalized MLEs perform well as density increases, correctly estimating UEC-representatives at a proportion of around 0.8 for density \(p=0.9\).

\subsubsection{Estimating time complexity}
\label{sec:estim-time-compl}
Finally, we used the \texttt{big\_O} Python package (\url{https://pypi.org/project/big-O/}) to estimate the time complexity of our \texttt{GrUES} implementation, opting for a more practical, application-oriented analysis as opposed to a theoretical analysis based on the algorithms in Sections~\ref{sec:pseud-its-equiv}~and~\ref{sec:gr-algor}.
In particular, we ran experiments on synthetic data independently varying (i) the number of nodes in the data-generating DAG, from 3 up to 50, (ii) the sample size drawn from the generating DAG, from 1\,000 up to 100\,000, and (iii) the length of the Markov chain computed by \texttt{GrUES}, from 5 up to 1\,000\,000.
The results of are shown in Table~\ref{tab:complex}.

\begin{table}[h]
	\centering
	\begin{filecontents*}{complexity_table.csv}
varied parameter,estimated complexity
number of nodes,Cubic: time $\approx 0.01 + 0.0000002n^3$ (sec)
sample size,Linear: time $\approx 0.02 + 0.0000006n$ (sec)
Markov chain length,Linear: time $\approx -0.003 + 0.0015n$ (sec)
\end{filecontents*}
\csvautobooktabular{complexity_table.csv}
	\caption{Empirical estimation of the time complexity of \texttt{GrUES}.}
	\label{tab:complex}
\end{table}

The complexity grows linearly with sample size and Markov chain length but cubically with number of nodes.
This suggests \texttt{GrUES} is reasonably efficient and scalable, considering causal structure learning for MECs is an NP-hard problem \citep{chickering2004large}.

\subsubsection*{Summary of the experiments and conclusions}
The four different experimental settings were aimed at assessing the performance of GrUES according to different possible experimental decisions, including the influence of the choice of prior, the number of variables in the model, and the effect of the length of the Markov chain relative to the size of the state space of possible UECs.

Comparing the first and second settings, we observe that the uniform prior performs comparatively well but slightly worse than the informative prior, suggesting that a well-chosen prior helps to achieve the goal of incorporating useful expert knowledge without giving \texttt{GrUES} an unfair advantage over baseline methods.

In the first and third settings, we see independence testing does worse as the number of nodes increases, but both the posterior and penalized MLEs given by \texttt{GrUES} do not deteriorate as much---indeed, they perform quite well for densities \(p \geq 0.5\).
This suggests that for larger, denser graphs, any of these three \texttt{GrUES} estimators is preferable to independence testing.

In the third and fourth settings, we see the MAP estimate deteriorates as \texttt{length} becomes much smaller than \(|\mathbb{U}^n|\).
This is reasonable, given that GrUES can only explore a small portion of the state space, making it difficult to obtain good estimates of the posterior.
The penalized MLEs still perform quite well at higher densities even when the length of the Markov chain is small.
This indicates that \texttt{GrUES} is still visiting the correct UEC-representative surprisingly often given that it only explores a small fraction of the UEC-representatives.

In summary, \texttt{GrUES} almost always performs better than independence testing, unless the density of the true DAG is around \(p \in \{0.2, 0.3, 0.4\}\).
This failure in performance may either be due to sample size or possibly due to the connectivity of the space of UEC-representatives as discussed in the first experimental setting.
Further studies into how \texttt{GrUES} connects the space of UEC-representatives would be interesting future work.
Although the \(\Delta_s^n\) prior seems to improve performance, \texttt{GrUES} performs reasonably well with a uniform prior.
We studied the performance of \texttt{GrUES} on a prior that encodes beliefs about the number of source nodes in the data-generating DAG.
It would interesting to explore what sort of prior beliefs a practitioner may be willing to assert about a marginal independence model, construct priors that allow us to represent these beliefs and test them with \texttt{GrUES}.

In cases of short Markov chains relative to the size of the state space, or in cases of very sparse graphs, the \(\mathtt{nuclear}\)-penalized MLE should be preferred over the MAP estimate.
However, the MAP estimate appears to be preferable when the Markov chain is reasonably long.
In either regime, the MAP and both penalized-MLEs appears to yield notably better estimates relative to independence testing for denser data-generating DAGs, namely, when $p>0.5$.
This suggests that \texttt{GrUES} performs well in situations where one believes that the data-generating DAG is on the denser side of the spectrum.
This is also the regime where knowing the true DAG model becomes increasingly less useful due to an increased number of parameters and an associated decrease in computational efficiency of probabilistic/causal inference with such graphical models.
In such situations it may in fact be more advantageous to approximate the DAG model with a marginal independence model.
In this regime, \texttt{GrUES} appears to offer a more reliable estimator for the desired marginal independence model than independence testing.

Finally, we note that the complexity analysis of \texttt{GrUES} suggests that it is relatively efficient, scaling linearly in the sample size and length of Markov chain and only cubically in the number of variables.
This suggests that \texttt{GrUES} can be used to the aforementioned ends for relatively high dimensional models where independence testing and/or estimating the complete DAG model may be less feasible.

\section*{Acknowledgements}
\label{sec:acknowledgements}

We thank Federica Milinanni, Felix Rios, Albin Toft and Alice Harting for helpful discussions.
We thank the organizers of the Algebraic Statistics Conference 2022 at University of Hawai`i at M\=anoa, Honolulu.
Danai Deligeorgaki, Alex Markham, and Liam Solus were partially supported by the Wallenberg Autonomous Systems and Software Program (WASP) funded by the Knut and Alice Wallenberg Foundation.
Pratik Misra was partially supported by the Brummer \& Partners MathDataLab.
Liam Solus was partially supported the G\"oran Gustafsson Stiftelse and Starting Grant No.~2019-05195 from The Swedish Research Council (Vetenskapsrådet).

\bibliographystyle{abbrvnat}
\bibliography{aStat}
\end{document}